\documentclass[english,11pt]{article}
\usepackage{amsfonts}
\usepackage{amsmath,amssymb}
\usepackage{amsthm}
\usepackage{natbib}
\usepackage{bm}
\usepackage{graphicx}
\usepackage{subcaption}
\usepackage{appendix}
\usepackage{float}
\usepackage{color}
\usepackage[colorlinks,linkcolor=blue,citecolor=blue,bookmarks=false,pagebackref]{hyperref}
\usepackage[top=1.0in, bottom=1.3in, left=1.1in, right=1.1in]{geometry}
\usepackage{setspace}
\allowdisplaybreaks

\theoremstyle{plain}

\theoremstyle{plain}

\theoremstyle{plain}

\theoremstyle{plain}
\newtheorem{proposition}{Proposition}[section]
\theoremstyle{plain}

\theoremstyle{remark}

\theoremstyle{definition}

\newcommand{\was}{\mathcal{W}}

\newcommand{\argmin}{\operatornamewithlimits{argmin\,}}
\mathchardef\mhyphen="2D

\floatstyle{ruled}
\newfloat{algorithm}{tbp}{loa}
\providecommand{\algorithmname}{Algorithm}
\floatname{algorithm}{\protect\algorithmname}

\title{Schr\"odinger Bridge Samplers}
\author{Espen Bernton\thanks{Department of Statistics, Columbia University, USA.}~,
Jeremy Heng\thanks{ESSEC Business School, Singapore.}~, Arnaud Doucet\thanks{Department of Statistics, University of Oxford, UK.}~~and Pierre E. Jacob\thanks{Department of Statistics, Harvard University, USA.}}
\date{}

\begin{document}

\maketitle

\begin{abstract}
\noindent Consider a reference Markov process with initial distribution $\pi_{0}$ and transition kernels $\{M_{t}\}_{t\in[1:T]}$, for some $T\in\mathbb{N}$. Assume that you are given distribution $\pi_{T}$, which is not equal to the marginal distribution of the reference process at time $T$. In this scenario, Schr\"odinger addressed the problem of identifying the Markov process with initial distribution $\pi_{0}$ and terminal distribution equal to $\pi_{T}$ which is the closest to the reference process in terms of Kullback--Leibler divergence. This special case of the so-called Schr\"odinger bridge problem can be solved using iterative proportional fitting, also known as the Sinkhorn algorithm.
We leverage these ideas to develop novel Monte Carlo schemes, termed Schr\"odinger bridge samplers, to approximate a target distribution $\pi$ on $\mathbb{R}^{d}$ and to estimate its normalizing constant. This is achieved by iteratively modifying the transition kernels of the reference Markov chain to obtain a process whose marginal distribution at time $T$ becomes closer to $\pi_T = \pi$, via regression-based approximations of the corresponding iterative proportional fitting recursion. We report preliminary experiments and make connections with other problems arising in the optimal transport, optimal control and physics literatures.
\end{abstract}
{\bf Keywords:} Annealed importance sampling; iterative proportional fitting; normalizing constant; sequential Monte Carlo samplers; Schr\"odinger bridge; Sinkhorn's algorithm; optimal transport.


\section{Introduction}\label{sec:introduction}
\subsection{Outline and literature review}\label{sec:problem}
Let $\pi$ be a distribution which admits a density, with respect to some dominating measure on a measurable space $(\mathsf{E},\mathcal{E})$, that can only be evaluated pointwise up to a normalizing constant $Z$. We are interested in approximating expectations with respect to $\pi$ as well as the value of $Z$.
State-of-the-art Monte Carlo methods to address this problem include Annealed Importance Sampling \citep[AIS;][]{crooks1998nonequilibrium,neal:2001} and Sequential Monte Carlo \citep[SMC;][]{del2006sequential}. The basis of these methods is to simulate $N$ non-homogeneous Markov chains with initial distribution $\pi_{0}$ and transition kernels $\{M_{t}\}_{t\in[1:T]}$, designed such that the marginal distribution of each Markov chain at time $T$ is approximately equal to $\pi$\footnote{We can also use an interacting particle system instead of independent Markov chains \citep{del2006sequential}.}. However,  this marginal distribution is typically not analytically available, prohibiting its direct application as a proposal distribution within importance sampling. In AIS and SMC, this intractability is circumvented by introducing an appropriate auxiliary target distribution on the path space $\mathsf{E}^{T+1}$ whose marginal at time $T$ coincides with $\pi$ and with respect to which importance weights can be calculated. This allows us to obtain consistent estimates of expectations with respect to $\pi$ and of its normalizing constant $Z$.

These methods have found many applications in physics and statistics, but can perform poorly when the marginal distribution of the samples at time $T$ differs significantly from $\pi$, resulting in importance weights with high variance. Building upon previous contributions for inference in partially observed diffusions and state-space models \citep{richard2007efficient,kappen2016adaptive,guarniero2017iterated}, the controlled SMC sampler methodology of \citet{heng2017controlled} uses ideas from optimal control to iteratively modify the initial distribution and transition kernels of the reference Markov process to reduce the Kullback--Leibler divergence between the induced path distribution and the auxiliary target distribution on $\mathsf{E}^{T+1}$. When applicable, controlled SMC samplers demonstrate clear improvements over AIS and SMC. However, a limitation of this approach is that one must be able to sample from a modified initial distribution.
Practically, this means that $\pi_{0}$ must be conjugate with respect to the policy of the underlying optimal control problem. Additionally, the transition kernels of the reference Markov process must also be conjugate with respect to the chosen policy.

We propose here an alternative approach which is more widely applicable. First, we only modify the transition kernels and not the initial distribution, and relax the requirement that these kernels have to be conjugate with respect to the policy. Second, instead of minimizing the KL divergence on path space $\mathsf{E}^{T+1}$ with respect to a \textit{fixed} auxiliary target distribution, the auxiliary target is itself being optimized across iterations.  We describe our algorithm as an approximation of iterative proportional fitting (IPF), an algorithm introduced in various forms by \citet{deming1940least,sinkhorn1967,ireland1968contingency,kullback1968probability}. In finite state-spaces, this algorithm is also known as Sinkhorn's algorithm and has recently gained much attention in machine learning \citep{cuturi2013sinkhorn,peyre2018computational}. Under weak regularity conditions, IPF is known to converge to the solution of the Schr\"odinger bridge problem in both finite and continuous state-spaces; see, e.g., \citep{sinkhorn1967,ruschendorf1995convergence}. However, whereas in finite state-spaces, the steps of IPF can be computed exactly, these steps are intractable in all but trivial scenarios in continuous state-spaces. Recent computational approaches proposed to approximate the IPF recursion in continuous state-spaces either rely on deterministic \citep{chen2016entropic} or stochastic \citep{reich2018data} discretization of the space using $N$ atoms, and then fall back on the finite state-space IPF algorithm. 
We propose here an approximate IPF scheme which instead relies on regression-based approximations in the spirit of \citet{heng2017controlled}. We demonstrate experimentally its performance on various problems.

The rest of this paper is organized as follows. In the remaining part of Section \ref{sec:introduction}, we define our notation, formalize the problem statement, review SMC samplers and their limitations, and give a brief overview of the proposed method. In Section \ref{sec:sb}, we discuss Schr\"odinger bridges and their various formulations, and introduce numerical algorithms to approximate them. In Section \ref{sec:ssb_samplers}, we discuss our main computational contribution, which we term the sequential Schr\"odinger bridge sampler. In Section \ref{sec:connections}, we discuss connections between the Schr\"odinger bridge problem and various other topics. Section \ref{sec:numerical_experiments} contains numerical experiments, and Section \ref{sec:discussion} concludes.

\subsection{Notation}\label{sec:notation}
Given integers $n\leq m$ and a sequence $\{x_{t}\}_{t\in\mathbb{N}}$, we define the set $[n:m]=\left\{ n,\ldots,m\right\}$  and write the subsequence $x_{n:m}=(x_{n}, x_{n+1}, \ldots,x_{m})$. Let $(\mathsf{E},\mathcal{E})$ be an arbitrary measurable space, $\mathcal{P}(\mathsf{E})$ and $\mathcal{M}(\mathsf{E})$ denote the set of all probability measures and Markov transition kernels on $\mathsf{E}$, respectively. Given $\mu,\nu\in\mathcal{P}(\mathsf{E})$, we write $\mu\ll\nu$ if $\mu$ is absolutely continuous with respect to $\nu$, and denote the corresponding Radon--Nikodym derivative as $\mathrm{d}\mu\mathrm{/d}\nu$. The Kullback--Leibler (KL) divergence from $\nu\in\mathcal{P}(\mathsf{E})$ to $\mu\in\mathcal{P}(\mathsf{E})$ is defined as
$$\mathrm{KL}(\mu|\nu)=\int_{\mathsf{E}}\log\frac{\mathrm{d}\mu}{\mathrm{d}\nu}(x)\mu(\mathrm{d}x)$$
if the integral is finite and $\mu\ll\nu$, and $\mathrm{KL}(\mu|\nu)=\infty$ otherwise. The set of all real-valued, $\mathcal{E}$-measurable and bounded functions on $\mathsf{E}$ is denoted by $\mathcal{B}(\mathsf{E})$. Given $\mu\in\mathcal{P}(\mathsf{E})$, $M\in\mathcal{M}(\mathsf{E})$ and $\varphi\in\mathcal{B}(\mathsf{E})$, we define the integral $\mu(\varphi)=\int_{\mathsf{E}}\varphi(x)\mu(\mathrm{d}x)$ and the function $M(\varphi)(\cdot)=\int_{\mathsf{E}}\varphi(y)M(\cdot,\mathrm{d}y)\in\mathcal{B}(\mathsf{E}).$ For ease of presentation, we will often assume that measures and transition kernels admit densities with respect to a $\sigma$-finite dominating measure $\mathrm{d}x$, in which case we write the densities of $\mu\in\mathcal{P}(\mathsf{E})$ and $M\in\mathcal{M}(\mathsf{E})$ as $\mu(\mathrm{d}x) = \mu(x)\mathrm{d}x$ and $M(x,\mathrm{d}y) = M(x,y)\mathrm{d}y$, respectively.

\subsection{Problem formulation and SMC samplers}\label{sec:problem}
In this article, we will restrict ourselves to $\mathsf{E}:=\mathbb{R}^{d}$, with $\mathcal{E}$ being the corresponding Borel $\sigma$-algebra. We are interested in sampling from a target distribution $\pi$ on $\mathsf{E}$ which admits a density $\pi(x) = Z^{-1}\gamma(x) \in \mathcal{P}(\mathsf{E})$ with respect to the Lebesgue measure $\mathrm{d}x$, assuming that we can evaluate $\gamma(x)$ pointwise.  We are also interested in estimating its normalizing constant $Z = \int_{\mathsf{E}}\gamma(x)\mathrm{d}x$. A standard strategy is to introduce a collection of auxiliary probability measures $\{\pi_t\}_{t\in[0:T]}\subset \mathcal{P}(\mathsf{E})$ that interpolate between an easy-to-sample distribution $\pi_0\in \mathcal{P}(\mathsf{E})$ and the target distribution $\pi_T = \pi$,  for some $T\in\mathbb{N}$. A typical example is the geometric interpolation where the auxiliary distributions admit densities of the following form:
\begin{equation}\label{eq:geometric_path}
\gamma_t(x_t) := \pi_0(x_t)^{1-\lambda_t}\gamma(x_t)^{\lambda_t}, \quad \pi_t(x_t)  :=  \gamma_t(x_t)/Z_t, \quad t\in[0:T],
\end{equation}
where $Z_t = \int_{\mathsf{E}}\gamma_t(x)\mathrm{d}x$ and $\{\lambda_t\}_{t\in[0:T]} \subset [0,1]$ is an increasing sequence satisfying $\lambda_0 = 0$ and $\lambda_T = 1$.

The rationale for introducing the sequence $\{\pi_t\}_{t\in[0:T]}$ is that if neighboring distributions $\pi_{t-1}$ and $\pi_t$ are not too different, it might be possible to construct \textit{forward} Markov transition kernels $\{M_t\}_{t\in[1:T]} \subset \mathcal{M}(\mathsf{E})$ such that samples from $\pi_{t-1}$ are approximately distributed as $\pi_t$ when moved with $M_t$. This sampling strategy results in a non-homogeneous Markov chain with initial distribution $\pi_0$ and transition kernels $\{M_t\}_{t\in[1:T]}$, giving rise to the path measure
\begin{equation}\label{eq:Q}
\mathbb{Q}(\mathrm{d}x_{0:T}) = \pi_{0}(\mathrm{d}x_0)\prod_{t=1}^T M_{t}(x_{t-1}, \mathrm{d}x_{t}).
\end{equation}
The distribution $\mathbb{Q}$ is such that the marginal distribution $q_T$ of $X_T$ is an approximation of the target $\pi_T$. This idea motivates using $q_T$ as a proposal distribution targeting $\pi_T$ in importance sampling, but the corresponding Radon--Nikodym derivative $\mathrm{d}\pi_T\mathrm{/d}q_T$ cannot be computed even up to a normalizing constant, as $q_T$ is typically intractable.

SMC samplers \citep{del2006sequential} avoid this intractability by instead performing importance sampling on path space $\left(\mathsf{E}^{T+1},\mathcal{E}^{T+1}\right)$ by defining the extended target distribution
\begin{equation}\label{eq:extended_target}
\mathbb{P}(\mathrm{d}x_{0:T}) = \pi_{T}(\mathrm{d}x_T)\prod_{t=1}^T L_{t-1}(x_t, \mathrm{d}x_{t-1}),
\end{equation}
where $\{L_t\}_{t\in [0:T-1]}$ is a sequence of arbitrary \textit{backward} Markov transition kernels, selected such that $\mathbb{P}\ll\mathbb{Q}$ and $\mathrm{d}\mathbb{P}\mathrm{/d}\mathbb{Q}$ can be evaluated pointwise up to a normalizing constant. As the distribution $p_T$ of $X_T$ under $\mathbb{P}$ is such that $p_T = \pi_T$, an importance sampling approximation of $\mathbb{P}$ provides directly an approximation of $\pi_T$ and an unbiased Monte Carlo estimate of $Z_T/Z_0$ using samples from $\mathbb{Q}$, thanks to the identity
\begin{equation}\label{eq:norm_const_identity}
\frac{Z_T}{Z_0}=\frac{Z_T}{Z_0}\mathbb{E}_{\mathbb{Q}}\left[\frac{\mathrm{d}\mathbb{P}}{\mathrm{d}\mathbb{Q}}(X_{0:T})\right]=\mathbb{E}_{\mathbb{Q}}\left[\prod_{t=1}^T w_{t}(X_{t-1},X_t)\right],
\end{equation}
where the \textit{incremental} importance weights are given by
\begin{equation}
w_{t}(x_{t-1},x_t):=\frac{\mathrm{d} L_{t-1} \otimes \gamma_t}{\mathrm{d}\gamma_{t-1} \otimes M_{t}}(x_{t-1}, x_t),
\end{equation}
and
\begin{align*}
L_{t-1} \otimes \gamma_t(\mathrm{d}x_{t-1},\mathrm{d}x_{t})&:=\gamma_t(\mathrm{d}x_{t})L_{t-1}(x_{t},\mathrm{d}x_{t-1}), \\
\gamma_{t-1} \otimes M_{t}(\mathrm{d}x_{t-1},\mathrm{d}x_t)&:=\gamma_{t-1}(\mathrm{d}x_{t-1})M_{t}(x_{t-1},\mathrm{d}x_t).
\end{align*}

However, the performance of any importance sampling scheme depends on the Kullback-Leibler discrepancy between the target $\pi_T$ and the proposal distributions \citep{chatterjee2018sample}, which can only increase when extending the domain of integration from $\mathsf{E}$ to $\mathsf{E}^{T+1}$. This can be seen from the decomposition
\begin{align}
\mathrm{KL}(\mathbb{P}|\mathbb{Q}) 
&=  \mathrm{KL}(\pi_T | q_T) + \int_{\mathsf{E}} \mathrm{KL}\left(\mathbb{P}(\mathrm{d}x_{0:T-1} | x_T) | \mathbb{Q}(\mathrm{d}x_{0:T-1} | x_T)\right) \pi_T(\mathrm{d}x_{T}),  \label{eq:decompose_kl}
\end{align}
in which the first term captures the discrepancy between $q_T$ and $\pi_T$, and the second captures the additional discrepancy arising from the introduction of the backward kernels $\{L_t\}_{t\in[0:T-1]}$. In other words, making the importance weights tractable comes at a cost.

The formula in \eqref{eq:decompose_kl} also shows that $\mathrm{KL}(\mathbb{P}|\mathbb{Q})$ is minimized by choosing the backward kernels to make the equality
\begin{equation}\label{eq:optimal_P}
\mathbb{P}(\mathrm{d}x_{0:T}) = \frac{\mathrm{d}\pi_T}{\mathrm{d}q_T}(x_T) \mathbb{Q}(\mathrm{d}x_{0:T})
\end{equation}
hold, as this is equivalent to $\mathbb{P}(\mathrm{d}x_{0:T-1} | x_T) = \mathbb{Q}(\mathrm{d}x_{0:T-1} | x_T)$ for $\pi_T$-almost every $x_T$. The optimal backward kernels $\{L^\text{opt}_t\}_{t\in [0:T-1]}$ must therefore satisfy the forward-backward relation
\begin{equation}\label{forward_backward}
\pi_0(\mathrm{d}x_0) \prod_{t=1}^T M_t(x_{t-1}, \mathrm{d}x_t) = q_T(\mathrm{d}x_T) \prod_{t=1}^T L^\text{opt}_{t-1}(x_{t}, \mathrm{d}x_{t-1}),
\end{equation}
where $q_t(\mathrm{d}x) L^\text{opt}_{t-1}(x,\mathrm{d}x')=q_{t-1}(\mathrm{d}x') M_t(x',\mathrm{d}x)$, as noticed by \citet{del2006sequential}. They also showed that these kernels minimize the variance of the weights $\prod_{t=1}^T w_{t}(X_{t-1},X_t)$ under $X_{0:T} \sim \mathbb{Q}$. Note that this choice remains intractable, and one must typically resort to approximations of the optimal backward kernels in practice.

Both \citet{crooks1998nonequilibrium} and \citet{neal:2001} restrict themselves to scenarios where $M_t$ is selected to be $\pi_t$-invariant and $L_{t-1}$ is the reversal of $M_t$\footnote{The kernel $L_{t-1}$ is the reversal of the $\pi_t$-invariant kernel $M_t$ if $\pi_t(\mathrm{d}x)M_t(x,\mathrm{d}x')=\pi_t(\mathrm{d}x')L_{t-1}(x',\mathrm{d}x)$. In particular, $L_{t-1}=M_t$ if $M_t$ is $\pi_t$-reversible}. Under these conditions, \eqref{eq:norm_const_identity} corresponds to a discrete-time version of the celebrated Jarzynski's identity \citep{jarzynski1997nonequilibrium}. The generalized framework of \citet{del2006sequential} presented here is key to the developments in this article, as we will use forward transition kernels that are only approximately invariant with respect to the measures $\pi_t$.

Imagine for a moment that we flip the roles of $\pi_0$ and $\pi_T$ and that we could sample from the path measure $\mathbb{P}^{(1)} := \mathbb{P}$ defined in \eqref{eq:optimal_P} to target $\pi_0$. Following the reasoning above, but applied backwards in time, the corresponding optimal forward kernels would give rise to a path measure $\mathbb{Q}^{(1)}$ such that $\mathbb{Q}^{(1)}(\mathrm{d}x_{1:T} | x_0) = \mathbb{P}^{(1)}(\mathrm{d}x_{1:T} | x_0)$. In turn, the path measure $\mathbb{Q}^{(1)}$ could hypothetically be used as a  path space proposal targeting $\pi_T$. Furthermore, this process could be iterated to construct forward and backward kernels $\{M^{(i)}_t\}_{t\in[1:T]}$ and $\{L^{(i)}_t\}_{t\in[0:T-1]}$ and corresponding path measures $\mathbb{Q}^{(i)}$ and $\mathbb{P}^{(i)}$. In what follows, we will describe this algorithm as \textit{iterative proportional fitting} (IPF). Under certain regularity conditions, the measures $\mathbb{Q}^{(i)}$ and $\mathbb{P}^{(i)}$ both converge to the solution of the so-called Schr\"odinger bridge problem. By numerically approximating the steps of IPF, we will leverage the reduction in discrepancy between $\mathbb{Q}^{(i)}$ and $\mathbb{P}^{(i)}$ to yield efficient SMC sampling from the target distribution $\pi_T$.

\subsection{Outline of the proposed method}
The formulas in equations \eqref{eq:decompose_kl} and \eqref{eq:optimal_P} are at the core of our approach. In Section \ref{sec:sb}, we will describe \eqref{eq:optimal_P} as providing the solution to the \textit{forward} Schr\"odinger half-bridge problem
\begin{equation}\label{eq:shb_forward}
\mathbb{P}^{(i)}(\mathrm{d}x_{0:T}) = \argmin_{\mathbb{H} \in \mathcal{P}_T(\pi_T)} \mathrm{KL}(\mathbb{H} | \mathbb{Q}^{(i-1)}),
\end{equation}
where for any $t \in [0:T]$, $\mathcal{P}_t(\pi_t)$ denotes the set of path measures on $\left(\mathsf{E}^{T+1},\mathcal{E}^{T+1}\right)$ that have $\pi_t$ as their $t$-th marginal distribution. Approximating its solution allows us to implicitly estimate the optimal backward kernels for a given set of forward kernels.
By alternating this step with solving the analogous \textit{backward} Schr\"odinger half-bridge problem
\begin{equation}\label{eq:shb_backward}
\mathbb{Q}^{(i)}(\mathrm{d}x_{0:T}) = \argmin_{\mathbb{H} \in \mathcal{P}_0(\pi_0)} \mathrm{KL}(\mathbb{H} | \mathbb{P}^{(i)}),
\end{equation}
we will also make refinements on the forward kernels.

Iterating between the forward and backward half-bridge problems can be seen as an instance of IPF \citep{deming1940least,ireland1968contingency,kullback1968probability}. Under the weak regularity conditions detailed by  \citet{ruschendorf1995convergence}, the iterates $\mathbb{P}^{(i)}$ and $\mathbb{Q}^{(i)},$ initialized at $\mathbb{Q}^{(0)} = \mathbb{Q}$, are known to converge as $i\to \infty$ to the solution of the Schr\"odinger bridge problem
\begin{equation} \label{eq:sb}
\mathbb{S}(\mathrm{d}x_{0:T}) = \argmin_{\mathbb{H} \in \mathcal{P}_{0,T}{(\pi_{0,T})}} \mathrm{KL}(\mathbb{H} | \mathbb{Q}),
\end{equation}
where $\mathcal{P}_{0,T}{(\pi_{0,T})} =  \mathcal{P}_0{(\pi_0)}\cap \mathcal{P}_T{(\pi_T)}$. The IPF algorithm is illustrated in Figure \ref{fig:IPF_illustration}, and a comparison between a reference process and the associated Schr\"odinger bridge is provided in Figure \ref{fig:langevin_ridgeplot}.

\begin{figure}[t]
\centering
\includegraphics[width=0.6\textwidth]{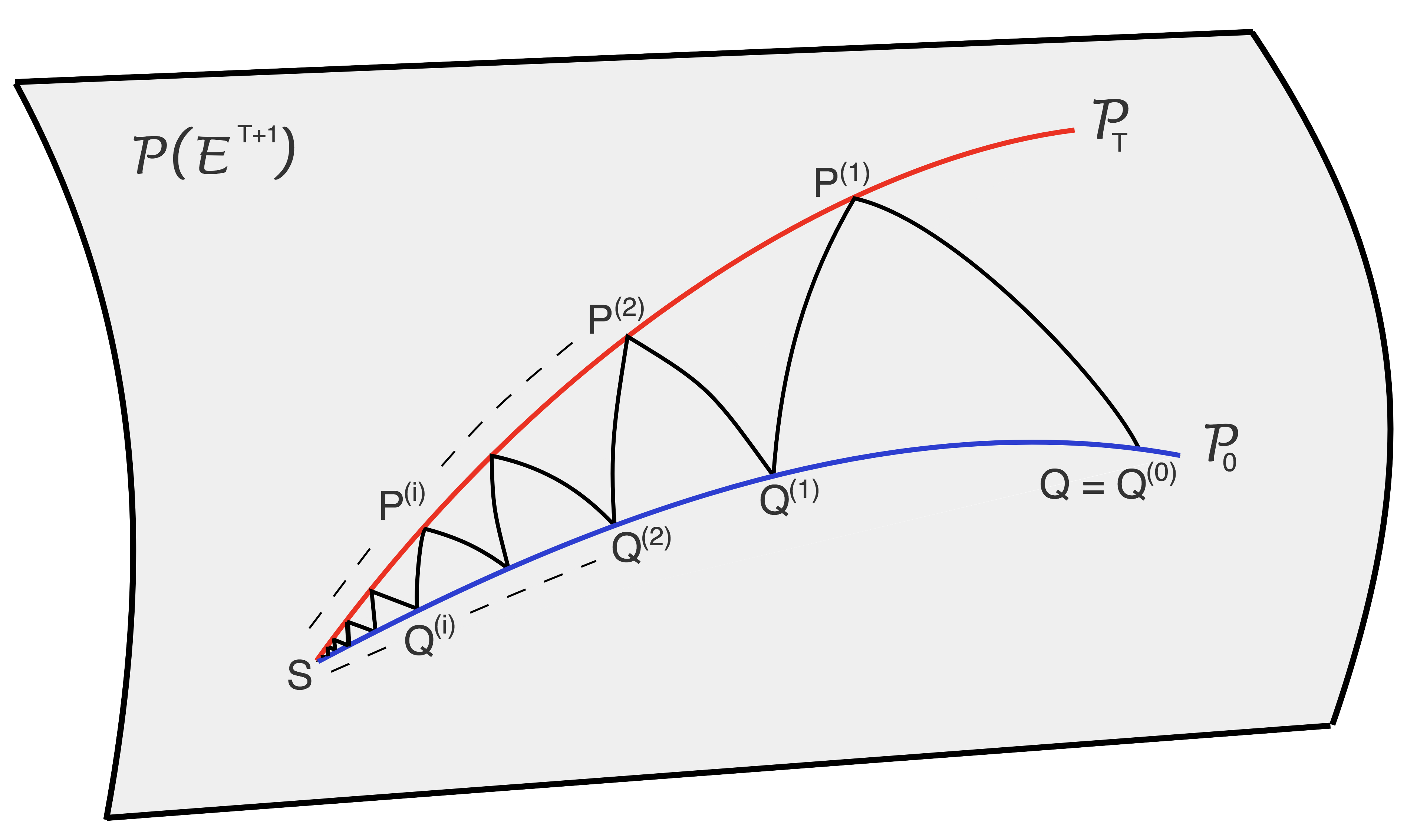}
        \caption{\small Illustration of the iterative proportional fitting procedure. The blue line represents $\mathcal{P}_0(\pi_0) \subset \mathcal{P}(\mathsf{E}^{T+1})$, denoted $\mathcal{P}_0$ in the figure, while the red line represents $\mathcal{P}_T(\pi_T) \subset \mathcal{P}(\mathsf{E}^{T+1})$, denoted $\mathcal{P}_T$. The black line illustrates that the alternating KL projections $\mathbb{Q}^{(i)} \in \mathcal{P}_0(\pi_0)$ and $\mathbb{P}^{(i)} \in \mathcal{P}_T(\pi_T)$ converge towards the Schr\"odinger bridge $\mathbb{S}$.}
        \label{fig:IPF_illustration}
\end{figure}

\begin{figure}[h]
\centering
        \begin{subfigure}[t]{\textwidth}
            \centering
            \includegraphics[width=0.85\textwidth]{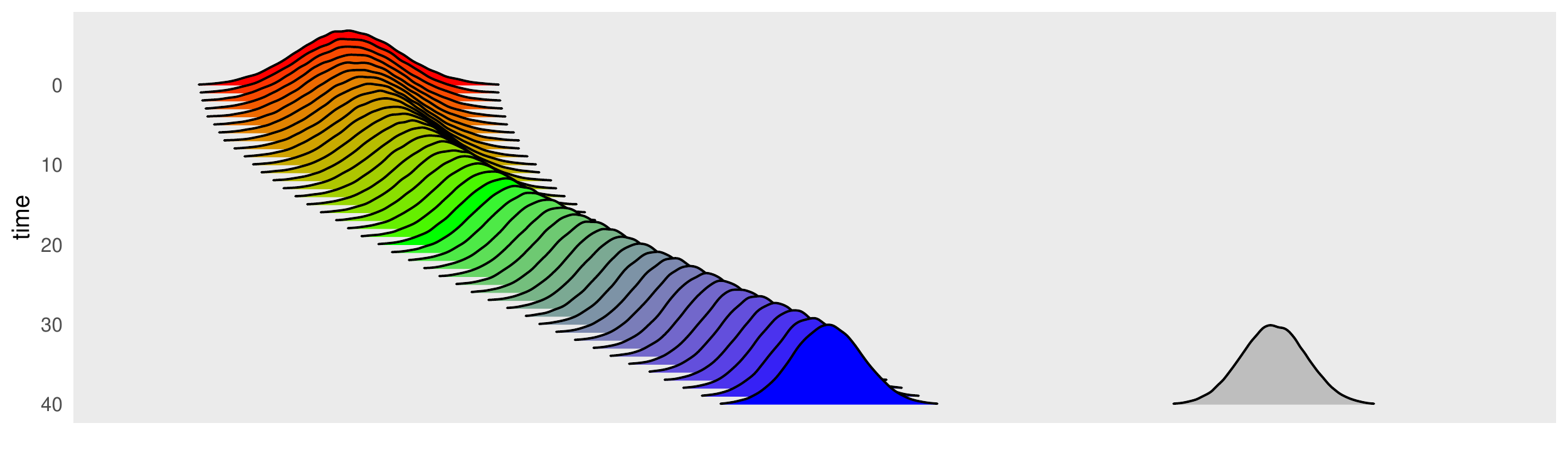}
            \caption{{ Marginals $q_t$ of a reference process $\mathbb{Q}(\mathrm{d}x_{0:T})$ with $T = 40$ and the target $\pi$.}}
            \label{fig:langevin_init_ridgeplot}
        \end{subfigure}

        \vspace*{0.7cm}

        \begin{subfigure}[t]{\textwidth}
            \centering
            \includegraphics[width=0.85\textwidth]{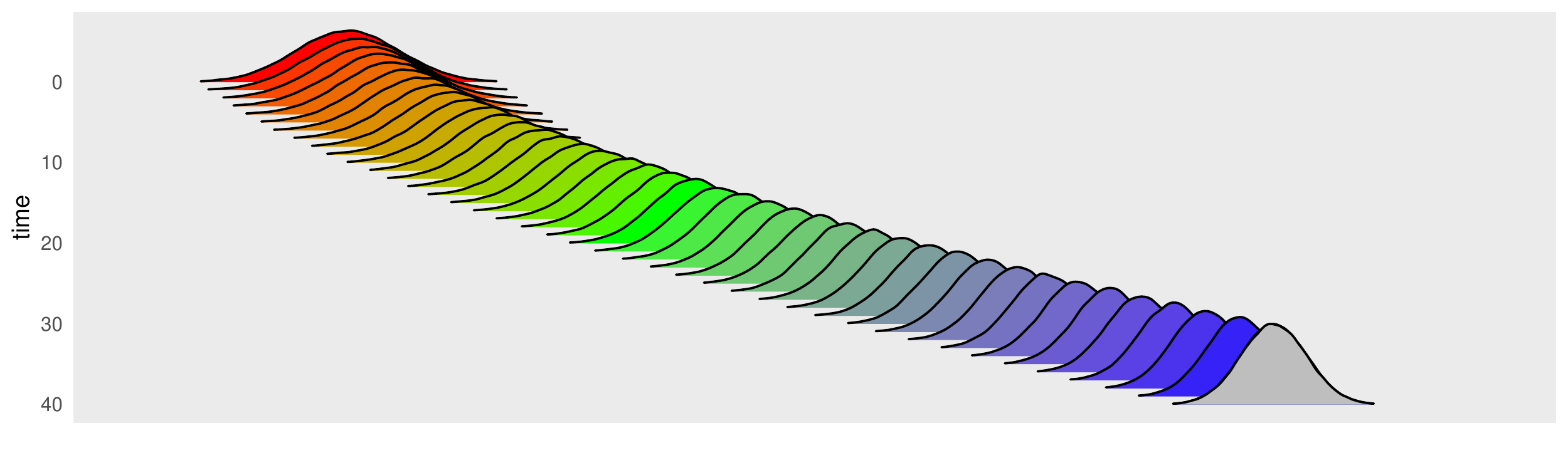}
            \caption{{ Marginals $s_t$ of the Schr\"odinger bridge $\mathbb{S}(\mathrm{d}x_{0:T})$  corresponding to the reference process  $\mathbb{Q}(\mathrm{d}x_{0:T})$.}}
            \label{fig:langevin_sb_ridgeplot}
        \end{subfigure}
        \caption{\small Illustration of the evolutions of the marginal distributions $q_t$ and $s_t$ of the reference process $\mathbb{Q}(\mathrm{d}x_{0:T})$ and the associated Schr\"odinger bridge $\mathbb{S}(\mathrm{d}x_{0:T})$, respectively. The colors transition from red to green to blue as $t$ increases from $0$ to $T = 40$. The Schr\"odinger bridge corresponds to the minimal modification of the reference process in terms of KL that interpolates between the initial distribution $\pi_0$ and the target distribution $\pi_T = \pi$, here given in grey.}
        \label{fig:langevin_ridgeplot}
\end{figure}

The Schr\"odinger bridge problem has a long history in physics, and was first studied by its namesake \citet{schrodinger1931umkehrung,schrodinger1932theorie} because of its connections to quantum mechanics. It was later rediscovered and posed in its modern formulation \eqref{eq:sb} in probability \citep[e.g.][]{dawson1987large, follmer1988random} and control theory \citep[e.g.]{mikami1990variational, dai1991stochastic}; see the recent survey of \citet{leonard2014survey} for a thorough overview of these links. Of particular importance in our method will be the latter of these formulations, namely that the full- and half-bridges are solutions to finite-horizon minimum KL control problems \citep[see e.g.][]{todorov2009efficient}. This perspective will lend computational approaches that allow us to approximate the iterations \eqref{eq:shb_forward} and \eqref{eq:shb_backward} in continuous spaces, where they typically cannot be implemented exactly.

Specifically, we utilize the class of path measures that arise as $\psi$-twisted versions of $\mathbb{Q}$, i.e.~measures of the form
\begin{equation}
\mathbb{Q}^\psi(\mathrm{d}x_{0:T}) = \pi_{0}^\psi(\mathrm{d}x_0)\prod_{t=1}^T M^\psi_{t}(x_{t-1}, \mathrm{d}x_{t}),
\end{equation}
where $\psi = \{\psi_t\}_{t\in[0:T]}$ is a collection of strictly positive functions on $\mathsf{E}$ referred to as a \emph{policy}, and
\begin{equation}\label{eq:twisted_kernel}
\pi_0^\psi(\mathrm{d}x_0) = \frac{\pi_0(\mathrm{d}x_0)\psi_0(x_0)}{\pi_0(\psi_0)}, \quad  M^\psi_{t}(x_{t-1}, \mathrm{d}x_t) = \frac{M_t(x_{t-1}, \mathrm{d}x_t)\psi_{t}(x_t)}{M_t(\psi_t)(x_{t-1})}, \quad t\in[1:T].
\end{equation}
In Section \ref{sec:sb}, we show that the Schr\"odinger bridge itself belongs to this class, i.e.~that $\mathbb{S} = \mathbb{Q}^{\psi^\star}$ for some policy $\psi^\star$. Similarly, the half-bridge problems \eqref{eq:shb_forward} and \eqref{eq:shb_backward} can be expressed as iterative refinements of a policy approximating $\psi^\star$, which can be estimated using parametric policy classes and approximate dynamic programming in the spirit of \citet{heng2017controlled}.

In practice, this approach requires to be able to sample from path measures of the form $\mathbb{Q}^\psi$, which places restrictions on the initial Markov process $\mathbb{Q}$ and the policy $\psi$. In our numerical examples, the reference process is obtained by discretizing the continuous-time (overdamped) Langevin dynamics
\begin{equation}\label{eq:langevin}
\mathrm{d}X_s = \frac{1}{2}\nabla \log \pi_s(X_s) \mathrm{d}s + \mathrm{d}W_s, \quad s\in[0,\tau],
\end{equation}
where $W_s$ denotes the standard Brownian motion, and $(\pi_s)_{s\in[0,\tau]}$ is a smooth curve of distributions such that $\pi_\tau = \pi$, e.g.~a continuous version of the geometric interpolation in \eqref{eq:geometric_path}. There are two main motivations for this choice. Firstly, if $\tau$ is large enough, the distribution of $X_s$ in \eqref{eq:langevin} is close to $\pi_s$ for any $s$; see e.g. \citet{chiang1987diffusion}. Provided the discretization of $[0,\tau]$ is fine enough, we expect the marginal distribution of the discretized process to be close to $\pi_s$ also. In contrast to using a Brownian motion as the reference process like in \citet{mikami2004monge,tzen2019theoretical}, such a reference process provides a terminal marginal not too different from the target $\pi_T$ and the corresponding Schr\"odinger bridge is typically easier to approximate numerically. Secondly, this perspective will allow us to more readily estimate the optimal backward kernels, and to make use of flexible function estimation methods to approximate the optimal policy.

We can exploit further the property that it is easier to approximate the Schr\"odinger bridge when the terminal marginal of the reference process is not too different from the target $\pi_T$ by composing the solutions of a sequence of intermediate Schr\"odinger bridge problems, for instance between adjacent distributions in the geometric interpolation \eqref{eq:geometric_path}. This underpins the sequential Schr\"odinger bridge (SSB) sampler proposed in Section \ref{sec:ssb}. In Figure \ref{fig:sb_ssb_comparison} we illustrate a continuous-time process and its corresponding discretization, and compare these reference processes to their associated (sequential) Schr\"odinger bridges.

\begin{figure}[t]
\centering
\includegraphics[width=0.7\textwidth]{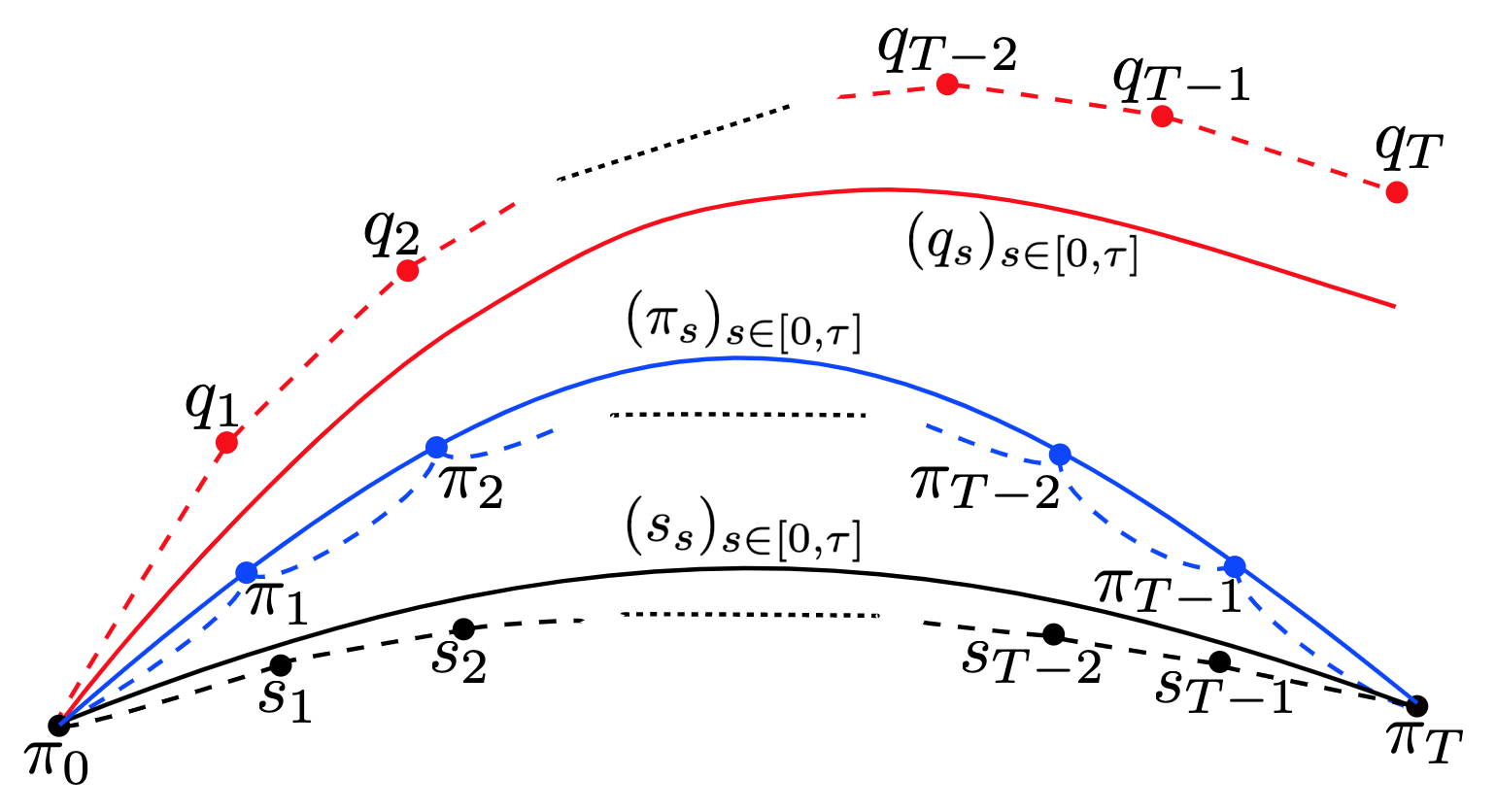}
        \caption{\small Illustration of the marginals of a continuous-time reference process $(q_s)_{s\in[0,\tau]}$ initialized at $\pi_0$ and its discretization $\{q_t\}_{t\in[0:T]}$, shown as solid and dashed red lines respectively. Due to the discretization error, the two paths do not overlap exactly. The solid and dashed black lines correspond to the marginal distributions of the continuous-time and discrete-time Schr\"odinger bridges $(s_s)_{s\in[0,\tau]}$ and $\{s_t\}_{t\in[0:T]}$ associated with the aforementioned reference processes. Respectively, they form continuous and discrete interpolations between the marginal distributions $\pi_0$ and $\pi_T = \pi_\tau$. The blue dashed line corresponds to the discrete-time sequential Schr\"odinger bridge, based on solving Schr\"odinger bridge problems between adjacent distributions in the discrete interpolation $\{\pi_t\}_{t\in[0:T]}$. The blue solid line corresponds to the smooth interpolation $(\pi_s)_{s\in[0,\tau]}$, which is also equal to the marginals of the continuous-time version of the sequential Schr\"odinger bridge in Section \ref{sec:ssb_samplers}.}
        \label{fig:sb_ssb_comparison}
\end{figure}

\section{Schr\"odinger bridges}\label{sec:sb}
We review different formulations of the Schr\"odinger bridge problem introduced in \eqref{eq:sb} and IPF to find its solution. We then reformulate IPF as an algorithm that acts in the space of policies. Parameterizing this space allows us to develop a numerical algorithm to approximate the Schr\"odinger bridge based on approximate dynamic programming. In the numerical experiments we work with kernels that arise from discretizing continuous-time processes, and thus also review some properties of the continuous-time formulation of the Schr\"odinger bridge problem.

\subsection{Dynamic and static formulations}
Recall that for a reference path measure $\mathbb{Q}$ and initial and terminal constraints $\pi_0$ and $\pi_T$ respectively, the \textit{dynamic} Schr\"odinger bridge problem takes the form \eqref{eq:sb}.
It can equivalently be expressed in its \textit{static} form as follows: notice that for any $\mathbb{H} \in \mathcal{P}(\mathsf{E}^{T+1})$ we have the decomposition
\begin{equation}\label{eq:decompose_kl_dynamic}
\small{\mathrm{KL}(\mathbb{H}|\mathbb{Q}) =  \mathrm{KL}(h_{0,T} | q_{0,T}) + \int_{\mathsf{E}^2} \mathrm{KL}\left(\mathbb{H}(\mathrm{d}x_{1:T-1} | x_0, x_T) | \mathbb{Q}(\mathrm{d}x_{1:T-1} | x_0, x_T)\right) h_{0,T}(\mathrm{d}x_{0},\mathrm{d}x_{T}),}
\end{equation}
where for any $0\leq s < t \leq T$, $h_{s,t}$ and $q_{s,t}$ denote the two-time marginals of $\mathbb{H}$ and $\mathbb{Q}$. Since the constraint in \eqref{eq:sb} applies only to the first term in the sum in \eqref{eq:decompose_kl_dynamic}, the second term can be minimized by setting $\mathbb{H}(\mathrm{d}x_{1:T-1} | x_0, x_T)  =  \mathbb{Q}(\mathrm{d}x_{1:T-1} | x_0, x_T)$ for $h_{0,T}$-almost every $(x_0,x_T)$. Hence, \eqref{eq:sb} can be reduced to the static problem
\begin{equation} \label{eq:static_sb}
s_{0,T}(\mathrm{d}x_{0},\mathrm{d}x_{T}) = \argmin_{h_{0,T} \in \mathcal{C}(\pi_0,\pi_T)} \mathrm{KL}(h_{0,T} | q_{0,T}),
\end{equation}
where $\mathcal{C}(\pi_0,\pi_T)\subset \mathcal{P}(\mathsf{E}^2)$ is the set of couplings on $\mathsf{E}^2$ with marginal distributions $\pi_0$ and $\pi_T = \pi$ respectively.

If $q_{0,T}\ll  q_0 \otimes q_T$ and there exists $h_{0,T} \in \mathcal{C}(\pi_0,\pi_T)$ such that $\mathrm{KL}(h_{0,T} | q_{0,T})< \infty$, \citet{ruschendorf1993note} showed that the minimum in \eqref{eq:static_sb} exists, is unique, and of the form
\begin{equation}\label{eq:static_sb_solution}
s_{0,T}(\mathrm{d}x_{0},\mathrm{d}x_{T}) = \varphi^\circ(x_0) q_{0,T}(\mathrm{d}x_{0},\mathrm{d}x_{T}) \varphi^\star(x_T),
\end{equation}
for two functions $\varphi^\circ, \varphi^\star: \mathsf{E} \to \mathbb{R}_+$. These functions are often called potentials, and are themselves unique up to a multiplicative constant. From \eqref{eq:decompose_kl_dynamic} and \eqref{eq:static_sb_solution}, we obtain
\begin{equation}
\mathbb{S}(\mathrm{d}x_{0:T}) = s_{0,T}(\mathrm{d}x_{0},\mathrm{d}x_{T})\mathbb{Q}(\mathrm{d}x_{1:T-1} | x_0, x_T) = \varphi^\circ(x_0) \mathbb{Q}(\mathrm{d}x_{0:T}) \varphi^\star(x_T).
\end{equation}
The marginal constraints on $\mathbb{S}(\mathrm{d}x_{0:T}) $ give rise to the so-called Schr\"odinger equations
\begin{align}
\pi_0(\mathrm{d}x_0) &= \varphi^\circ(x_0) q_0(\mathrm{d}x_0) \int_{\mathsf{E}} \mathbb{Q}(\mathrm{d}x_T|x_0)\varphi^\star(x_T), \label{eq:schrodinger_eq1}\\
\pi_T(\mathrm{d}x_T) &= \varphi^\star(x_T) q_T(\mathrm{d}x_T)\int_{\mathsf{E}} \mathbb{Q}(\mathrm{d}x_0|x_T)\varphi^\circ(x_0), \label{eq:schrodinger_eq2}
\end{align}
where we recall that, from $(\ref{eq:Q})$, $\pi_0=q_0$ in our context. Finding the solution to the Schr\"odinger bridge problem can therefore be reduced to finding the potentials $\varphi^\circ$ and $\varphi^\star$ that solve the Schr\"odinger equations \eqref{eq:schrodinger_eq1} and \eqref{eq:schrodinger_eq2}.

To succinctly express the transition probabilities and marginal distributions under $\mathbb{S}$, we define the harmonic functions
\begin{align}
\psi_T^\star(x_T) &= \varphi^\star(x_T), \label{eq:harmonic_T} \\
\psi_t^\star(x_t) &= \int_\mathsf{E} \psi_T^\star(x_T) \mathbb{Q}(\mathrm{d}x_T | x_t) , \quad t\in[0:T-1], \label{eq:harmonic_t}
\end{align}
and the co-harmonic functions
\begin{align}
\psi_0^\circ(x_0) &= \varphi^\circ(x_0), \label{eq:coharmonic_0}\\
\psi_t^\circ(x_t) &= \int_\mathsf{E} \psi_0^\circ(x_0) \mathbb{Q}(\mathrm{d}x_0 | x_t), \quad t\in[1:T]. \label{eq:coharmonic_t}
\end{align}
Using this notation, the one- and two-time marginals of $\mathbb{S}$ can be expressed as
$$s_{t}(\mathrm{d}x_{t}) = \psi_t^\circ(x_t) \psi_t^\star(x_t) q_{t}(\mathrm{d}x_{t})\quad \text{and} \quad s_{s,t}(\mathrm{d}x_{s},\mathrm{d}x_{t}) = \psi_s^\circ(x_s) \psi_t^\star(x_t) q_{s,t}(\mathrm{d}x_s,\mathrm{d}x_t)$$
for any $t\in[0:T]$ and $0\leq s < t \leq T$, respectively. Thus, the transition probabilities under $\mathbb{S}$ take the form
\begin{equation}
\mathbb{S}(\mathrm{d}x_t | x_s) = \frac{\mathbb{Q}(\mathrm{d}x_t | x_s)\psi_t^\star(x_t)}{\psi_s^\star(x_s)} , \quad 0\leq s < t \leq T.
\end{equation}
From this representation we can see that if $\mathbb{Q}$ is Markovian, then so is $\mathbb{S}$. In this case, we also have that $\psi_s^\star(x_s) = \int_{\mathsf{E}}\psi_t^\star(x_t)\mathbb{Q}(\mathrm{d}x_t | x_s)$. In particular, if $\mathbb{Q}$ is of the form \eqref{eq:Q}, then $\mathbb{Q}(\mathrm{d}x_t | x_{t-1}) = M_t(x_{t-1},\mathrm{d}x_t)$ and $\psi_{t-1}^\star(x_{t-1}) = M_t(\psi^\star_t)(x_{t-1})$. In other words, $\mathbb{S}$ is the $\psi^\star$-twisted version of $\mathbb{Q}$, where the policy $\psi^\star = \{\psi^\star_t\}_{t\in[0:T]}$ is the collection of harmonic functions.

In fact, the Schr\"odinger bridge is a special case of a \textit{reciprocal process}:  $\{Y_t\}_{t\in[0:T]}$ is said to be reciprocal if for any $0\leq t_1 < t_2\leq T$, $\{Y_t\}_{t\in[t_1+1:t_2-1]}$ is conditionally independent of $\{Y_t\}_{t\in[0:t_1] \cup [t_2:T]}$ given $(Y_{t_1}, Y_{t_2})$. \citet{jamison1974reciprocal} showed that a \textit{class} of reciprocal processes can be obtained from a Markov reference process $\{X_t\}_{t\in[0:T]}\sim\mathbb{Q}$ by pinning the endpoints $X_0 = x_0$ and $X_T = x_T$, and assigning the pair $(x_0,x_T)$ different distributions \citep[see also][for further discussions of reciprocal processes in the discrete-time setting]{beghi1996relative}. As a consequence,
if two different Markov reference processes $\mathbb{Q}$ and $\bar{\mathbb{Q}}$ induce the same dynamics between the pinned endpoints $x_0$ and $x_T$, they also define the same class of reciprocal processes. \citet{jamison1974reciprocal} also showed that $\mathbb{S}$ is the unique reciprocal process in this class that is simultaneously Markov and satisfies the marginal constraints $s_0 = \pi_0$ and $s_T = \pi_T$. Thus, since the reference processes $\mathbb{Q}$ and $\bar{\mathbb{Q}}$ define the same reciprocal class, they also define the same Schr\"odinger bridge between $\pi_0$ and $\pi_T$. This is for instance the case for reference processes of the form $\bar{\mathbb{Q}}(\mathrm{d}x_{0:T}) = \bar{\varphi}^\circ(x_0) \mathbb{Q}(\mathrm{d}x_{0:T}) \bar{\varphi}^\star(x_T)$, where $\bar{\varphi}^\circ$ and $\bar{\varphi}^\star$ are approximations of the Schr\"odinger potentials. As we will see in the next section, this representation holds for the IPF iterates and plays an important role in establishing their convergence to the Schr\"odinger bridge.

\subsection{Iterative proportional fitting}

In our notation, each iteration of IPF initialized at  $\mathbb{Q}^{(0)} = \mathbb{Q}$ can be expressed as solving two Schr\"odinger half-bridge problems
\begin{align}\label{eq:IPF}
\mathbb{P}^{(i)}(\mathrm{d}x_{0:T}) &= \argmin_{\mathbb{H} \in \mathcal{P}_T(\pi_T)} \mathrm{KL}(\mathbb{H} | \mathbb{Q}^{(i-1)}),\\
\mathbb{Q}^{(i)}(\mathrm{d}x_{0:T}) &= \argmin_{\mathbb{H} \in \mathcal{P}_0(\pi_0)} \mathrm{KL}(\mathbb{H} | \mathbb{P}^{(i)}),
\end{align}
for $i \geq 1$; in each KL projection, only one of the two marginal constraints are enforced.
Let $\mathbb{S}^{(2i+1)} = \mathbb{P}^{(i+1)}$ and $\mathbb{S}^{(2i)} = \mathbb{Q}^{(i)}$ for any $i \geq 0$. Then, if the minimizer in \eqref{eq:static_sb} exists, \citet[][Proposition 2.1]{ruschendorf1995convergence} shows that the initial and terminal marginals of $\mathbb{S}^{(i)}$, denoted $s_0^{(i)}$ and $s_T^{(i)}$ respectively,  converge to $\pi_0$ and $\pi_T = \pi$ in both KL divergence and Total Variation (TV) as $i\to \infty$. In the following proposition, which is inspired by \citet[][Theorem 2]{altschuler2017near}, we provide further insight into the convergence properties of IPF. The proof can be found in Appendix \ref{appendix:proofs}.
\begin{proposition}\label{prop:IPF_convergence}
For any $\varepsilon > 0$, IPF returns a distribution $\mathbb{S}^{(i)}$ satisfying $\mathrm{KL}(\pi_0 | s_0^{(i)}) + \mathrm{KL}(\pi_T | s_T^{(i)}) < \varepsilon$ in fewer than $\left\lceil\mathrm{KL}(\mathbb{S} | \mathbb{Q})/\varepsilon\right\rceil$ iterations.
\end{proposition}
\noindent This result illustrates the benefit of starting with an initial distribution $\mathbb{Q}$ that is close to $\mathbb{S}$. Under the additional assumption that $\mathbb{Q}(\mathrm{d}x_T | x_0) \geq c \pi_T(\mathrm{d}x_T)$ for $q_0$-almost every $x_0$ and some $c>0$, the sequence $\mathbb{S}^{(i)}$  converges to $\mathbb{S}$ in both KL and TV \citep[][Theorem 3.5]{ruschendorf1995convergence}.

Unlike the problem in \eqref{eq:sb}, the half-bridge problems have explicit solutions. Using the KL decomposition in \eqref{eq:decompose_kl}, we know that
\begin{equation}\label{eq:update_P}
\mathbb{P}^{(i)}(\mathrm{d}x_{0:T}) = \frac{\mathrm{d}\pi_T}{\mathrm{d}q^{(i-1)}_T}(x_T) \mathbb{Q}^{(i-1)}(\mathrm{d}x_{0:T}).
\end{equation}
By analogous reasoning, we also have that
\begin{equation}\label{eq:update_Q}
\mathbb{Q}^{(i)}(\mathrm{d}x_{0:T}) = \frac{\mathrm{d}\pi_0}{\mathrm{d}p^{(i)}_0}(x_0) \mathbb{P}^{(i)}(\mathrm{d}x_{0:T}).
\end{equation}

Next, we show that these iterations can be expressed as policy refinements. To set up an inductive argument, assume that $\mathbb{Q}^{(i-1)} = \mathbb{Q}^{\psi^{(i-1)}}$ for some policy $\psi^{(i-1)} = \{\psi_t^{(i-1)}\}_{t\in[0:T]}$ and $i\geq 1$, i.e. $\mathbb{Q}^{(i-1)}$ is the $\psi^{(i-1)}$-twisted version of $\mathbb{Q}$ in \eqref{eq:Q}. This is trivially true for $i=1$, with  $\psi_t^{(0)} =1$ for all $t\in[0:T]$. By the representation  \eqref{eq:update_P} and \citet[][Proposition 1]{heng2017controlled}, we have that $\mathbb{P}^{(i)} = \mathbb{Q}^{\psi^{(i-1)}\cdot \phi^{(i)}}$, where we have used the notation $\psi\cdot\phi = \{\psi_t \cdot \phi_t\}_{t\in[0:T]}$ in which $\psi_t \cdot \phi_t$ denotes pointwise multiplication, and $\phi^{(i)}$ satisfies the backward recursion (in time):
\begin{align}
\phi_T^{(i)}(x_T) &= \frac{\mathrm{d}\pi_T}{\mathrm{d}q^{\psi^{(i-1)}}_T}(x_T),\label{eq:policy_recursion_end}\\
\phi_t^{(i)}(x_t) &= M_{t+1}^{\psi^{(i-1)}}(\phi_{t+1}^{(i)})(x_t),
\quad t\in[0:T-1].\label{eq:policy_recursion}
\end{align}
In other words, we have the representation
\begin{equation}
\mathbb{P}^{(i)}(\mathrm{d}x_{0:T}) = \pi_{0}^{\psi^{(i-1)}\cdot \phi^{(i)}}(\mathrm{d}x_0)\prod_{t=1}^T M^{\psi^{(i-1)}\cdot \phi^{(i)}}_{t}(x_{t-1}, \mathrm{d}x_{t}).
\end{equation}
With $\mathbb{P}^{(i)}$ expressed this way, the solution of the second half-bridge problem \eqref{eq:update_Q} does not require any additional calculation as it is given by
\begin{equation}
\mathbb{Q}^{(i)}(\mathrm{d}x_{0:T}) = \pi_0(\mathrm{d}x_0) \prod_{t=1}^T M^{\psi^{(i-1)}\cdot \phi^{(i)}}_{t}(x_{t-1}, \mathrm{d}x_{t}).
\end{equation}
Hence, $\mathbb{Q}^{(i)} = \mathbb{Q}^{\psi^{(i)}}$, where $\psi^{(i)}_0 = 1$ and $\psi_t^{(i)} = \psi_t^{(i-1)}\cdot\phi_t^{(i)}$ for $t \in[1:T]$.
Since $\psi^{(i)}_0 = 1$ for all $i\geq 0$, we only have to store $\psi_t^{(i)}, t\in[1:T]$ to define $\mathbb{Q}^{(i)}$ at each iteration (see Section \ref{sec:AIPF}).

Using the policy refinement perspective, the convergence of IPF can be stated informally as follows. As $i\to \infty$, the sequence of policies $\psi^{(i)}$ converges to the policy $\psi^\star$ defined by the harmonic functions \eqref{eq:harmonic_T} and \eqref{eq:harmonic_t}. Alternatively, the convergence can also be expressed in terms of the convergence of a fixed point method to solve \eqref{eq:schrodinger_eq1} and \eqref{eq:schrodinger_eq2} and obtain the corresponding Schr\"odinger potentials. For $i\geq 1$, define
\begin{align}
\alpha^{(i)}(x_0) &=\alpha^{(i-1)}(x_0)\frac{\mathrm{d}\pi_0}{\mathrm{d}p_0^{(i)}}(x_0) = \alpha^{(i-1)}(x_0)
\frac{\pi_0(\phi_0^{(i)})}{\phi_0^{(i)}(x_0)},\label{eq:IPF_potential1}\\  
\beta^{(i)}(x_T) & =\beta^{(i-1)}(x_T) \frac{\mathrm{d}\pi_T}{\mathrm{d}q_T^{(i-1)}}(x_T) = \beta^{(i-1)}(x_T) \phi_T^{(i)}(x_T),\label{eq:IPF_potential2} 
\end{align}
with initialization at $\alpha^{(0)}=1$ and $\beta^{(0)} = \mathrm{d}\pi_T/\mathrm{d}q_T$. These potentials are related to $\mathbb{Q}^{(i+1)}$ and $\mathbb{P}^{(i+1)}$ via $\mathbb{Q}^{(i+1)}(\mathrm{d}x_{0:T}) = \alpha^{(i+1)}(x_0)\mathbb{Q}(\mathrm{d}x_{0:T})\beta^{(i)}(x_T)$ and $\mathbb{P}^{(i+1)}(\mathrm{d}x_{0:T}) = \alpha^{(i)}(x_0)\mathbb{Q}(\mathrm{d}x_{0:T})\beta^{(i)}(x_T)$ for every $i\geq 0 $. The convergence of IPF to the Schr\"odinger bridge is then equivalent to $\alpha^{(i)}$ and $\beta^{(i)}$ converging to $\varphi^\circ$ and $\varphi^\star$ as $i \to \infty$, respectively.

In all but simple cases, such as the Gaussian setting treated in Section \ref{sec:lqg}, the policies $\psi^{(i)}$ are not available in closed form due to the intractability of both the Radon--Nikodym derivative \eqref{eq:policy_recursion_end} and the recursion \eqref{eq:policy_recursion}. In what follows, we utilize the connection to optimal control that underlies the framework developed by \citet{heng2017controlled} to build numerical approximations of the policies.

\subsection{Approximate iterative proportional fitting}\label{sec:AIPF}
To simplify notation, let the policy $\psi = \{\psi_t\}_{t\in[1:T]}$ denote our current approximation of
the optimal policy $\psi^\star$ defined in \eqref{eq:harmonic_T}-\eqref{eq:harmonic_t}. In terms of path measures,
$\mathbb{Q}^{\psi}$ with $\psi_0=1$ forms our approximation of the Schr\"odinger bridge $\mathbb{S}=\mathbb{Q}^{\psi^\star}$.
From \eqref{eq:policy_recursion_end} and \eqref{eq:policy_recursion}, the refinement of our current policy $\psi$ is defined by
the backward recursion $\phi_T=\mathrm{d}\pi_T/\mathrm{d}q_T^\psi$, and $\phi_t=M_{t+1}^{\psi}(\phi_{t+1})$ for $t\in[1:T-1]$.
Assuming that the Radon--Nikodym derivative $\phi_T$ can be evaluated pointwise,
we detail how one can use approximate dynamic programming methods
to approximate this recursion in Section \ref{sec:ADP}.
However, $q_T^{\psi}$ is typically intractable, making $\phi_T$ intractable also.
To circumvent this difficulty, we propose estimators of $\phi_T$ in Section \ref{sec:estimate_RN}.
An algorithmic description of the resulting approximate IPF procedure
is summarized in Algorithm \ref{algorithm:aIPF}. The computational cost of step (a) and step (b) is linear in $N$ and $T$. However the computational 
cost of step (c) is strongly dependent on the choice of the regression basis. In  particular, rich function classes might lead to significant overhead.

\subsubsection{Approximate dynamic programming}\label{sec:ADP}
The approximate dynamic programming (ADP) method we will consider estimates the intractable policy refinement
$\phi=\{\phi_t\}_{t\in[1:T]}$ by least squares projections (in log-scale) onto a collection of function classes $\{\mathsf{F}_t\}_{t\in[1:T]}$.
Given $N\in\mathbb{N}$ trajectories $X_{0:T}^n=(X_0^n,\ldots,X_T^n),n\in[1:N]$ sampled from $\mathbb{Q}^{\psi}$,
ADP performs the backward recursion
\begin{align}
\hat{\phi}_T &= \argmin_{f\in\mathsf{F}_T} \sum_{n = 1}^N\left| \log f(X_T^n) - \log \phi_T(X_T^n)\right|^2, \\ 
\hat{\phi}_t &= \argmin_{f \in\mathsf{F}_t} \sum_{n = 1}^N\left| \log f(X_t^n) - \log M^{\psi}_{t+1}(\hat{\phi}_{t+1})(X_{t}^n)\right|^2, \quad t\in[1:T-1],
\end{align}
and returns the (approximately) refined policy $\psi\cdot\hat{\phi}$ forming our next approximation $\mathbb{Q}^{\psi\cdot\hat{\phi}}$ of
$\mathbb{S}$ (with $\psi_0\cdot\hat{\phi}_0=1$ for notational convenience). We refer the reader to \citet[][Section 4]{heng2017controlled} for analysis of the error and large $N$ behavior of this scheme.

The requirement that we can evaluate the conditional expectation $M^{\psi}_{t+1}(\hat{\phi}_{t+1})$ and
sample trajectories from the resulting path measure $\mathbb{Q}^{\psi\cdot\hat{\phi}}$
places restrictions on the choice of initial kernels $\{M_t\}_{t\in[1:T]}$ and the function classes $\{\mathsf{F}_t\}_{t\in[1:T]}$.
We will defer a detailed discussion of these choices to Section \ref{sec:sb_langevin}.


\subsubsection{Estimating the Radon--Nikodym derivative}\label{sec:estimate_RN}
The ADP implementation in the previous section requires the evaluation of $\phi_T=\mathrm{d}\pi_T/\mathrm{d}q_T^\psi$
at the locations $\{X^n_T\}_{n\in[1:N]}$, where $X_{0:T}^n\sim\mathbb{Q}^{\psi}$ for $n\in[1:N]$.
In our context, the estimator of the density $x_T\mapsto q_T^{\psi}(x_T)$ considered in \cite{reich2018data} has the form
$x_T \mapsto N^{-1}\sum_{n=1}^{N} M_{T}^{\psi}(X_{T-1}^{n},x_T)$.
As our transition kernels will be defined by discretizing a Langevin diffusion, this density estimator can perform poorly when the
discretization step size is small, as the kernel $x_T\mapsto M_{T}^{\psi}(x_{T-1},x_T)$ will be highly concentrated around $x_{T-1}$.
Here, we discuss how to obtain unbiased and asymptotically consistent estimates of these Radon--Nikodym derivatives which can be computed in parallel for the $N$ locations.

For simplicity, we denote by $X_{0:T}$ a generic trajectory from $\mathbb{Q}^\psi$ and suppress the notational dependence on $n$.
Our approach makes use of the identity
\begin{equation}
\frac{\mathrm{d}\pi_T}{\mathrm{d}q_T^{\psi}}(x_T) = \int_{\mathsf{E}^T}\frac{\mathrm{d}\mathbb{H}}{\mathrm{d}\mathbb{Q}^{\psi}}(x_{0:T})\mathbb{Q}^{\psi}(\mathrm{d}x_{0:T-1}| x_T),
\end{equation}
which holds for any policy $\psi$ and any $\mathbb{H}\in\mathcal{P}_T(\pi_T)$ such that $\mathbb{H} \ll \mathbb{Q}^\psi$. As in \eqref{eq:extended_target}, we consider path measures $\mathbb{H}$ defined in terms of backward Markov transition kernels
\begin{equation}\label{eq:H}
\mathbb{H}(\mathrm{d}x_{0:T}) = \pi_{T}(\mathrm{d}x_T)\prod_{t=1}^T L_{t-1}(x_t, \mathrm{d}x_{t-1}).
\end{equation}
If we can generate $M+1$ samples $\{X_{0:T-1}^{(m)}\}_{m\in[0:M]}$ from the distribution $\mathbb{Q}^{\psi}(\mathrm{d}x_{0:T-1}| x_T)$, an unbiased estimator of $\phi_T(x_T)$ is given by
\begin{equation}\label{eqn:RN_estimator}
\tilde{\phi}_T(x_T)= \frac{1}{M+1}\sum_{m = 0}^M\frac{\mathrm{d}\mathbb{H}}{\mathrm{d}\mathbb{Q}^{\psi}}(X^{(m)}_{0:T-1},x_T).
\end{equation}
Assuming the realization $X_T = x_T$, then one such sample is already available, since in this case $X_{0:T-1}\sim \mathbb{Q}^{\psi}(\mathrm{d}x_{0:T-1}| x_T)$. Note that in practice, we can replace the normalized measure $\pi_T$ in \eqref{eq:H} with its unnormalized counterpart $\gamma_T$.
In the setting where $\{L_t\}_{t\in[0:T-1]}$ corresponds to the optimal choice of backward kernels, i.e.~such that
\begin{equation}
\prod_{t=1}^{T}L_{t-1}(x_t,\mathrm{d}x_{t-1}) = \mathbb{Q}^{\psi}(\mathrm{d}x_{0:T-1} | x_T),
\end{equation}
then the single sample $X_{0:T-1}$ would be sufficient, as in this case $\tilde{\phi}_T(x_T)$ equals $\phi_T(x_T)$ almost surely, for any $M\geq 0$. This choice is intractable, and we discuss methods for constructing backward kernels that are useful in practice in Section \ref{sec:sb_langevin}.

In cases where $M = 0$ is not sufficient, we can use \textit{conditional} SMC to produce more samples \citep{andrieu:doucet:holenstein:2010}.
In particular, this methodology defines a Markov kernel $K_{x_T}\in\mathcal{M}(\mathsf{E}^T)$ with invariant distribution is $\mathbb{Q}^{\psi}(\mathrm{d}x_{0:T-1} | x_T)$ by utilizing the backward kernels. By setting $X_{0:T-1}^{(0)} = X_{0:T-1}$ and iterating
\begin{equation}
X_{0:T-1}^{(m)} \sim K_{x_T}(X_{0:T-1}^{(m-1)}, \mathrm{d}x_{0:T-1}), \quad m\in[1:M],
\end{equation}
we construct samples that are correlated, but have exactly $\mathbb{Q}^{\psi}(\mathrm{d}x_{0:T-1} | x_T)$ as their marginal distribution. The resulting estimator $\tilde{\phi}_T(x_T)$ therefore stays unbiased, and converges to $\hat{\phi}_T(x_T)$ at the rate of $M^{-1/2}$. A version of the conditional SMC method without resampling is summarized in Algorithm \ref{algorithm:csmc}, and details of its resampling counterpart are discussed in \citet{andrieu:doucet:holenstein:2010}.


The efficiency of this approach and the required number of samples $M+1$ depend on the design of the backward kernels, which remains a challenging task. In fact, approximating the optimal backward kernels for $\mathbb{Q}^{\psi}$ may be harder than for the initial measure $\mathbb{Q}$, as $M_t^{\psi}$ no longer necessarily targets $\pi_t$. In the next section, we illustrate this point for specific choices of reference kernels, derived from discretizing continuous-time processes.

\begin{algorithm}
\caption{\label{algorithm:aIPF} Approximate IPF}
\textbf{Input:} Initial kernels $\{M_t\}_{t\in[1:T]}$, function classes $\{\mathsf{F}_t\}_{t\in[1:T]}$, number of particles $N\in\mathbb{N}$, number of IPF iterations $I\in\mathbb{N}$.
\begin{enumerate}
\item Initialize: Set $\hat{\psi}^{(0)}_t = 1$ for $t \in [0:T]$.
\item For $1\leq i \leq I$:
\begin{enumerate}
\item Sample trajectories $\{X_{0:T}^n\}_{n\in[1:N]}$ from $\mathbb{Q}^{\hat{\psi}^{(i-1)}}$:
for each $n\in [1:N]$, sample $X^n_0 \sim \pi_0(\mathrm{d}x_0)$ and $X^n_t \sim M_t^{\hat{\psi}^{(i-1)}}(X^n_{t-1},\mathrm{d}x_t)$ for $t \in[1:T]$.
\item For each $n\in[1:N]$, compute an estimator $\tilde{\phi}_T^{(i-1)}(X_T^n)$ of $\phi_T^{(i-1)}(X_T^n) = \mathrm{d}\pi_T/\mathrm{d}q_T^{\hat{\psi}^{(i-1)}}(X^n_T)$ using Algorithm \ref{algorithm:csmc}.
\item Approximate dynamic programming: perform the recursion
\begin{align*}
\hat{\phi}_T^{(i)} &= \argmin_{f\in\mathsf{F}_T} \sum_{n = 1}^N\left| \log f(X_T^n) - \log \tilde{\phi}_T^{(i-1)}(X_T^n)\right|^2, \\
\hat{\phi}_t^{(i)} &= \argmin_{f \in\mathsf{F}_t} \sum_{n = 1}^N\left| \log f(X_t^n) - \log M^{\hat{\psi}^{(i-1)}}_{t+1}(\hat{\phi}^{(i)}_{t+1})(X_{t}^n)\right|^2, \quad t\in[1:T-1].
\end{align*}
\item Set $\hat{\psi}^{(i)}_0 = 1$ and $\hat{\psi}^{(i)}_t = \hat{\psi}^{(i-1)}_t \cdot \hat{\phi}^{(i)}_t$ for $t\in[1:T]$.
\end{enumerate}
\end{enumerate}
\textbf{Output:} Policy $\hat{\psi}^{(I)}$.
\end{algorithm}

\begin{algorithm}
\caption{\label{algorithm:csmc} Conditional SMC to produce estimator of $\mathrm{d}\pi_T/\mathrm{d}q_T^\psi$ (without resampling).}
\textbf{Input:} Initial trajectory $X_{0:T-1}\sim \mathbb{Q}^{\psi}(\mathrm{d}x_{0:T-1}| x_T)$, backward kernels $\{L_t\}_{t\in[0:T-1]}$, number of CSMC particles $P\in\mathbb{N}$, number of CSMC iterations $M\in\mathbb{N}\cup \{0\}$.
\begin{enumerate}
\item Initialize: Set $X_{0:T-1}^{(0)} = X_{0:T-1}$.
\item For $m \in [1:M]$:
\begin{enumerate}
\item For $p \in [1:P-1],$ sample independent trajectories $\{Y_{0:T-1}^p\}_{p\in[1:P-1]}$ by sampling $Y_{T-1}^p \sim L_{T-1}(x_T, \mathrm{d}x_{T-1})$ and $Y_{t-1}^p \sim L_{t-1}(Y^p_{t}, \mathrm{d}x_{t-1})$ for $t\in[1:T-1]$. Set $Y_{0:T-1}^P = X_{0:T-1}^{(m-1)}$.
\item Sample an index $A$ from the categorical distribution on $[1:P]$ with probabilities
\begin{equation*}
\left. \frac{\mathrm{d}\mathbb{Q}^{\psi}}{\mathrm{d}\mathbb{H}}(Y_{0:T-1}^p,x_{T}) \middle/ \sum_{q=1}^P \frac{\mathrm{d}\mathbb{Q}^{\psi}}{\mathrm{d}\mathbb{H}}(Y_{0:T-1}^q, x_{T}),\right.  \quad p\in[1:P],
\end{equation*}
where $\mathbb{H}(\mathrm{d}x_{0:T-1}|x_T) = \prod_{t=1}^T L_{t-1}(x_t, \mathrm{d}x_{t-1})$.
\item Set $X_{0:T-1}^{(m)} = Y_{0:T-1}^A$.
\end{enumerate}
\end{enumerate}
\textbf{Output:} Estimator $$\tilde{\phi}_T(x_T) = \frac{1}{M+1}\sum_{m = 0}^M\frac{\mathrm{d}\mathbb{H}}{\mathrm{d}\mathbb{Q}^{\psi}}(X^{(m)}_{0:T-1},x_T).$$
\end{algorithm}

\subsection{Leveraging continuous-time Schr\"odinger bridges}\label{sec:sb_continuous}
In the following, we will select the initial kernels $\{M_t\}_{t\in[1:T]}$
to be the Euler--Maruyama discretization of the Langevin dynamics in \eqref{eq:langevin} \citep{parisi1981,grenandermille1994}.
Following \citet{heng2017controlled}, we do not employ Metropolis--Hastings corrections as this results in
conditional expectations and path measures that are intractable to evaluate and sample, respectively.
Compared to standard AIS and SMC samplers, the latter is a limitation of our proposed methodology.
As our algorithmic choices will be guided by approximating the solution of
the Schr\"odinger bridge problem for the continuous-time process, the absence of Metropolis--Hastings corrections
will not be crucial in the small step-size regime. By leveraging the structure of continuous-time Schr\"odinger bridges,
the following proposes a version of the IPF algorithm that allows flexible function classes to be used.

\subsubsection{Schr\"odinger bridges for diffusions}\label{sec:sb_langevin}
First, we review some of the continuous-time problem's features, restricting ourselves to reference processes of the kind
\begin{equation}\label{eq:langevin_reference}
\mathrm{d}X_s = b_s(X_s) \mathrm{d}s + \mathrm{d}W_s, \quad s\in[0,\tau],
\end{equation}
where $b:[0,\tau]\times\mathsf{E}\rightarrow\mathsf{E}$, $W_s$ denotes the standard Brownian motion.
Let $Q$ denote the law of this process on the space of continuous $\mathsf{E}$-valued paths over $[0,\tau]$.
Important special cases include Brownian motion, for which $b_s = 0$ for all $s\in(0,\tau]$, as well as the Langevin dynamics which we consider in the next section.

As in the discrete-time setting, the continuous-time Schr\"odinger bridge problem with respect to $Q$ can be expressed as the minimization of KL divergence over processes $H$ that satisfy the marginal constraints $H_0 = \pi_0$ and $H_\tau = \pi$ \citep[see e.g.][]{leonard2014survey}. Following \citet{dai1991stochastic}, it can equivalently be described as a stochastic control problem, i.e. we seek the change in drift $s\mapsto u_s$ such that the process
\begin{equation}\label{eq:langevin_controlled}
\mathrm{d}Y_s = \left\{b_s(Y_s)  + u_s(Y_s)\right\}\mathrm{d}s + \mathrm{d}W_s, \quad s\in[0,\tau],
\end{equation}
satisfies $Y_0\sim\pi_0$ and $Y_\tau \sim \pi$, and $s\mapsto u_s$ minimizes the cost function $\mathbb{E}\int_0^\tau \|u_s(Y_s)\|^2 \mathrm{d}s$. Using this perspective, \citet{dai1991stochastic} showed, under some assumptions, that the optimal controller $s\mapsto u^\star_s$ has the form $u^\star_s = \nabla \log \psi^\star_s$ and
the transition probabilities are given by Doob's $h$-transform
\begin{align}\label{eqn:continuous_transitions}
S(\mathrm{d}x_t|x_s)=\frac{Q(\mathrm{d}x_t|x_s)\psi_t^\star(x_t)}{\psi_s^\star(x_s)},\quad 0\leq s<t\leq \tau.
\end{align}
The harmonic functions $s\mapsto\psi_s^\star$ can be defined using the continuous-time analog of the equations \eqref{eq:harmonic_T} and \eqref{eq:harmonic_t}.

Consider the case where $b_s = \frac{1}{2}\nabla \log \pi_s$ for some smooth curve of distributions $(\pi_s)_{s\in[0,\tau]}$ such that $\lim_{s\to 0} = \pi_0$ and $\lim_{s \to \tau}\pi_s = \pi$. Examples of such curves include continuous versions of the geometric interpolation in \eqref{eq:geometric_path}. Provided the drift function $s\mapsto \nabla \log \pi_s$ changes \textit{adiabatically}, or infinitely slowly in time, the distribution of the particle $X_s$ is equal to $\pi_s$ for all $s\in[0,\tau]$ \citep{chiang1987diffusion,patra2017shortcuts}. In all but trivial cases, this effectively requires that $\tau\to\infty$. However, for $\tau$ large enough, one expects the distribution of $X_\tau$ to be close to $\pi$.

\subsubsection{Schr\"odinger bridges for discretized Langevin dynamics}\label{sec:sb_langevin}
From the sampling perspective, the above discussion motivates using the kernels $\{M_t\}_{t\in[1:T]}$ defined by the Euler--Maruyama discretization
\begin{equation}\label{eq:em_langevin}
M_t(x_{t-1},x_t) = \mathcal{N}\left(x_t; x_{t-1} + \frac{h}{2}\nabla \log \pi_t(x_{t-1}), h\mathcal{I}_d \right), \quad t\in[1:T],
\end{equation}
to initialize the path measure $\mathbb{Q}$, where $\{\pi_t\}_{t\in[0:T]}$ is a discretization of the curve $(\pi_s)_{s\in[0,\tau]}$. If the step-size $h = \tau/T$ is small and $\tau$ is large, we expect the marginal distribution $q_T$ to be close to $\pi_T = \pi$, thus providing us with a good starting point. This can also be seen by noting that if one chooses the backward kernel $L_{t-1}(x_t,x_{t-1}) = M_t(x_t, x_{t-1})$, the resulting SMC sampler would approximate AIS
in the small $h$ setting \citep{heng2017controlled}.
However, except for Gaussian policies, sampling from Markov kernels that are twisted by non-conjugate policies is challenging. We now exploit the continuous-time perspective to circumvent this difficulty.

Suppose that we have an approximation $(\psi_s)_{s\in[0,\tau]}$ of the policy $(\psi_s^\star)_{s\in[0,\tau]}$ that defines Doob's $h$-transform \eqref{eqn:continuous_transitions}.
Consider the Euler--Maruyama discretization of \eqref{eq:langevin_controlled} in the case where $b_s = \frac{1}{2}\nabla \log \pi_s$ and $u_s = \nabla \log \psi_s$
and define
\begin{equation}\label{eq:twisted_kernel_approx}
M^{\bar{\psi}}_t(x_{t-1},x_t) = \mathcal{N}\left(x_t; x_{t-1} + \frac{h}{2}\nabla \log \pi_t(x_{t-1}) + h\nabla \log \psi_t(x_{t-1}), h\mathcal{I}_d \right),
\end{equation}
for $t\in[1:T].$ Provided that the functions $\{\psi_t\}_{t\in[0:T]}$ are appropriately smooth, the first-order Taylor expansion of $\log \psi_t$ around $x_{t-1}$ can be written
\begin{equation}
\log \psi_t(x_t) = \log \bar{\psi_t}(x_{t-1}, x_t) + \mathcal{O}(\|x_t - x_{t-1}\|^2),
\end{equation}
as $\|x_t - x_{t-1}\| \to 0$, where $\log \bar{\psi_t}(x_{t-1}, x_t) = \log \psi_t (x_{t-1}) + \nabla \log \psi_t (x_{t-1})^\top (x_t - x_{t-1})$.
These kernels can also be expressed as
\begin{equation}
M^{\bar{\psi}}_t(x_{t-1},x_t) = \frac{M_t(x_{t-1},x_t)\bar{\psi_t}(x_{t-1}, x_t)}{M_t(\bar{\psi_t})(x_{t-1})},
\end{equation}
where $\log M_t(\bar{\psi_t})(x_{t-1})  = \log \psi_t(x_{t-1}) + \frac{h}{2}\nabla \log \psi_t(x_{t-1})^\top \left[\nabla \log \pi_t(x_{t-1}) +\nabla \log \psi_t(x_{t-1})\right]$. In the small $h$ regime, we expect $\|x_t - x_{t-1}\|^2$ to be of size $\mathcal{O}(h)$ under $M_t(x_{t-1},\mathrm{d}x_t)$. Hence, $\{M^{\bar{\psi}}_t\}_{t\in[1:T]}$ are likely to provide good approximations of the kernels $\{M^{\psi}_t\}_{t\in[1:T]}$ when $h$ is small. Similar arguments also hold for the analogously defined kernels $\{M^{\tilde{\psi}}_t\}_{t\in[1:T]}$, where the second-order approximation $\log \tilde{\psi_t}(x_{t-1}, x_t) = \log \bar{\psi_t}(x_{t-1}, x_t)  + \frac{1}{2}(x_t - x_{t-1})^\top \nabla^2\log\psi_t(x_{t-1})(x_t - x_{t-1})$, in which $\nabla^2 f$ denotes the Hessian of $f$.

The benefit of approximating $\log \psi_t(x_t)$ with its linearization $\log\bar{\psi}(x_{t-1},x_t)$ around $x_{t-1}$ for small step sizes is that sampling from $M^{\bar{\psi}}_t(x_{t-1},\mathrm{d}x_t)$ and evaluating the integrals $\log M_t(\bar{\psi_t})(x_{t-1})$ only require access to the pointwise evaluation of the functions $\log \psi_t$ and $\nabla \log \psi_t$ (the second-order approximation additionally requires evaluations of $\nabla^2 \log \psi_t$). Thus, using these approximations one could feasibly run the ADP algorithm without requiring the function classes and kernels to be conjugate. This would allow us to make use of flexible function estimation methods such as neural networks \citep{hure2018deep} within the approximate IPF algorithm. A version of approximate IPF using Euler--Maruyama kernels is given in Algorithm \ref{algorithm:aIPF_langevin} in Appendix \ref{appendix:algos}.

The continuous-time perspective also suggests a natural choice of backward kernels. For diffusion processes of the form \eqref{eq:langevin_controlled},  \citet{haussmann1986time} showed that the time-reversed process $\tilde{Y}_s=Y_{\tau-s}$ satisfies
\begin{equation}\label{eq:langevin_controlled_reversal}
\mathrm{d}\tilde{Y}_s = \left\{ \nabla \log \rho_s(\tilde{Y}_s) - b_s(\tilde{Y}_s)  - u_s(\tilde{Y}_s)\right\}\mathrm{d}s + \mathrm{d}\tilde{W}_s, \quad s\in[0,\tau],
\end{equation}
where $\rho_s$ is the marginal density of the process \eqref{eq:langevin_controlled} at time $s$, and $\tilde{W}_s$ is another standard Brownian motion. The backward kernels could therefore be sensibly chosen to be the Euler--Maruyama discretization of \eqref{eq:langevin_controlled_reversal}, but the term $\nabla \log \rho_s$ is typically intractable. If $b_s = \frac{1}{2}\nabla \log \pi_s$ and $u_s = \nabla \log \psi_s$ is an approximation of the optimal forward drift $\nabla \log \psi_s^\star$, a generic choice could be to replace $\rho_s$ by $\pi_s$. In this case, the backward kernels would amount to
\begin{equation}\label{eq:twisted_backward_kernel_approx}
L^{\bar{\psi}}_{t-1}(x_t,x_{t-1}) = \mathcal{N}\left(x_{t-1}; x_{t} + \frac{h}{2}\nabla \log \pi_{t-1}(x_{t}) - h\nabla \log \psi_t(x_{t}), h\mathcal{I}_d \right),
\end{equation}
for  $t\in[1:T].$ As for the forward process, a second-order approximation to $\log \psi_t$ can also be used. However, achieving low variance of the Radon--Nikodym derivative estimator \eqref{eqn:RN_estimator} relies on the assumption
that the marginal of $Q$ at time $s$ is close to $\pi_s$, which might not be the case.
In Section \ref{sec:ssb}, we circumvent the intractability of $\rho_s$ by designing a process in which the marginal distributions can be consistently approximated by $\pi_s$, making the kernel in \eqref{eq:twisted_backward_kernel_approx} more suitable.

\subsection{Schr\"odinger bridge sampling}
Given the policy $\hat{\psi}^{(I)}$ after $I \geq 1$ iterations of the approximate IPF algorithm, one can proceed to use the proposal $\mathbb{Q}^{\hat{\psi}^{(I)}}$ within importance sampling or SMC on path space. The terminal marginal $q_T^{\hat{\psi}^{(I)}}$ is likely to be closer to the target distribution $\pi_T = \pi$ than the initial $q_T$, but this condition is not in itself enough to guarantee efficient sampling. In particular, as mentioned in the previous sections, constructing backward kernels $\{L_t\}_{t\in[0:T-1]}$ to define the target distribution on path space remains challenging. Moreover, the incremental weights within an importance sampling or SMC scheme would more naturally be based on the sequence of Schr\"odinger bridge marginals $\{s_t\}_{t\in[0:T]}$ rather than $\{\pi_t\}_{t\in[0:T]}$, but these distributions are intractable (even up to normalizing constants). In Section \ref{sec:ssb_samplers}, we propose a scheme based on the \textit{multi-marginal} Schr\"odinger bridge problem which helps mitigate these issues.

\subsection{Example: linear quadratic Gaussian} \label{sec:lqg}
We illustrate the computation of Schr\"odinger bridges on a simple example in which the initial, target, and transitions are all Gaussian. In this simple linear quadratic Gaussian (LQG) scenario, IPF can be implemented exactly. Let $\pi_0(\mathrm{d}x_0) = \mathcal{N}(x_0; \mu_0, \Sigma_0)\mathrm{d}x_0$ and $\pi_T(\mathrm{d}x_T) = \mathcal{N}(x_T; \mu_T, \Sigma_T)\mathrm{d}x_T$ for some $\mu_0,\mu_T\in \mathbb{R}^d$ and $\Sigma_0,\Sigma_T\in \mathbb{R}^{d\times d}$, and for each $t\in[1:T]$, let $M_t(x_{t-1},\mathrm{d}x_t) = \mathcal{N}(x_t; K_tx_{t-1} + r_t, H_t)\mathrm{d}x_t$ for some $r_t \in \mathbb{R}^d$ and $K_t, H_t\in \mathbb{R}^{d\times d}$. By conjugacy, it can be shown that the exact IPF iterations can be written in terms of policies of the form $-\log\psi_t^{(i)}(x_t) = x_t^\top A_t^{(i)}x_t + x_t^\top b_t^{(i)} + c_t^{(i)}$ for every $i \geq 1$ and $t\in[0:T]$. We derive expressions for the coefficients $A_t^{(i)}, b_t^{(i)}, c_t^{(i)}$ and the resulting twisted Markov kernels in Appendix \ref{appendix:lqg}. We compare the exact potentials $\psi^{(i)}$ and the corresponding path measure $\mathbb{Q}^{\psi^{(i)}}$ with their approximations $\hat{\psi}^{(i)}$ and $\mathbb{Q}^{\hat{\psi}^{(i)}}$ constructed with approximate IPF. We choose the function classes $\mathsf{F}=\{\mathsf{F}_t\}_{t\in[0:T]}$ to contain only Gaussians. This choice is \textit{well-specified} in the sense that $\psi^{(i)}\in\mathsf{F}$ for every $i\geq 1$.

In particular, we consider a prior distribution $\pi_0$ with $\mu_0 = 0$ and $\Sigma_0 = \mathcal{I}_d$ and a log-likelihood defined by $\ell(x) = -(y-x)^\top R^{-1} (y-x)/2$ for some observation $y\in\mathbb{R}^d$ and symmetric positive definite $R\in\mathbb{R}^{d\times d}$, giving rise to a posterior distribution $\pi=\pi_T$ with $\Sigma_T = \left(\Sigma_0^{-1} + R^{-1}\right)^{-1}$ and $\mu_T = \Sigma_T\left(\Sigma_0^{-1}\mu_0 + R^{-1}y\right)$. By conjugacy, the distributions defining the geometric path \eqref{eq:geometric_path} are Gaussian: $\pi_t(\mathrm{d}x_t) = \mathcal{N}(x_t; \mu_t, \Sigma_t)\mathrm{d}x_t$ with $\Sigma_t = \left(\Sigma_0^{-1} + \lambda_t R^{-1}\right)^{-1}$ and $\mu_t = \Sigma_t \left(\Sigma_0^{-1}\mu_0 + \lambda_t R^{-1}y\right)$. Let $y = (\xi,\dots,\xi)^\top$ and $R$ have 1's on the diagonal and $\rho$'s on the off-diagonal. Here, we take $d = 2$, $\xi = 8$ and $\rho = 0.8$.

We consider two different settings of the reference process parameters $\{K_t\}_{t\in[1:T]}$, $\{r_t\}_{t\in[1:T]}$ and $\{H_t\}_{t\in[1:T]}$. In the first, they are constructed from discretizing Brownian motion over the time-interval $[0,\tau]$ with $\tau = 2$ and step size $h = 1/20$. That is, we take $K_t = \mathcal{I}_d$, $r_t = 0$ and $H_t = h \mathcal{I}_d$ for every $t\in[1:T]$, with $T = 40$. In Section \ref{sec:connections}, we discuss the connection between Schr\"odinger bridges with Brownian motion reference process and the optimal transport problem with quadratic cost. In the second setting, we consider the discretization of the Langevin dynamics discussed in Section \ref{sec:sb_langevin}, i.e. $\lambda_t = t/T$, $K_t = \mathcal{I}_d - h\Sigma_t^{-1}/2$, $r_t = h\Sigma_t^{-1}\mu_t/2$ and $H_t = h \mathcal{I}_d$ for $t\in[1:T]$, with $T$ and $h$ as in the Brownian setting. In both cases, we take the backward kernels to be \eqref{eq:twisted_backward_kernel_approx}.

To measure the discrepancy between the marginals of the Schr\"odinger bridge $s_t$ and the marginals obtained with $i$ iterations of IPF $q_t^{(i)}$, we compute the corresponding 2-Wasserstein distance. Since the distributions involved are Gaussian, the 2-Wasserstein distance between them is known in closed form \citep[see e.g.][Chapter 2.6]{peyre2018computational}: for Gaussian distributions $\nu_i$ with means $\gamma_i$ and covariances $\Gamma_i$ for $i = 1,2$, we have $\was_2^2(\nu_1,\nu_2) = \|\gamma_1 - \gamma_2\|^2 + B(\Gamma_1,\Gamma_2)^2,$
where $B$ is the Bures metric between positive definite matrices \citep{bures1969extension}.

In Figures \ref{fig:lqg_brownian_IPF} and \ref{fig:lqg_langevin_IPF} we plot $\log \was_2(s_t, q_t^{(i)})$ for $t \in[1:T]$, $i \in [0:5]$, and exact and approximate IPF for different number of conditional SMC iterations $M$ (using $P = 128$ CSMC particles). In both Figures \ref{fig:brownian_IPF_exact} and \ref{fig:langevin_IPF_exact}, we observe that the marginals given by the exact IPF iterations appear to converge exponentially fast, as seen in Proposition \ref{prop:IPF_convergence}.
Moreover, we note that the errors are smaller for the Langevin reference process and for early times as we are initializing both reference processes from $\pi_0$.
For the approximate schemes, we observe that performing only one or two iterations of IPF leads to large reductions in the distance to the Schr\"odinger bridge, but that the benefit of performing further IPF iterations depends on the variance of the Radon--Nikodym estimates via the value of $M$. For less variable estimates, the benefit of further IPF iterations is larger. This also illustrates that the choice of backward kernels \eqref{eq:twisted_backward_kernel_approx} is suboptimal in these settings, which motivates the method proposed in Section \ref{sec:ssb}. Increasing the number of particles $N$ did not lead to qualitatively different behavior across the IPF iterations, but reduced the width of the confidence bands (not illustrated).

\begin{figure}[hp]
        \begin{subfigure}[t]{0.32\textwidth}
            \centering
            \includegraphics[width=\textwidth]{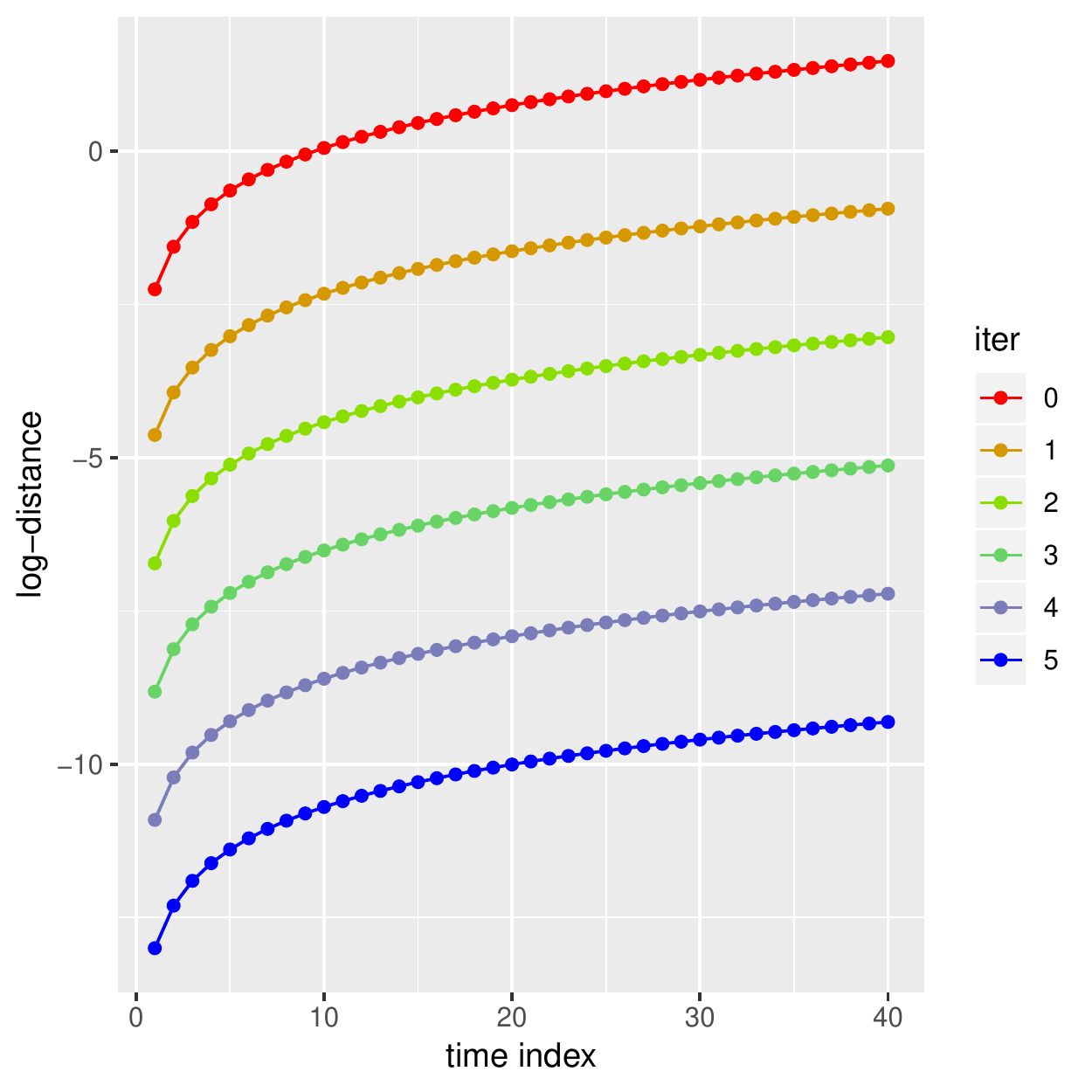}
            \caption{{Exact IPF.}}
            \label{fig:brownian_IPF_exact}
        \end{subfigure}
        \begin{subfigure}[t]{0.32\textwidth}
            \centering
            \includegraphics[width=\textwidth]{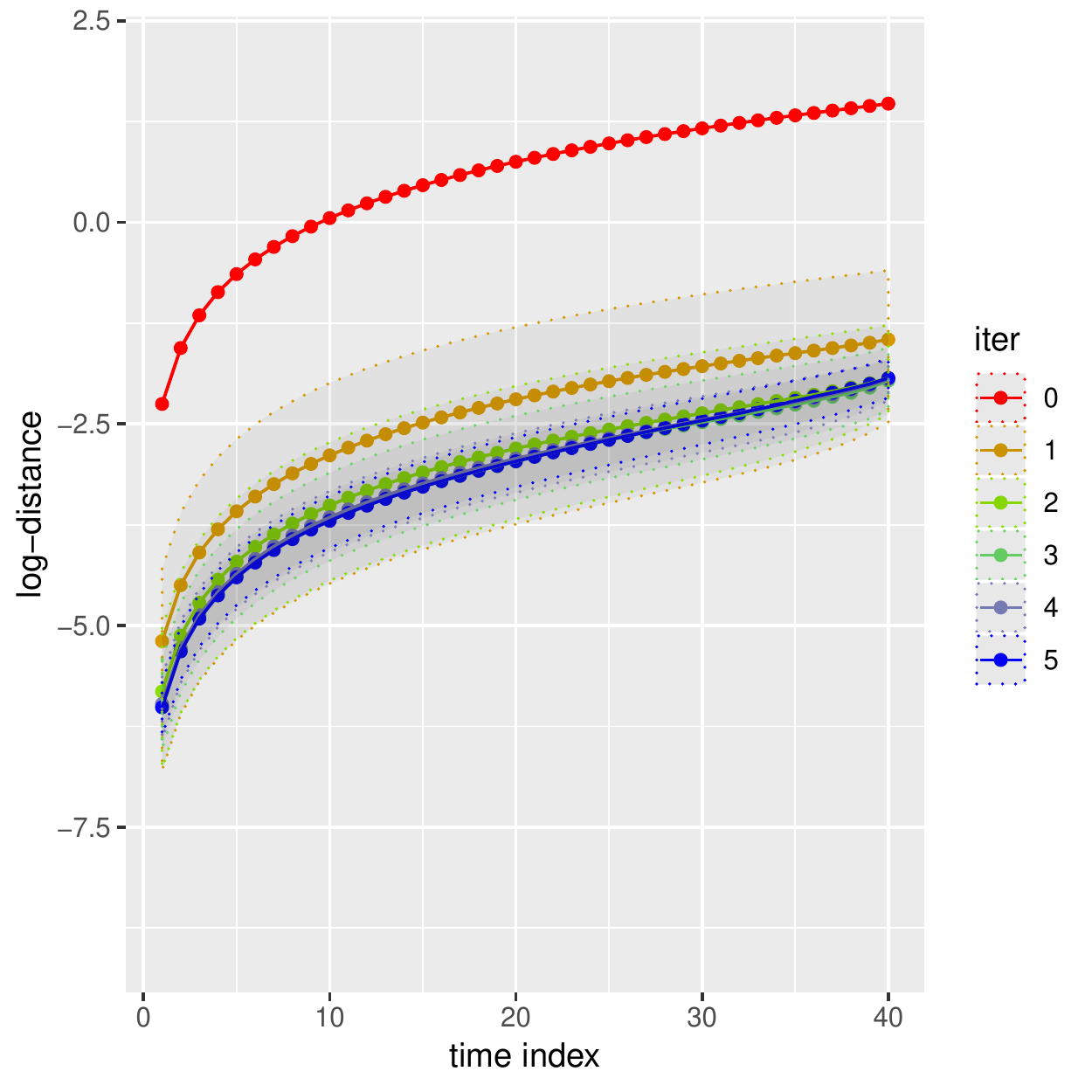}
            \caption{{\small Approx.~IPF, $M=0$.}}
            \label{fig:brownian_IPF_approx_IPF5_reps100_n1000_m0}
        \end{subfigure}
        \begin{subfigure}[t]{0.32\textwidth}
            \centering
            \includegraphics[width=\textwidth]{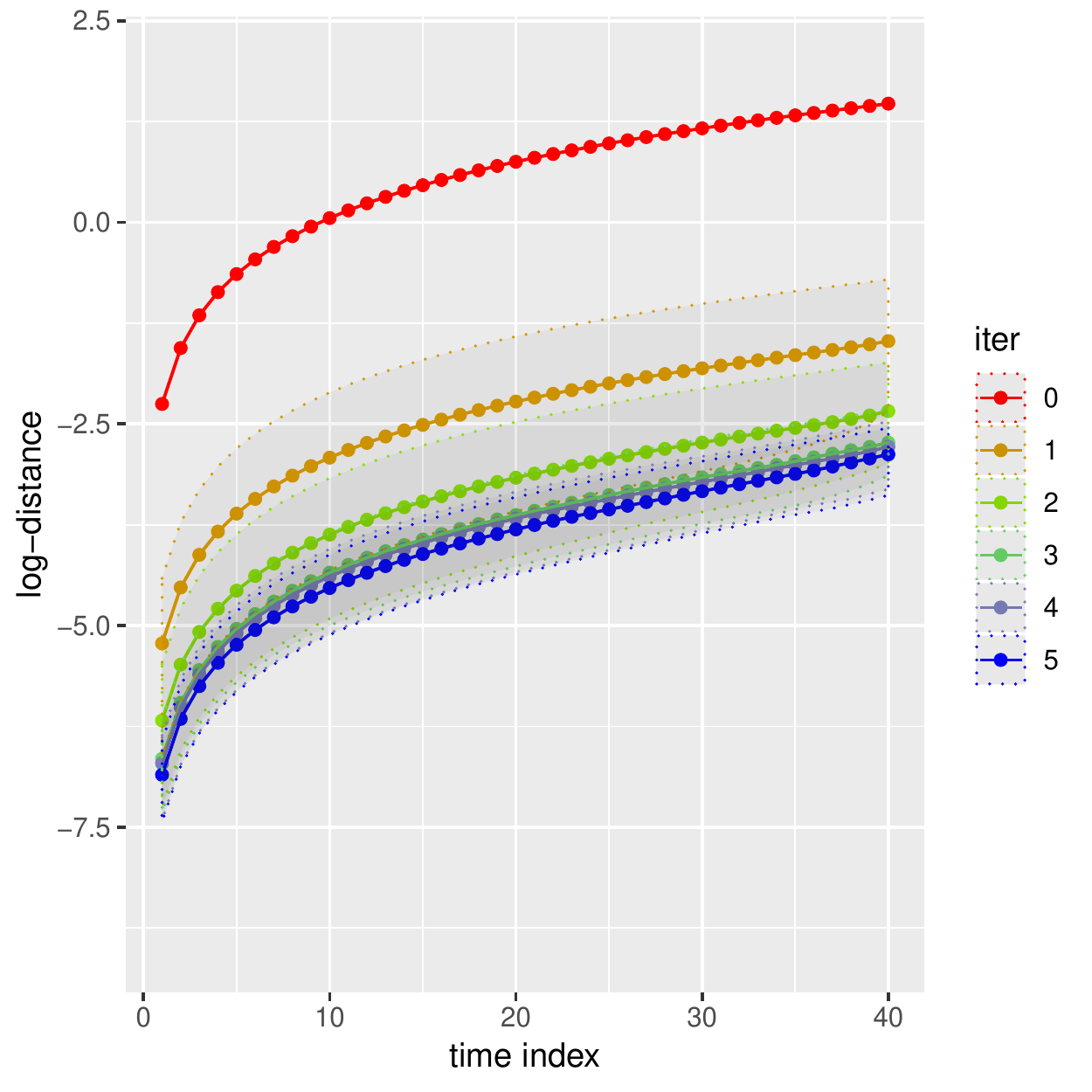}
            \caption{{\small Approx.~IPF, $M=5$.}}
            \label{fig:brownian_IPF_approx_IPF5_reps100_n1000_m5}
        \end{subfigure}

	\vspace*{0.8cm}
         \begin{subfigure}[t]{0.32\textwidth}
            \centering
            \includegraphics[width=\textwidth]{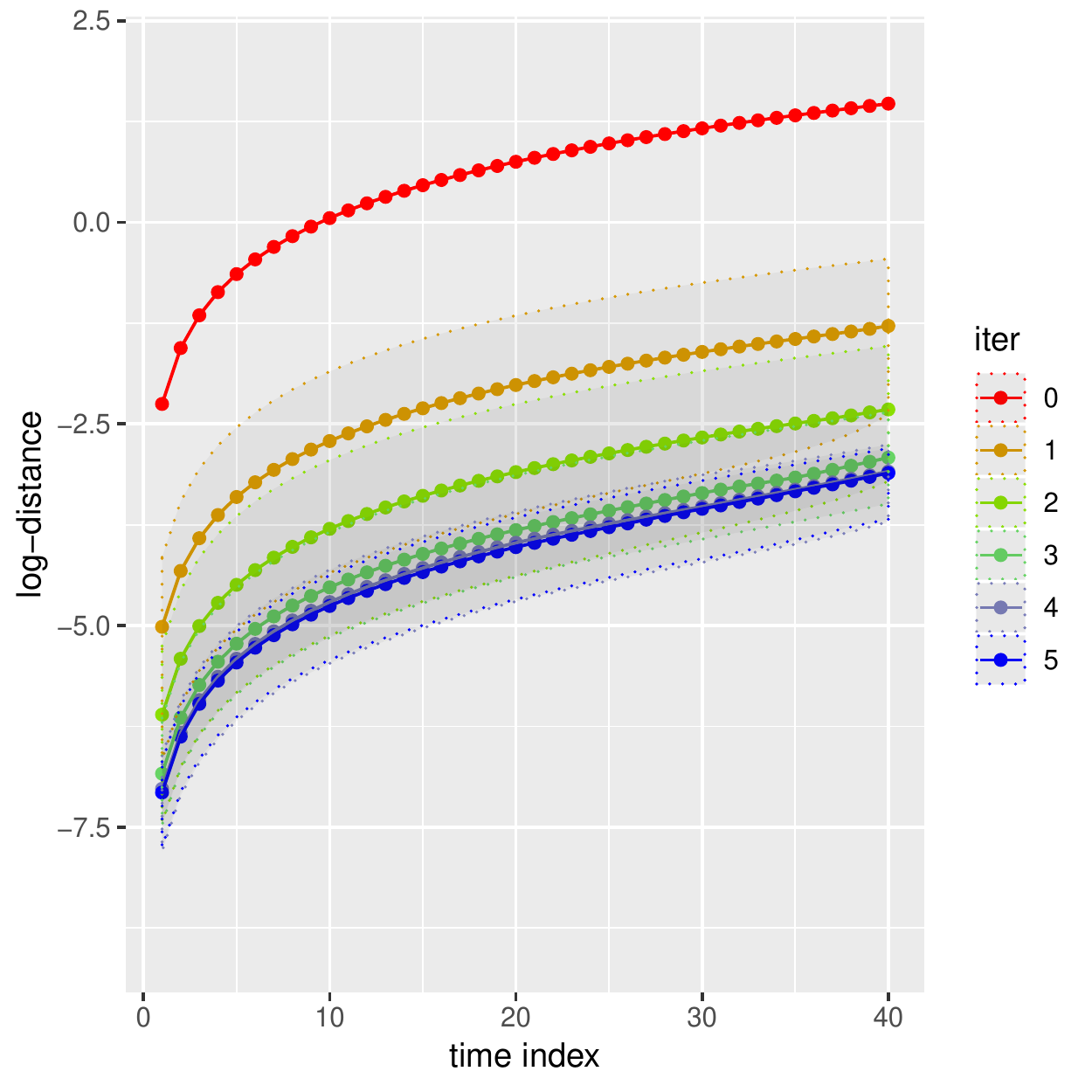}
            \caption{{\small Approx.~IPF, $M=10$.}}
            \label{fig:brownian_IPF_approx_IPF5_reps100_n1000_m10}
        \end{subfigure}
                 \begin{subfigure}[t]{0.32\textwidth}
            \centering
            \includegraphics[width=\textwidth]{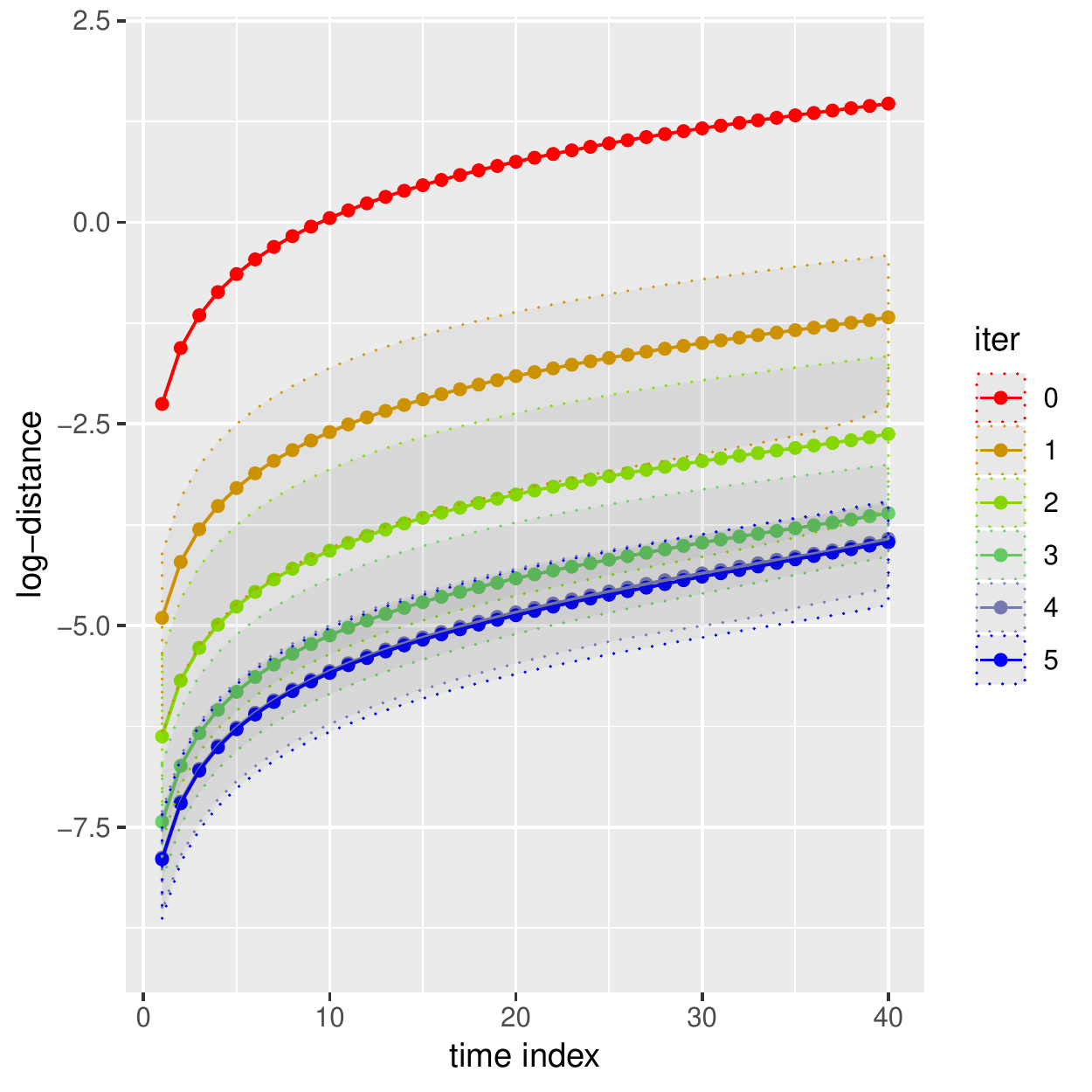}
            \caption{{\small  Approx.~IPF, $M=50$.}}
            \label{fig:brownian_IPF_approx_IPF5_reps100_n1000_m50}
        \end{subfigure}
        \begin{subfigure}[t]{0.32\textwidth}
            \centering
            \includegraphics[width=\textwidth]{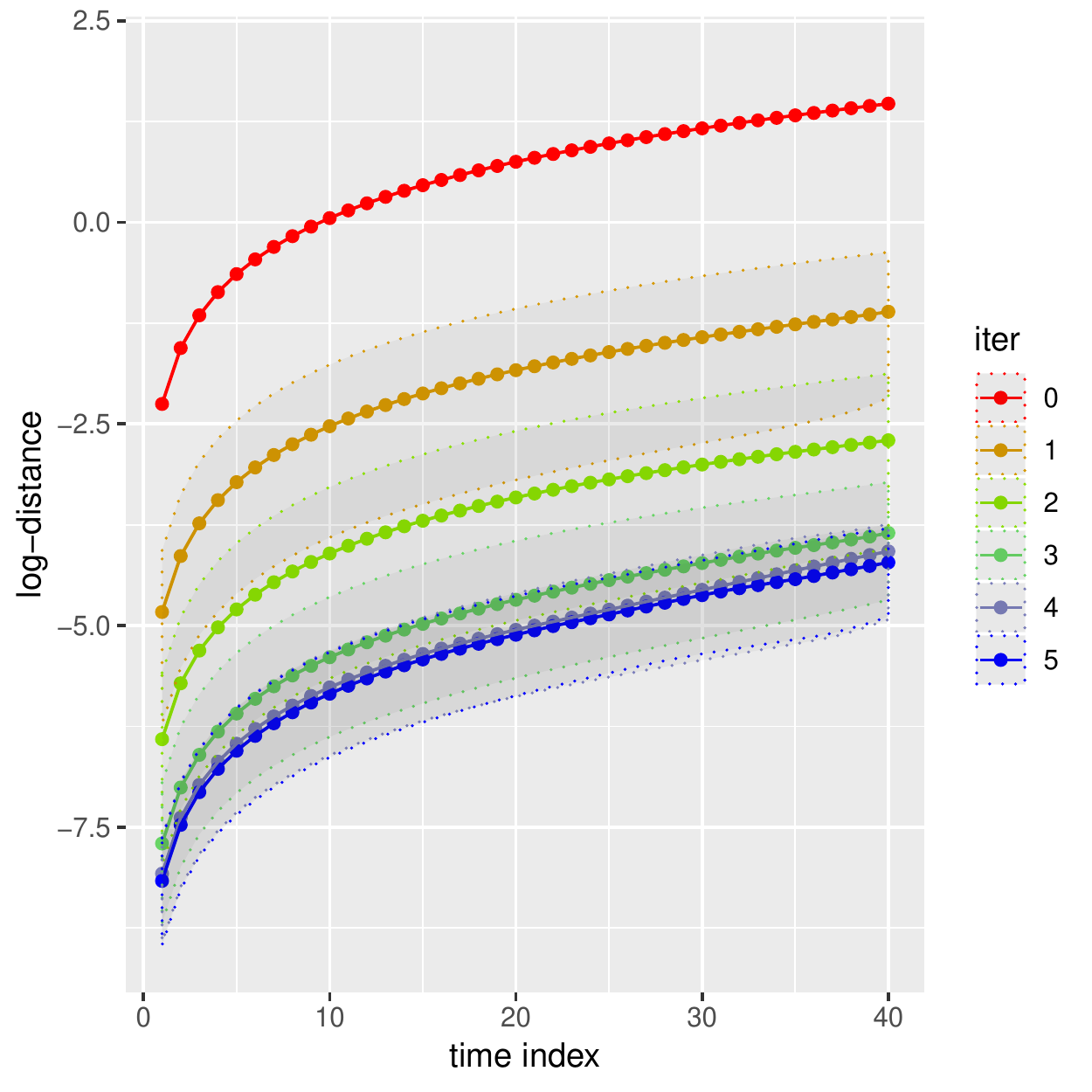}
            \caption{{\small  Approx.~IPF, $M=100$.}}
            \label{fig:brownian_IPF_approx_IPF5_reps100_n1000_m100}
        \end{subfigure}
        \caption{\small Distances between marginals of the Schr\"odinger bridge $s_t$ and the marginals of the IPF iterates $q_t^{(i)}$, measured as $\log \was_2(s_t,q_t^{(i)})$, for the LQG setting of Section \ref{sec:lqg} with discretized Brownian diffusion reference dynamics. Figure \ref{fig:brownian_IPF_exact} corresponds to the exact computation of the IPF iterates for $i \in [0:5]$. The remaining plots correspond to the proposed particle-based approximation of the IPF iterations using $N =1,000$ particles and different values of conditional SMC iterations $M$ (using $P = 128$ CSMC particles). The solid lines correspond to the median value of the log-distance calculated over $100$ independent simulations, and corresponding confidence bands represent the $5\%$ and $95\%$ quantiles. Note that the vertical axis of Figure \ref{fig:brownian_IPF_exact} is on a different scale than those of the other figures, as the exact IPF iterations yield smaller distances in general.}
        \label{fig:lqg_brownian_IPF}
\end{figure}

\begin{figure}[hp]
        \begin{subfigure}[t]{0.32\textwidth}
            \centering
            \includegraphics[width=\textwidth]{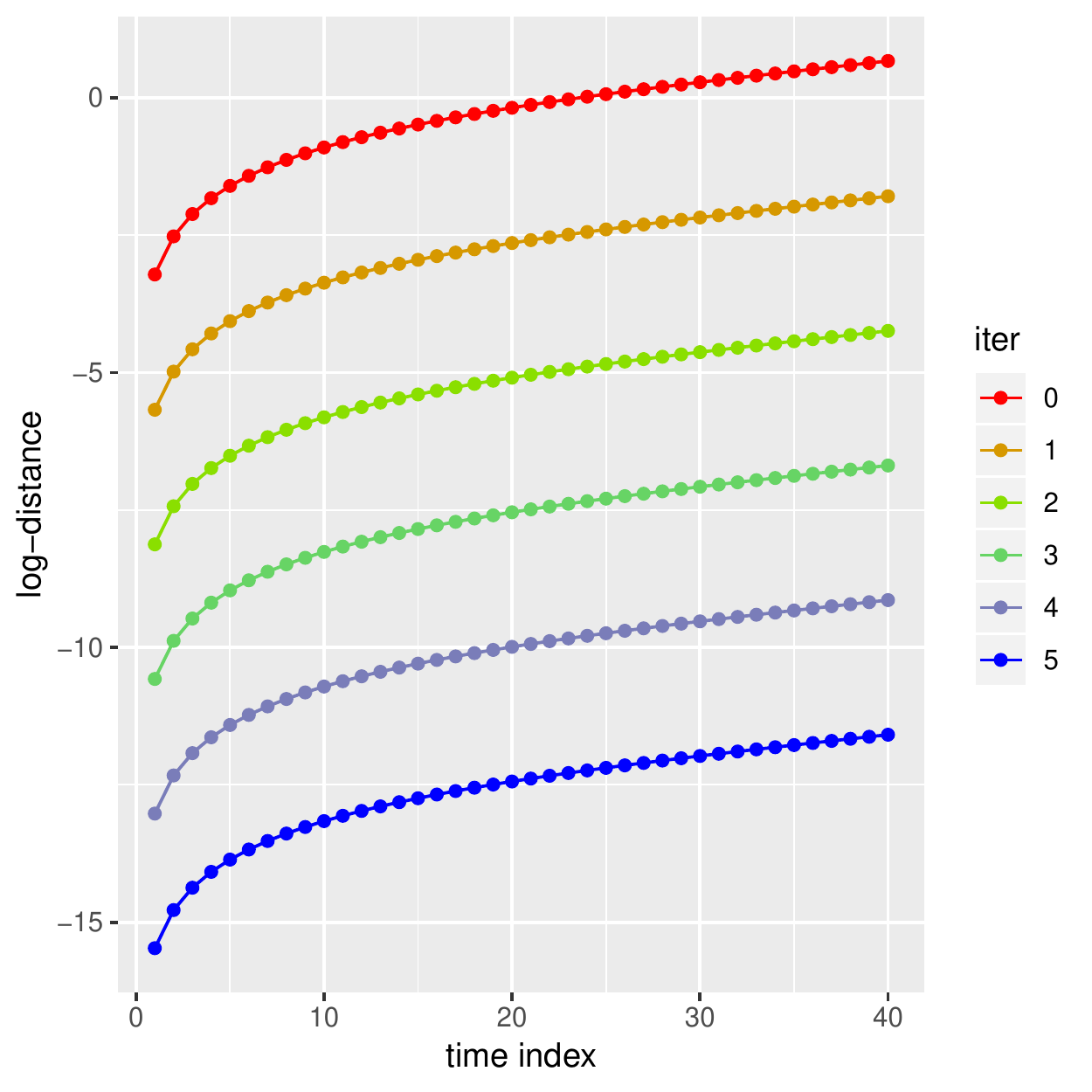}
            \caption{{Exact IPF.}}
            \label{fig:langevin_IPF_exact}
        \end{subfigure}
        \begin{subfigure}[t]{0.32\textwidth}
            \centering
            \includegraphics[width=\textwidth]{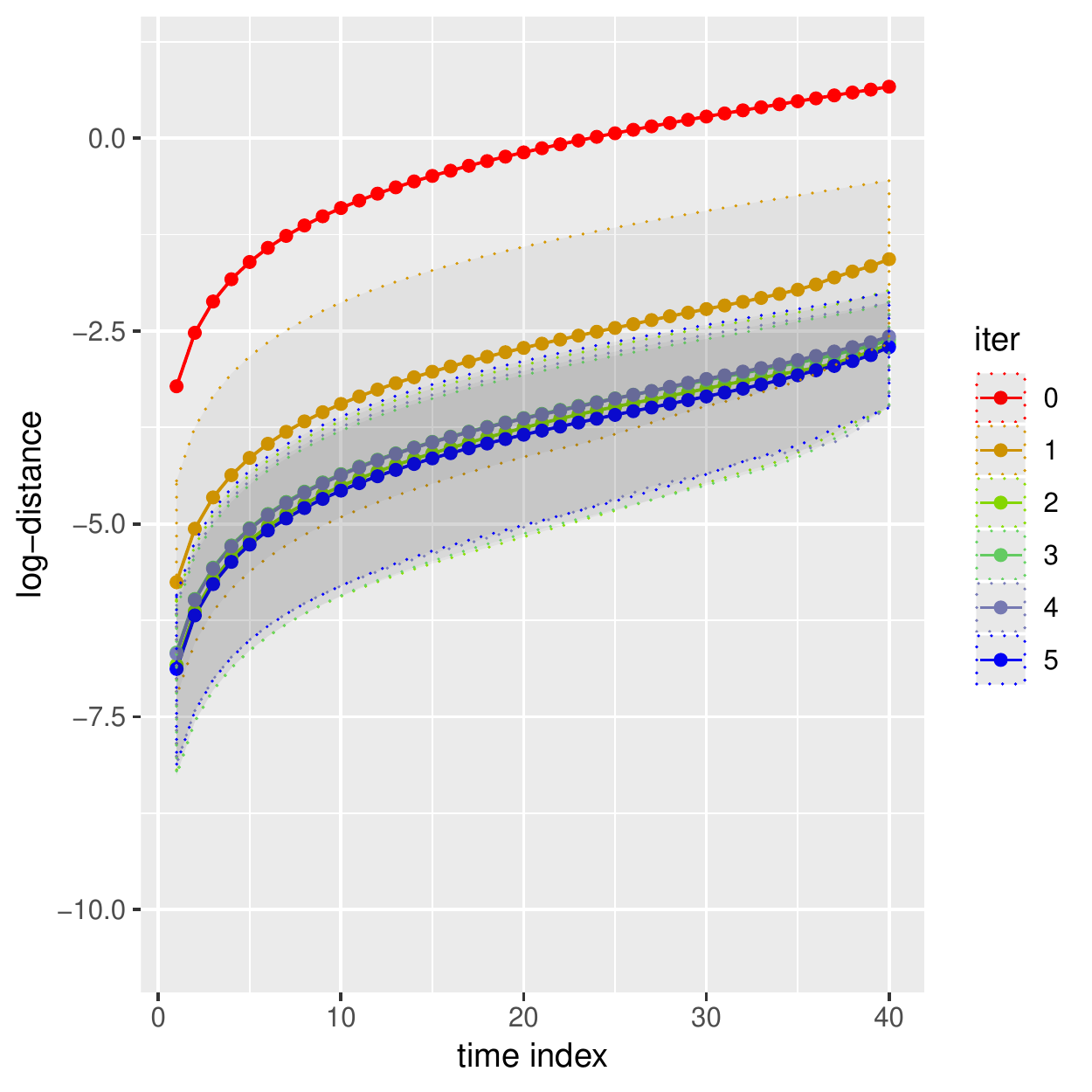}
            \caption{{\small Approx.~IPF, $M=0$.}}
            \label{fig:langevin_IPF_approx_IPF5_reps100_n1000_m0}
        \end{subfigure}
        \begin{subfigure}[t]{0.32\textwidth}
            \centering
            \includegraphics[width=\textwidth]{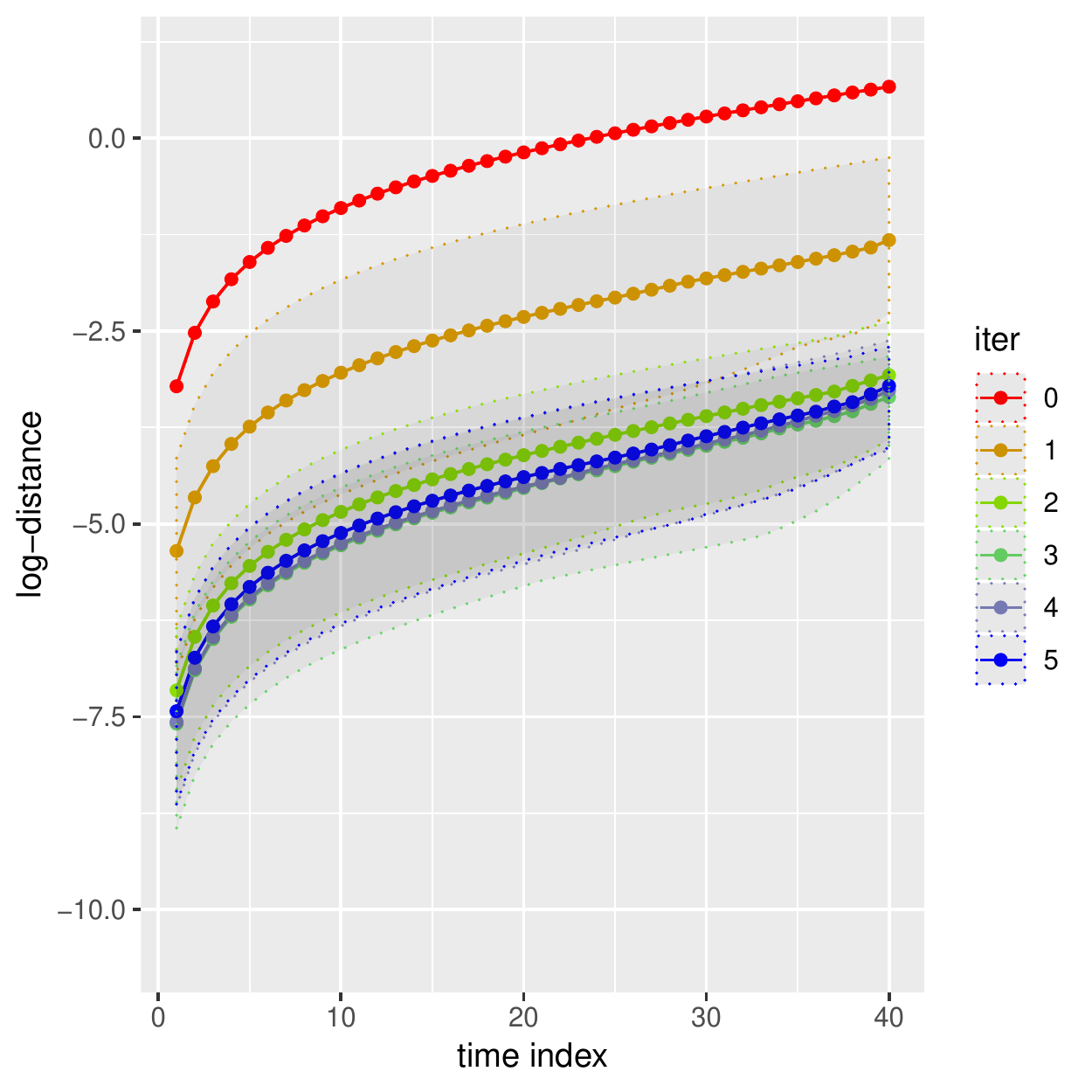}
            \caption{{\small Approx.~IPF, $M=5$.}}
            \label{fig:langevin_IPF_approx_IPF5_reps100_n1000_m5}
        \end{subfigure}
	
	\vspace*{0.8cm}
         \begin{subfigure}[t]{0.32\textwidth}
            \centering
            \includegraphics[width=\textwidth]{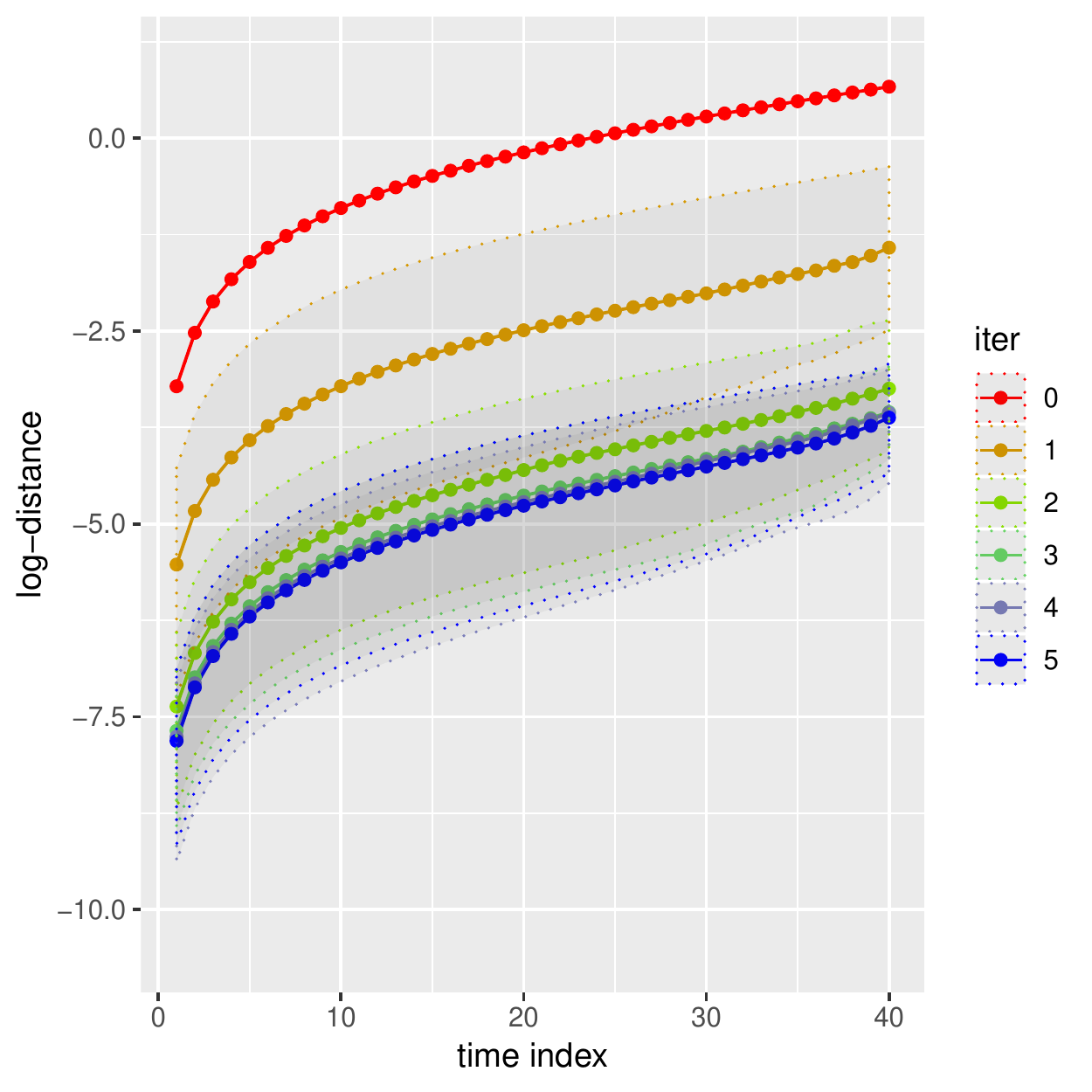}
            \caption{{\small Approx.~IPF, $M=10$.}}
            \label{fig:langevin_IPF_approx_IPF5_reps100_n1000_m10}
        \end{subfigure}
                 \begin{subfigure}[t]{0.32\textwidth}
            \centering
            \includegraphics[width=\textwidth]{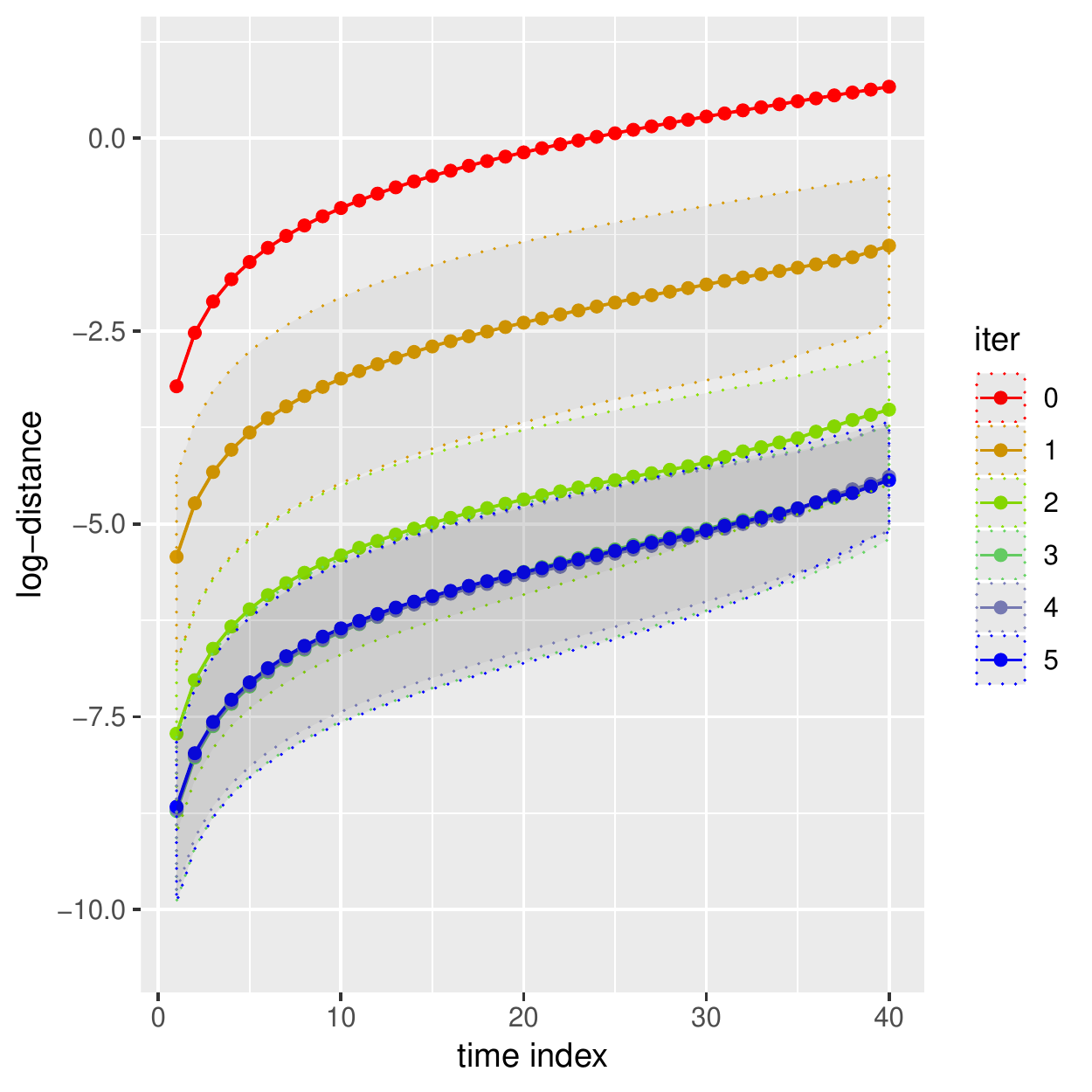}
            \caption{{\small Approx.~IPF, $M=50$.}}
            \label{fig:langevin_IPF_approx_IPF5_reps100_n1000_m50}
        \end{subfigure}
        \begin{subfigure}[t]{0.32\textwidth}
            \centering
            \includegraphics[width=\textwidth]{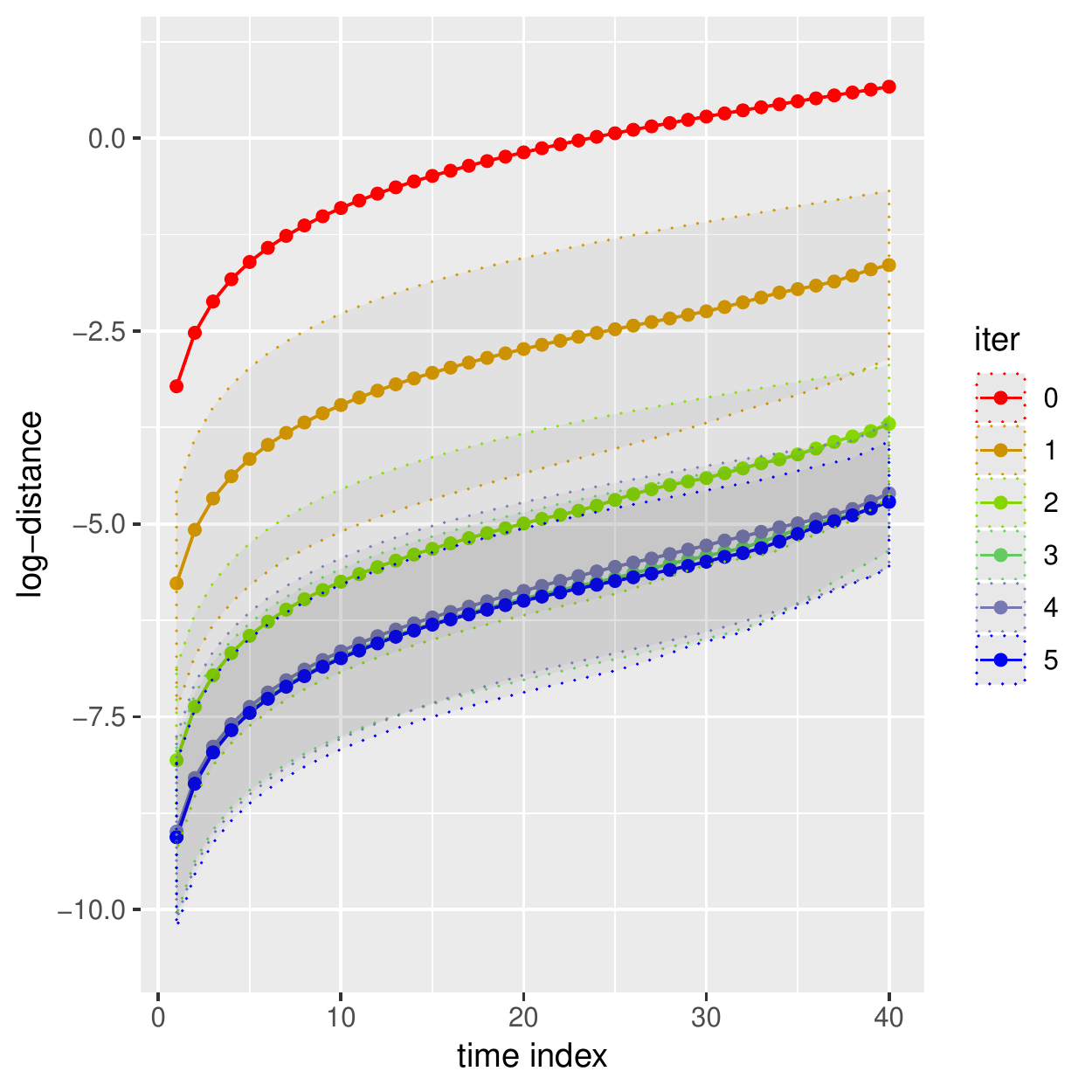}
            \caption{{\small Approx.~IPF, $M=100$.}}
            \label{fig:langevin_IPF_approx_IPF5_reps100_n1000_m100}
        \end{subfigure}
        \caption{\small Distances between marginals of the Schr\"odinger bridge $s_t$ and the marginals of the IPF iterates $q_t^{(i)}$, measured as $\log \was_2(s_t,q_t^{(i)})$, for the LQG setting of Section \ref{sec:lqg} with discretized Langevin diffusion reference dynamics. Figure \ref{fig:langevin_IPF_exact} corresponds to the exact computation of the IPF iterates for $i \in [0:5]$. The remaining plots correspond to the proposed particle-based approximation of the IPF iterations using $N =1,000$ particles and different values of conditional SMC iterations $M$  (using $P = 128$ CSMC particles). The solid lines correspond to the median value of the log-distance calculated over $100$ independent simulations, and corresponding confidence bands represent the $5\%$ and $95\%$ quantiles. Note that the vertical axis of Figure \ref{fig:langevin_IPF_exact} is on a different scale than those of the other figures, as the exact IPF iterations yield smaller distances in general.}
        \label{fig:lqg_langevin_IPF}
\end{figure}

\section{Sequential Schr\"odinger bridge samplers}\label{sec:ssb_samplers}
The Monte Carlo approximation of the Schr\"odinger bridge proposed in Section \ref{sec:sb} has several limitations.
In particular, choosing backward kernels to achieve low variance Radon--Nikodym derivative estimation is difficult, in part because the optimal choice depends on the intractable marginal distributions of the forward process.
In this section, we construct forward kernels such that the corresponding marginal distributions are approximately equal to the sequence $\{\pi_t\}_{t\in[0:T]}$, and therefore circumvent this issue. The scheme we introduce in Section \ref{sec:ssb} is based on estimating and composing a sequence of  intermediate Schr\"odinger bridges, which will be seen to approximate the solution of the \textit{multi-marginal} Schr\"odinger bridge problem discussed in Section \ref{sec:multi_sb}.

\subsection{Multi-marginal Schr\"odinger bridges}\label{sec:multi_sb}
The multi-marginal Schr\"odinger bridge problem is defined by
\begin{equation} \label{eq:multi_sb}
\mathbb{S}_\mathcal{K}(\mathrm{d}x_{0:T}) = \argmin_{\mathbb{H} \in \mathcal{P}_{\mathcal{K}}(\pi_\mathcal{K})} \mathrm{KL}(\mathbb{H} | \mathbb{Q}),
\end{equation}
where $\mathcal{K} = \{t_k\}_{k\in[1:K]}\subset[0:T]$ with $t_1 = 0$, $t_K = T$ and $t_k < t_{k+1}$ for each $k$, and $\mathcal{P}_{\mathcal{K}}(\pi_\mathcal{K}) = \cap_{k\in[1:K]}\mathcal{P}_{t_k}(\pi_{t_k})$. In other words, the set of admissible path measures that have $\pi_{t_k}$ as their $t_k$-marginal for every $k\in[1:K]$. An important special case is where $\mathcal{K} = [0:T]$. Similar to the two-marginal Schr\"odinger bridge considered earlier (in which $\mathcal{K}=\{0,T\}$), the multi-marginal Schr\"odinger bridge can be written as
\begin{equation} \label{eq:multi_sb_potentials}
\mathbb{S}_\mathcal{K}(\mathrm{d}x_{0:T}) =  \prod_{k=1}^K\varphi_{t_k}(x_{t_k})\mathbb{Q}(\mathrm{d}x_{0:T}),
\end{equation}
where the potentials $\{\varphi_{t_k}\}_{k\in[1:K]}$ are unique up to a multiplicative constant and solve the $K$ Schr\"odinger equations: for each $k \in[1:K]$,
\begin{equation}\label{eq:schrodinger_equations}
\pi_{t_k}(x_{t_k}) = \int_{\mathsf{E}^T} \left[\prod_{\ell=1}^K\varphi_{t_\ell}(x_{t_\ell})  \mathbb{Q}(x_{0:T})\right]\mathrm{d}x_{-t_k},
\end{equation}
where we have used the notation $x_{-t} = (x_0,\dots,x_{t-1},x_{t+1},\dots,x_T)$.

Furthermore, the solution of the multi-marginal problem can be approximated using a generalization of the IPF algorithm;
this scheme reduces to the iterations in \eqref{eq:IPF} for $\mathcal{K}=\{0,T\}$.
In particular, \citet{kullback1968probability} introduced a scheme which systematically cycles through KL projections onto $\mathcal{P}_{t_k}(\pi_{t_k})$ for each $k\in[1:K]$. In the following, we will not consider \citet{kullback1968probability}'s scheme as it is challenging to
implement for several reasons. Firstly, this scheme would suffer from the same difficulties as discussed for the two-marginal problem, as we would still be required to construct a full set of $T$ backward kernels for each iteration of IPF.
Secondly, each iteration of this algorithm has computational complexity of $\mathcal{O}(KT)$. In Section \ref{sec:ssb}, we will introduce a different scheme
that has reduced complexity, at the cost of potentially allowing errors to accumulate across time.

\subsection{Sequential Schr\"odinger bridges} \label{sec:ssb}
Instead of solving the multi-marginal Schr\"odinger bridge problem using the cyclic IPF scheme considered by \citet{kullback1968probability}, we instead introduce a sequential approach. In particular, we sequentially solve the intermediate two-marginal Schr\"odinger bridge problems
\begin{equation}
\mathbb{S}_k(\mathrm{d}x_{t_{k}:t_{k+1}}) = \argmin_{\mathbb{H} \in \mathcal{P}_{t_{k},t_{k+1}}(\pi_{t_{k},t_{k+1}})} \mathrm{KL}(\mathbb{H} | \mathbb{Q}_{t_{k}:t_{k+1}}), \quad {k \in [1:K-1]},
\end{equation}
where for each $0\leq s <u  \leq T$, we define $\mathbb{Q}_{s:u} = \pi_{s}(\mathrm{d}x_s) \prod_{t = s+1}^u M_t(x_{t-1},\mathrm{d}x_t).$ As shown in Section \ref{sec:sb}, we know that for each $k\in[1:K-1]$ we can write
\begin{equation}
\mathbb{S}_k(\mathrm{d}x_{t_{k}:t_{k+1}}) = \pi_{t_k}(\mathrm{d}x_{t_k}) \prod_{t = t_k +1}^{t_{k+1}}M^{\psi^\star_k}_t(x_{t-1},\mathrm{d}x_t) 
\end{equation}
where $\psi_k^\star = \{\psi^\star_{k,t}\}_{t\in[t_k: t_{k+1}]}$ denotes the corresponding harmonic functions.

By the Markov property of the initial path measure $\mathbb{Q}$, solving the set of two-marginal Schr\"odinger bridge problems is equivalent to solving the multi-marginal problem. Their solutions can be related explicitly in the following way. Let $\varphi_{t_1}(x_{t_1}) = \left[\psi^\star_{1,t_1}(x_{t_1})\right]^{-1}$ and $\varphi_{t_K}(x_{t_K}) = \psi^\star_{K-1,t_K}(x_{t_K})$, and for any $k\in[2:K-1]$ let
\begin{equation}
\varphi_{t_k}(x_{t_k}) = \frac{\psi^\star_{k-1,t_k}(x_{t_k})}{\psi^\star_{k,t_k}(x_{t_k})}.
\end{equation}
For $\{\varphi_{t_k}\}_{k\in[1:K]}$ defined this way, we have that
\begin{equation}
\mathbb{Q}(\mathrm{d}x_{0:T})\prod_{k=1}^K\varphi_{t_k}(x_{t_k}) =  \pi_0(\mathrm{d}x_0) \prod_{k=1}^{K-1} \prod_{t = t_k +1}^{t_{k+1}}M^{\psi^\star_k}_t(x_{t-1},x_t)\mathrm{d}x_t.
\end{equation}
By construction of the policies $\{\psi_k^\star\}_{k\in[1:K-1]}$ via the two-marginal problems, the $t_k$-marginal of this path measure is equal to $\pi_{t_k}$. Hence, the potentials $\{\varphi_{t_k}\}_{k\in[1:K]}$ solve the Schr\"odinger equations  \eqref{eq:schrodinger_equations}, and we can write the solution to the multi-marginal problem as
\begin{equation}
\mathbb{S}_\mathcal{K}(\mathrm{d}x_{0:T}) =  \pi_0(\mathrm{d}x_0) \prod_{k=1}^{K-1} \prod_{t = t_k +1}^{t_{k+1}}M^{\psi^\star_k}_t(x_{t-1},x_t)\mathrm{d}x_t.
\end{equation}

The reason for introducing the intermediate problems is that approximating the Schr\"odinger bridge between nearby distributions $\pi_{t_k}$ and $\pi_{t_{k+1}}$ with reference process $\mathbb{Q}_{t_k:t_{k+1}}$ is typically easier than estimating the full Schr\"odinger bridge between $\pi_0$ and $\pi_T$ with reference process $\mathbb{Q}$. In particular, one can expect the Radon--Nikodym derivative estimators \eqref{eqn:RN_estimator} to have smaller variance.
However, applying approximate IPF to find $\mathbb{S}_k$ requires being able to initialize particles from $\pi_{t_k}$. This can be done approximately by initializing particles from $\pi_0$ and propagating them through the approximations of the bridges $\mathbb{S}_1,\dots, \mathbb{S}_{k-1}$, suggesting a sequential approach to estimating the multi-marginal bridge $\mathbb{S}_\mathcal{K}$. 

In Section \ref{sec:ssb_langevin}, we approximate the multi-marginal Schr\"odinger bridge in the case where the reference process $\mathbb{Q}$ is the Euler--Maruyama discretization of the Langevin dynamics introduced in Section \ref{sec:sb_langevin}. In Section \ref{sec:ssb_sampling}, we discuss how to use the sequential estimation approach for sampling, and in Section \ref{sec:adaptive_sequence} develop an adaptive construction of the set $\mathcal{K}$.

\subsection{Sequential Schr\"odinger bridges for discretized Langevin dynamics}\label{sec:ssb_langevin}
We illustrate the sequential Schr\"odinger bridge (SSB) approach in the case where the initial forward kernels correspond to the Euler--Maruyama discretization of the continuous-time Langevin dynamics in \eqref{eq:langevin}. The main contrast with the corresponding two-marginal problem discussed in Section \ref{sec:sb_langevin} is that the backward kernels introduced in \eqref{eq:twisted_backward_kernel_approx} are likely to be more efficient within the SSB methodology. To illustrate this point, we consider the setting where $\mathcal{K} = [0:T]$. The resulting multi-marginal problem can be seen as a discretization of the following control problem: find the control $s\mapsto u_s = \nabla \log \psi_s$ such that the process defined in \eqref{eq:langevin_controlled} with $b_s = \frac{1}{2}\nabla \log \pi_s$
satisfies $Y_s\sim\pi_s$ for every $s\in[0,\tau]$, and $s\mapsto \nabla \log \psi_s$ minimizes the cost function $\mathbb{E}\int_0^\tau \|u_s(Y_s)\|^2 \mathrm{d}s$. This problem arises in different literatures, e.g.~in physics, where any feasible potential is said to yield a \textit{shortcut to adiabaticity} \citep{patra2017shortcuts}. We elaborate on these connections in Section \ref{sec:connections}.

As noted in Section \ref{sec:sb_langevin}, the time-reversed version of the controlled continuous-time process \eqref{eq:langevin_controlled} can be expressed as \eqref{eq:langevin_controlled_reversal}. However, in the setting we consider here, the marginal distribution $\rho_s$ of the forward process is no longer intractable, as it is forced to equal $\pi_s$. Hence, the backward kernels \eqref{eq:twisted_backward_kernel_approx} that arise as the Euler--Maruyama discretization of the time-reversed process are likely to become increasingly efficient as our policy approximation improves, provided the discretization of $[0,\tau]$ is fine enough.

\subsection{Sequential Schr\"odinger bridge sampling} \label{sec:ssb_sampling}
Given the policy $\hat{\psi} = \{\hat{\psi}_k\}_{k\in[1:K-1]}$ produced by the sequential application of approximate IPF, one can proceed to use the proposal $\mathbb{Q}^{\hat{\psi}}$ within importance sampling or SMC on path space. It should be noted that the calculation of the incremental importance weights (and, if applicable, the corresponding resampling step) can be integrated into the forward IPF sweep. Similar to other SMC algorithms, one can also incorporate rejuvenation steps into the SSB algorithm. In particular, for each iteration of IPF targeting $\mathbb{S}_k$ one can move the particles $\{X_{t_k}^n\}_{n\in[1:N]}$ approximating $\pi_{t_k}$ using Markov kernel that is invariant to $\pi_{t_k}$. This serves two purposes: first, these moves can improve the particle approximation of $\pi_{t_k}$, and second,  ``refreshing" the particles can prevent overfitting the estimated policies to the fixed set of samples $\{X_{t_k}^n\}_{n\in[1:N]}$.

In Algorithm \ref{algorithm:ssb_sampler}, we present a version of the SSB sampler without resampling and rejuvenation and for a fixed set $\mathcal{K}$. Adaptive constructions of $\mathcal{K}$ will considered in Section \ref{sec:adaptive}, and versions using resampling and rejuvenation can be useful in practice. The potential benefit of refreshing the particles is illustrated in
Section \ref{sec:lqg_highdim}. The computational complexity of Algorithm \ref{algorithm:ssb_sampler} is $\mathcal{O}(N\sum_{k=1}^{K-1} I_k(t_{k+1}-t_k))$, which simplifies to $\mathcal{O}(NIT)$ in the case where the number of iterations $I_k = I$ for all $k$.
In Section \ref{sec:adaptive_IPF}, we discuss methods for warm starting the IPF algorithm and adapting the number iterations to the difficulty of the policy estimation problems, to further reduce the computational cost. Given the output of Algorithm \ref{algorithm:ssb_sampler}, we can approximate expectation of $\varphi:\mathsf{E}\rightarrow\mathbb{R}$ under the target distribution $\pi$ and its normalizing constant $Z$ using the estimators
\begin{align}
\hat{\pi}^N(\varphi)=\sum_{n=1}^N\varphi(X_{T}^n),\quad \hat{Z}^N = \frac{1}{N}\sum_{n=1}^N w_{0:T}^n.
\end{align}
Although one can expect these estimators to be consistent as $N\rightarrow\infty$ \citep{beskos2016convergence}, we note that $\hat{Z}^N$ is not unbiased as a result of adaptation via IPF. To obtain an unbiased normalizing constant estimator, one can simply re-run importance sampling or SMC using the estimated
policy $\hat{\psi}$.

\begin{algorithm}
\caption{\label{algorithm:ssb_sampler} Sequential Schr\"odinger bridge sampler (without resampling or rejuvenation)}
\textbf{Input:} Initial kernels $\{M_t\}_{t\in[1:T]}$, function classes $\{\mathsf{F}_t\}_{t\in[0:T]}$, number of particles $N\in\mathbb{N}$, set of indices $\{0,T\}\subset \mathcal{K} \subset [0:T]$ with $K = |\mathcal{K}|$, number of iterations $I_{k}\in\mathbb{N}$ for each $k\in[1:K-1]$.
\begin{enumerate}
\item Initialize: for each $n\in[1:N]$, sample $X_0^n \sim \pi_0$.
\item For $k \in[1:K-1]$,
\begin{enumerate}
\item Initialize: Set $\hat{\psi}^{(0)}_{k,t} = 1$ for $t \in [t_k:t_{k+1}]$.
\item Perform $I_k$ iterations of approximate IPF (Algorithm \ref{algorithm:aIPF}) to estimate $\mathbb{S}_k$, using the samples $\{X_{t_k}^n\}_{n\in[1:N]}$ as approximate draws from $\pi_{t_k}$. Output the policy $\hat{\psi}_{k}$.
\item Propagate the samples $\{X_{t_k}^n\}_{n\in[1:N]}$ using $\{M^{\hat{\psi}_{k}}_t\}_{t\in[(t_k+1):t_{k+1}]}$ to produce the approximate draws $\{X_{t_{k+1}}^n\}_{n\in[1:N]}$ from $\pi_{t_{k+1}}$, and compute the weights
$$w_{t_k:t_{k+1}}^n = \frac{\gamma_{t_{k+1}}(X^n_{t_{k+1}})\prod_{t =t_k+1}^{ t_{k+1}}L_{t-1}^{\hat{\psi}_k}(X_t^n,X_{t-1}^n)}{\gamma_{t_k}(X^n_{t_k})\prod_{t =t_k+1}^{ t_{k+1}}M^{\hat{\psi}_k}_t(X_{t-1}^n,X_t^n)}, \quad n \in [1:N].$$
\end{enumerate}
\item Compute $w_{0:T}^n = \prod_{k=1}^{K-1} w_{t_k:t_{k+1}}^n$ for each $n \in [1:N]$.
\end{enumerate}
\textbf{Output:} Trajectories $\{X_{0:T}^n\}_{n\in[1:N]}$ and importance weights $\{w_{0:T}^n\}_{n\in[1:N]}$.
\end{algorithm}

\subsection{Sequential Schr\"odinger bridge sampling with adaptive IPF}\label{sec:adaptive}
In this section, we discuss different methods to reduce the computational cost of the SSB sampler. The first approach aims to reduce the number of IPF iterations by using warm starts and stopping IPF when the policy approximations appear to have converged. The second approach adaptively constructs the set $\mathcal{K}$ by triggering IPF only when the effective sample size (ESS) of the particle system falls below a given threshold.

\subsubsection{IPF with early stopping and warm starts} \label{sec:adaptive_IPF}
Over the course of the IPF iterations, we can monitor changes in the policy approximations and stop when they appear to have converged. Since the policies are estimated using a finite number of random particles, the measure by which we evaluate convergence needs to be able to account for the resulting noise in the approximations.  Hence, we perform hypothesis tests to check whether the estimated policies have reached their stationary distribution under the approximate IPF scheme.

In practice, we use parametric function classes for policy approximations, which means that monitoring the convergence of the policies can  be reduced to tracking the evolution of the corresponding parameters. In the numerical experiments of Sections \ref{sec:ssb_lqg} and \ref{sec:numerical_experiments},  we consider, at IPF iteration $i$, a window of the differences between the parameters at iterations $j$ and $j-1$ for the last $J$ iterations of IPF (i.e.~$j = i - J, \dots, j = i$). For each parameter, we perform a $t$-test of whether the mean of these differences are equal to zero. The IPF iterations are stopped when none of the tests are significant, controlling for false discoveries using e.g.~the Benjamini-Hochberg procedure \citep{benjamini1995controlling}, or when a prescribed number of iterations is reached. If the IPF iterations are stopped early, we can reduce the variance in the policy approximations by averaging each parameter over the window for which the test was performed.

In the important special case where $\mathcal{K} = [0:T]$, in which we approximate the solution of the two-marginal Schr\"odinger bridge between $\pi_{t-1}$ and $\pi_t$ for each $t\in[1:T]$, we can also warm start the IPF iterations. Since in the analogous continuous-time problem we expect the curve of policies $(\psi^\star_s)_{s\in[0,\tau]}$ to be smooth, we can also expect that an extrapolation of the approximations of $\{\psi^\star_r\}_{r\in[1:t]}$ could provide a good starting point for the approximation of $\psi^\star_{t+1}$, provided the discretization of $[0,\tau]$ is fine enough. In practice, we often use an expression of the form $\hat{\psi}_t + (\hat{\psi}_t - \hat{\psi}_{t-1})$, or simply $\hat{\psi}_t$ itself. In Section \ref{sec:ssb_lqg}, we illustrate the benefits of warm starts, and how the combination of warm starts and early stopping can yield large reductions in computational cost without increasing errors.

\subsubsection{Adaptive construction of $\mathcal{K}$} \label{sec:adaptive_sequence}
Instead of a priori defining the set $\mathcal{K}$, we can sequentially add time indices $t_{k+1}$ to it by triggering IPF steps only when the effective sample size (ESS) of the particle system falls below a given threshold $0 \leq e_k \leq N$. That is, in the version of the SSB sampler without resampling and starting from time step $t_k$, we propagate the particles using the initial Markov
kernels $\{M_t\}$ and define $t_{k+1}$ to be the first $t>t_k$ for which $\left(\sum_{n=1}^N (W_{1:t}^n)^2\right)^{-1} < e_k$, where
$$W_{1:t}^n = \frac{w_{1:t}^n}{\sum_{n=1}^N w_{1:t}^n}, \quad w_{1:t}^n = \prod_{s=1}^t w_s^n.$$
In the version with resampling after the final step of IPF, $W_{1:t}^n$ and  $w_{1:t}^n$ in the above expressions should be substituted with $W_{t_k:t}^n$ and  $w_{t_k:t}^n$.
In turn, we can perform IPF iterations approximating the bridge over $[t_k:t_{k+1}]$ until the ESS of the associated particle system increases to a prescribed threshold or until the ESS appears to stop increasing, as an alternative to the early-stopping approach described in the preceding section.

\subsection{Example: linear quadratic Gaussian (continued)} \label{sec:ssb_lqg}
Recall the linear quadratic Gaussian setting described in Section \ref{sec:lqg}. Here, we estimate the solution to the corresponding multi-marginal Schr\"odinger bridge problem with discretized Langevin dynamics reference process and $\mathcal{K} = [0:T]$. In Figures \ref{fig:langevin_ssb_IPF_exact} and \ref{fig:langevin_ssb_IPF_approx_maxIPF100_reps100_n1000_adaptFALSE_warmstartnone} we plot $\log \was_2(\pi_t, q_t^{(I)})$ for $t \in[1:T]$ and various values of $I$ between $0$ and $100$, calculated using exact and approximate IPF with $M=0$ iterations of conditional SMC. As for the two-marginal Schr\"odinger bridge problem, the marginals constructed in the exact IPF iterations appear to converge exponentially fast, albeit at a slower rate. On the other hand, the approximate scheme appears to behave more similarly to the exact algorithm than in the two-marginal problem. This is likely because our construction of the backward kernels are better adapted to the multi-marginal problem.

Figure \ref{fig:langevin_ssb_IPF_approx_adapt_maxIPF100_reps100_n1000} plots $\log \was_2(\pi_t, q_t^{(I)})$ for marginals $q_t^{(I)}$ obtained with and without the warm starts and early stopping schemes discussed in Section \ref{sec:adaptive_IPF}, with the non-adaptive algorithm using $I = 100$ IPF iterations, and the adaptive one performing at least 3 iterations for each value of $t\in[1:T]$ with a maximum of $I=100$. To monitor the convergence of the policies, we use a window of size $J = \min\{15,i\}$ at iteration $i \geq 3$. The plot illustrates the benefit of warm starts, and that the combination of warm starts and early stopping appears to yield at least as good approximations as the non-adaptive scheme. The reduction in computational cost is dramatic, and the number of IPF iterations performed at each $t\in[1:T]$ using the early stopping criterion is illustrated in Figure \ref{fig:langevin_ssb_IPFnumber_adapt_maxIPF100_reps100_n1000}. In terms of wall-clock time, the non-adaptive algorithm took on average $8.68 s$, whereas the adaptive one took $0.906 s$ on an Intel Core i5 (2.5 GHz).

To assess the SSB algorithm from the sampling perspective, we compute the estimates $\{\hat{Z}_t\}_{t\in[1:T]}$ of the normalizing constants $\{Z_t\}_{t\in[1:T]}$ corresponding to $\{\pi_t\}_{t\in[0:T]}$. In Figure \ref{fig:lqg_langevin_ssb_normconst}, we compare an SMC sampler with the same discretized Langevin dynamics Markov kernels against the SSB sampler using warm starts and early stopping.
In Figure \ref{fig:langevin_ssb_RMSE_lognormconst_adapt_maxIPF100_reps100_n1000.pdf}, we plot the root mean squared error (RMSE) of $\log \hat{Z}_t$, for the two different methods.
At time $t=T$, the RMSE of the estimator obtained with standard SMC was $86$ times higher than the estimator obtained with the SSB sampler, but took only $0.122s$ to compute on average. In other words, the RMSE of $\log \hat{Z}_T$ was reduced by a factor of $86$ at the cost of $7.4$ times greater running time. In comparison, increasing the number of particles (and hence running time) of the standard SMC algorithm by a factor of $7.4$ is only expected to reduce the RMSE by a factor of $\sqrt{7.4}$.

\begin{figure}[hp]
\centering
        \begin{subfigure}[t]{0.4\textwidth}
            \centering
            \includegraphics[width=\textwidth]{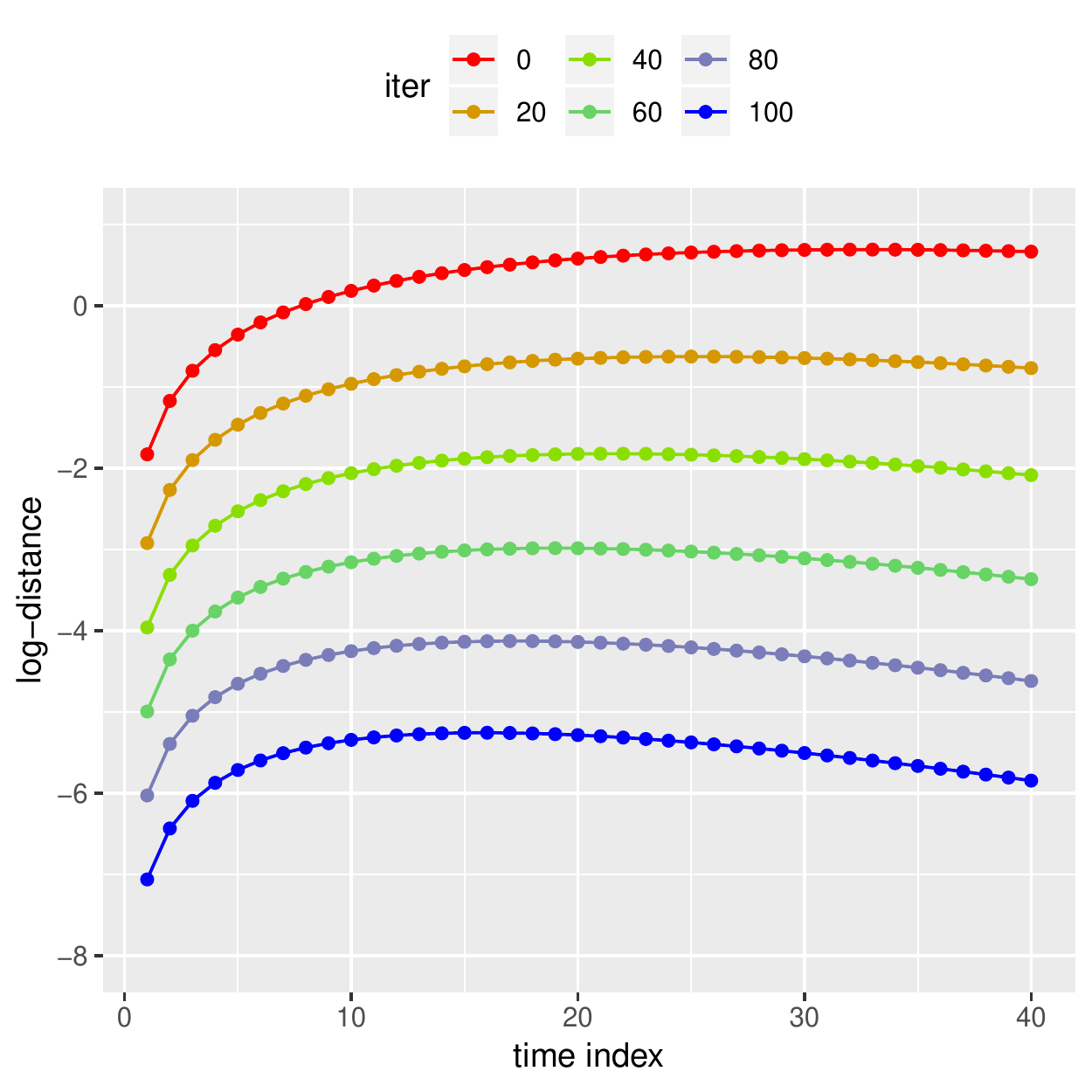}
            \caption{{\small Exact IPF.}}
            \label{fig:langevin_ssb_IPF_exact}
        \end{subfigure}
        \hspace*{1cm}
        \begin{subfigure}[t]{0.4\textwidth}
            \centering
            \includegraphics[width=\textwidth]{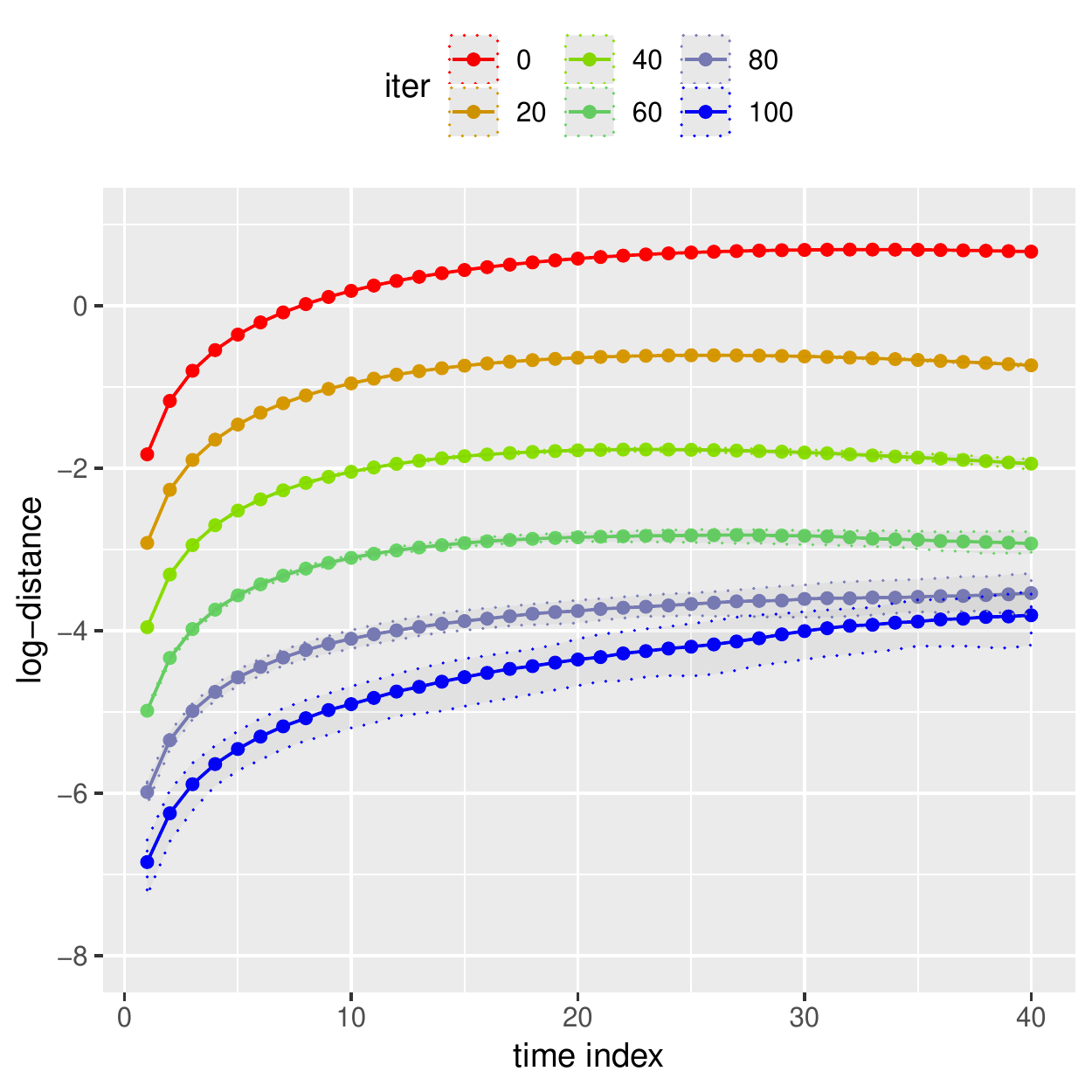}
            \caption{{\small Approximate IPF.}}
            \label{fig:langevin_ssb_IPF_approx_maxIPF100_reps100_n1000_adaptFALSE_warmstartnone}
        \end{subfigure}

         \begin{subfigure}[t]{0.4\textwidth}
            \centering
            \includegraphics[width=\textwidth]{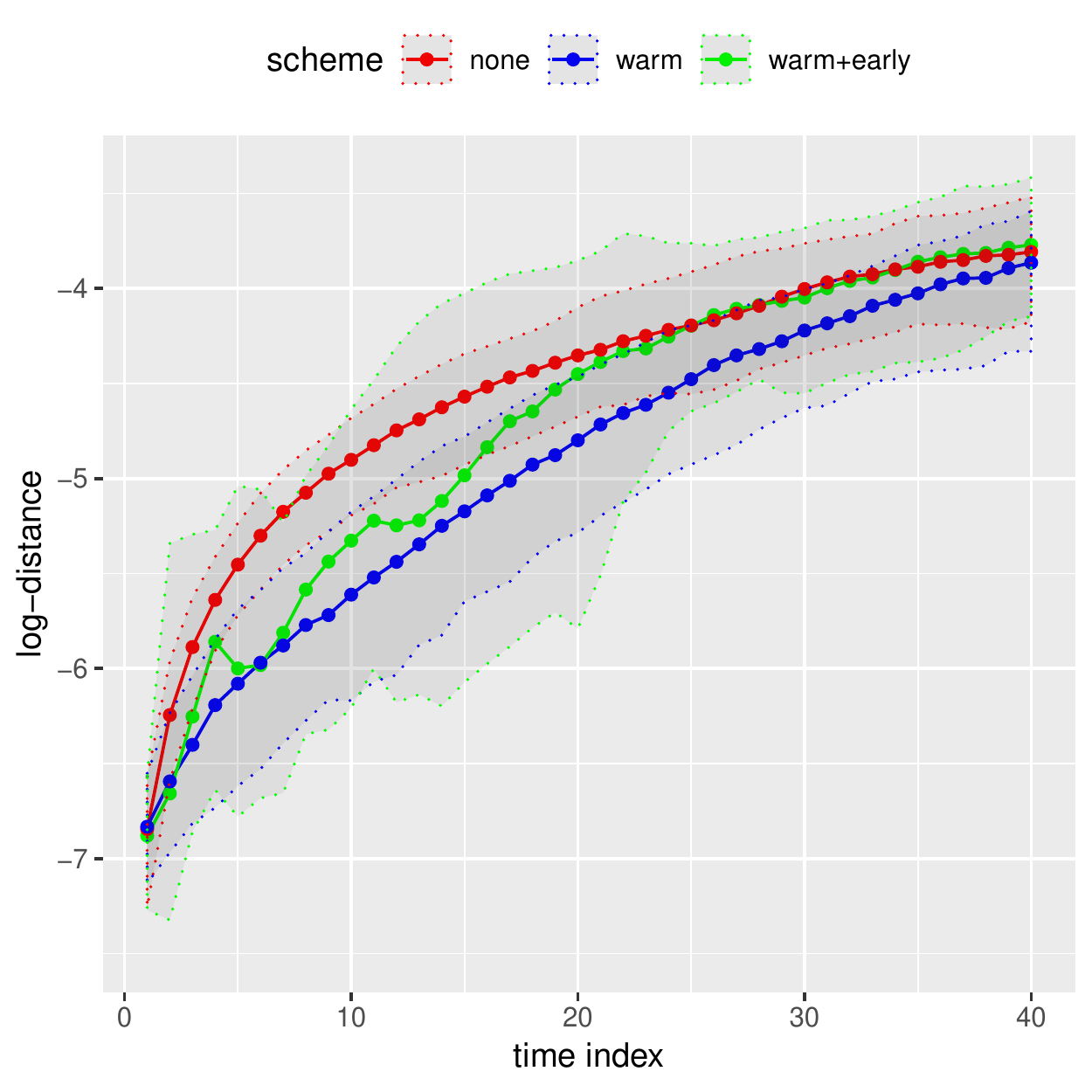}
            \caption{{\small Comparison between approximate IPF schemes with warm starts and early stopping.}}
            \label{fig:langevin_ssb_IPF_approx_adapt_maxIPF100_reps100_n1000}
        \end{subfigure}
        \hspace*{1cm}
        \begin{subfigure}[t]{0.4\textwidth}
            \centering
            \includegraphics[width=\textwidth]{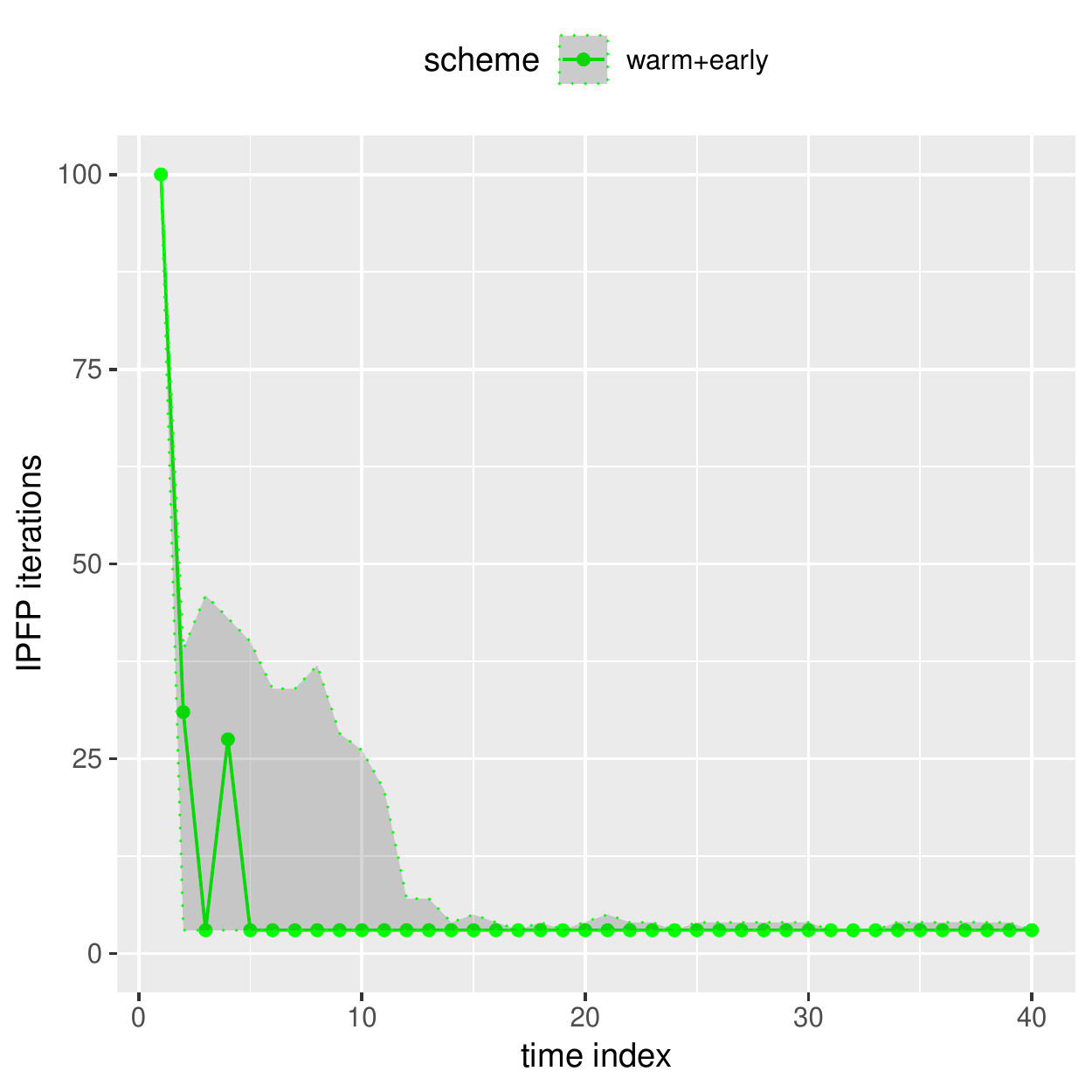}
            \caption{{\small Number of iterations in approximate IPF scheme using warm starts and early stopping.}}
            \label{fig:langevin_ssb_IPFnumber_adapt_maxIPF100_reps100_n1000}
        \end{subfigure}

        \caption{ {\small Distances between $\pi_t$ and marginals of the IPF-based approximations $q_t^{(i)}$ of the multi-marginal Schr\"odinger bridge, measured as $\log \was_2(\pi_t,q_t^{(i)})$, for the LQG setting of Section \ref{sec:ssb_lqg} with discretized Langevin diffusion reference dynamics. Figure \ref{fig:langevin_ssb_IPF_exact} correspond to the exact computation of the IPF iterates for different values of $i$ between $0$ and $100$. Figure \ref{fig:langevin_ssb_IPF_approx_maxIPF100_reps100_n1000_adaptFALSE_warmstartnone} corresponds to the proposed particle-based approximation using $N =1,000$ particles and $M=0$ iterations of CSMC. Figure \ref{fig:langevin_ssb_IPF_approx_adapt_maxIPF100_reps100_n1000} illustrates the differences between the approximations obtained with and without the warm start and early stopping schemes discussed in Section \ref{sec:adaptive_IPF}. Figure \ref{fig:langevin_ssb_IPFnumber_adapt_maxIPF100_reps100_n1000} shows the number of IPF iterations performed in the early stopping scheme at each time $t\in[1:T]$, with a lower bound set to 3 iterations.
        In all the plots, solid lines correspond to median values calculated over $100$ independent simulations, and the corresponding confidence bands represent the $5\%$ and $95\%$ quantiles.}
        \label{fig:ssb_lqg_langevin_IPF}}
\end{figure}

\begin{figure}[hp]
\centering
        \begin{subfigure}[t]{0.4\textwidth}
            \centering
            \includegraphics[width=\textwidth]{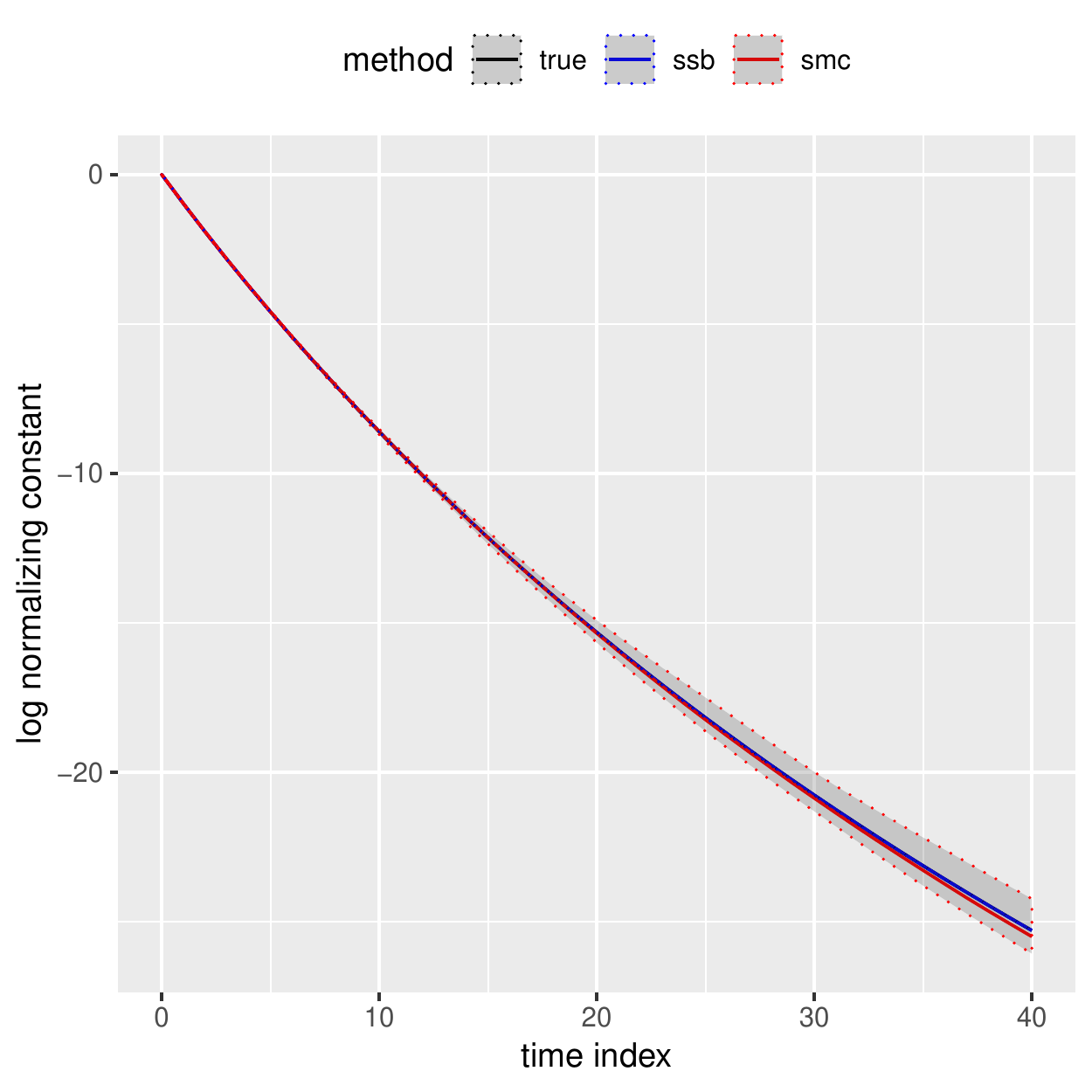}
            \caption{{\small Values of  $\log \hat{Z}_t$.}}
            \label{fig:langevin_ssb_normconst_adapt_maxIPF100_reps100_n1000}
        \end{subfigure}
        \hspace*{1cm}
        \begin{subfigure}[t]{0.4\textwidth}
            \centering
            \includegraphics[width=\textwidth]{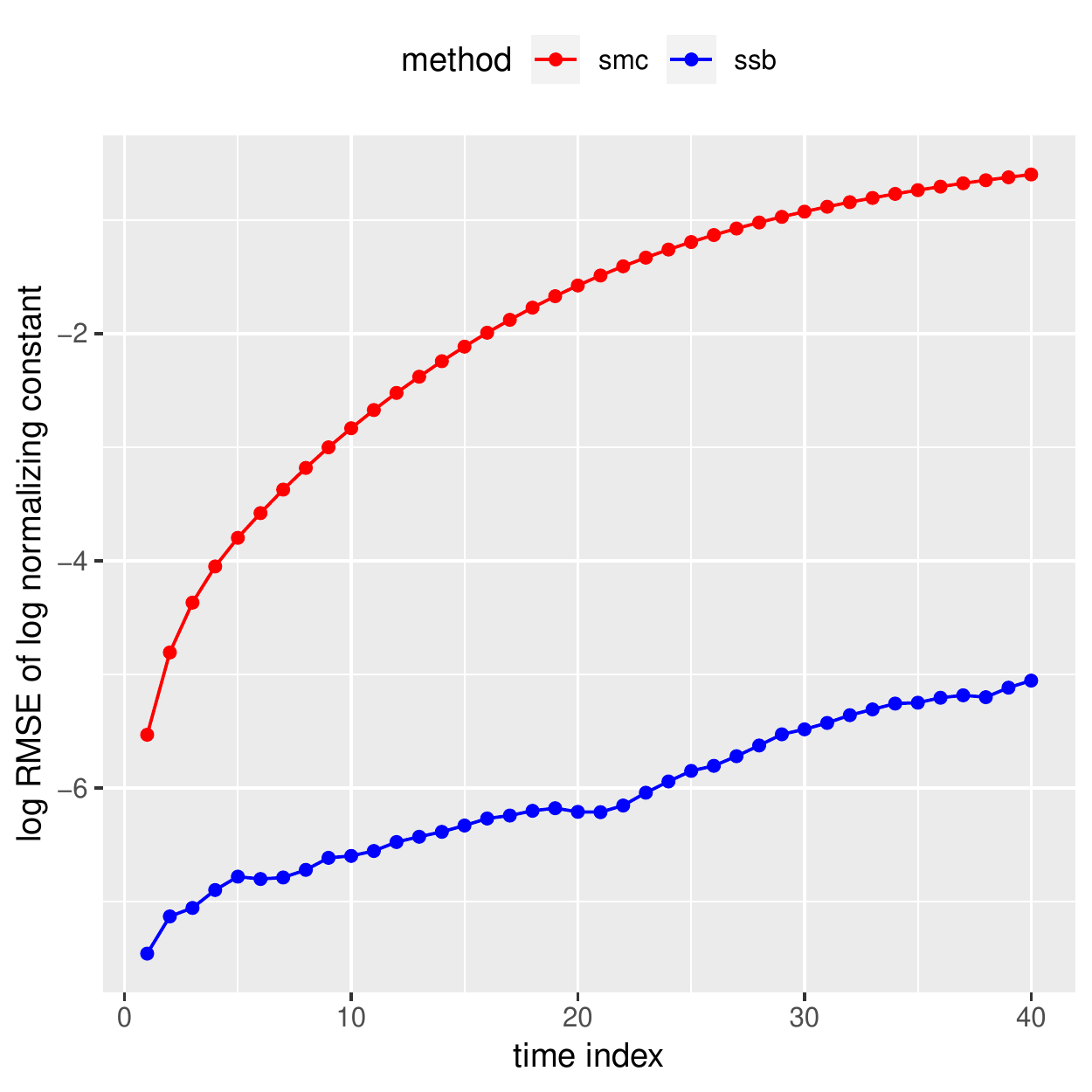}
            \caption{{\small Log RMSE of $\log \hat{Z}_t$.}}
            \label{fig:langevin_ssb_RMSE_lognormconst_adapt_maxIPF100_reps100_n1000.pdf}
        \end{subfigure}
  \caption{ {\small On the left, we plot the logarithm of the estimates $\{\hat{Z}_t\}_{t\in[1:T]}$  of the normalizing constants $\{Z_t\}_{t\in[1:T]}$ corresponding to $\{\pi_t\}_{t\in[1:T]}$, based on an SMC sampler using Markov kernels derived from discretizing Langevin diffusion (in red) and the SSB sampler using the same discretized Langevin Markov kernels as the reference process (in blue). The solid lines correspond to median values calculated over $100$ independent simulations using $N =1,000$ particles and $M=0$ iterations of CSMC, and the corresponding confidence bands represent the $5\%$ and $95\%$ quantiles. Note that the true log normalizing constants (in black) are obscured by the blue line, and that the blue confidence bands are too narrow to be visible. To further illustrate the behavior of the different estimators, we compute the root mean squared error of $\log \hat{Z}_t$, the log of which is shown on the right.}   \label{fig:lqg_langevin_ssb_normconst}}
\end{figure}

\section{Connections with other problems} \label{sec:connections}
In Sections \ref{sec:sb} and \ref{sec:ssb}, we have discussed the two- and multi-marginal Schr\"odinger bridge problems in their minimum KL and stochastic control formulations. Here, we elaborate on a few other perspectives that arise in different literatures. In particular, we review the two-marginal Schr\"odinger bridge problem as a regularization of an optimal transport problem, and discuss how our method might be used to approximate the 2-Wasserstein distance between $\pi_0$ and $\pi_T$. Similarly, by analogy with the continuous-time formulation, the multi-marginal Schr\"odinger bridge can be viewed as an approximation of a solution of the \textit{flow transport problem}. In the physics literature, both the two- and multi-marginal problems with Langevin diffusion reference dynamics can be said to yield \textit{shortcuts to adiabaticity}. Recently, \citet{chen2019multi} also develop connections between the multi-marginal problem and measure-valued splines. We also briefly discuss Schr\"odinger bridge-based particle filtering and some difficulties associated with applying the methodology in this setting.

\subsection{Optimal transport}
The first formal connections between the Schr\"odinger bridge problem and optimal transport were developed by \citet{mikami2004monge}, who considered the control problem \eqref{eq:langevin_controlled} for reference dynamics of the form $\mathrm{d}X_s = \sigma\mathrm{d}W_s$ on $[0,1]$. If $S_{0,1}^\sigma$ denotes the joint distribution of the optimally controlled process at times $s=0$ and $s=1$, he showed that as $\sigma \to 0$, $S_{0,1}^\sigma$ converges to a minimizer of the following Monge-Kantorovich optimal transport problem:
\begin{equation}
\was_2^2(\pi_0,\pi_T) =  \min_{h_{0,T} \in \mathcal{C}(\pi_0,\pi_T)} \int_{\mathsf{E}\times\mathsf{E}}  \| x_0- x_T\|^2 h_{0,T}(\mathrm{d}x_0,\mathrm{d}x_T),
\end{equation}
where the notation $\was_2(\pi_0,\pi_T)$ reflects that this defines the 2-Wasserstein distance. Additionally, when the reference path measure $Q^\sigma$ is given by the $\sigma$-scaled reversible Brownian motion, it is shown that $\sigma^2\mathrm{KL}(S^\sigma_{0,1} | Q^\sigma_{0,1}) \to \was_2^2(\pi_0,\pi_T)$.
For more general Markov reference processes, the objective function of the Schr\"odinger bridge problem would converge to an optimal transport objective, with a cost function $c(x_0,x_T)$, defined by the rate function of an associated large deviations principle result \citep{leonard2012schrodinger}.

A useful perspective for making these connections is a formulation of the Schr\"odinger bridge problem developed by \citet{leonard2014survey}, derived by considering the Fokker-Planck equation associated with \eqref{eq:langevin_controlled}; see also \citet{chen2017optimal}. In particular, the problem studied by \citet{mikami2004monge} can be expressed as
\begin{align}
& \min_{\rho, \psi} \int_0^\tau\int_\mathsf{E} \|\nabla \log \psi_s(x)\|^2\rho_s(\mathrm{d}x)\mathrm{d}s, \label{eq:benamou_objective}\\
& \frac{\mathrm{d}\rho_s}{\mathrm{d}s} = -\text{div}(\rho_s \nabla \log \psi_s ) + \frac{\sigma^2}{2}\Delta \rho_s, \label{eq:benamou_FP}\\
& \rho_0(\mathrm{d}x) = \pi_0(\mathrm{d}x), \quad \rho_\tau(\mathrm{d}x) = \pi_T(\mathrm{d}x). \label{eq:benamou_constraints}
\end{align}
When $\sigma = 0$, in which \eqref{eq:benamou_FP} becomes the \textit{continuity equation} \citep[][p.169]{ambrosio2005}, the above problem reduces to the Benamou-Brenier fluid mechanics formulation of the quadratic optimal transport problem \citep{benamou2000computational}.

Recently, these connections have been utilized to create fast approximate solvers of optimal transport problems with general cost functions. In particular, we have seen that the original Schr\"odinger bridge problem \eqref{eq:sb} can be reduced to the static problem \eqref{eq:static_sb}.
If we can write $q^{\sigma}_{0,T}(\mathrm{d}x_0,\mathrm{d}x_T) = Z_{\sigma}^{-1}\exp\{-c(x_0,x_T)/\sigma^2\}\lambda(\mathrm{d}x_0,\mathrm{d}x_T)$ for some cost function $c: \mathsf{E}\times\mathsf{E} \to \mathbb{R}_+$, a $\sigma$-finite measure $\lambda$ on $\mathsf{E}\times\mathsf{E}$ and normalizing constant $Z_{\sigma}$, then\footnote{See e.g.~\citet{leonard2014survey} for a discussion of the definition of the KL divergence with respect to an unbounded measure.}
\begin{align}
\mathrm{KL}(h_{0,T} | q^\sigma_{0,T})  &= \mathrm{KL}(h_{0,T} | \lambda ) + \frac{1}{\sigma^2} \int_{\mathsf{E}\times\mathsf{E}} c(x_0,x_T) h_{0,T}(\mathrm{d}x_0,\mathrm{d}x_T) + \log Z_\sigma.
\end{align}
The corresponding Schr\"odinger bridge can then be written as
\begin{equation} \label{eq:entropic_regularization}
s^\sigma_{0,T}(\mathrm{d}x_{0},\mathrm{d}x_{T}) = \argmin_{h_{0,T} \in \mathcal{C}(\pi_0,\pi_T)}   \int_{\mathsf{E}\times\mathsf{E}} c(x_0,x_T) h_{0,T}(\mathrm{d}x_0,\mathrm{d}x_T)
 +{\sigma^2} \mathrm{KL}(h_{0,T} | \lambda).
\end{equation}
If $\lambda$ is equal to the Lebesgue measure and $c(x,y) = \|x-y\|^2$, the above setting reduces to the one considered by \citet{mikami2004monge}. When $\lambda$ is taken to be $\pi_0\otimes\pi_T$, the minimization in \eqref{eq:entropic_regularization}  is often called the entropically regularized optimal transport problem, and its objective function evaluated at the minimizer $s^\sigma_{0,T}$ called the Sinkhorn divergence \citep{cuturi2013sinkhorn, peyre2018computational}.

The output of Algorithm \ref{algorithm:aIPF} applied to the Schr\"odinger bridge problem with discretized Brownian reference dynamics can be used to approximate the quadratic optimal transport cost in several ways. The simplest is the estimator $N^{-1}\sum_{n=1}^N\|X^n_0 - X_T^n\|^2$, where $\{X_{0:T}^n\}_{n\in[1:N]}$ are the particle trajectories obtained in the last IPF iteration. Since the obtained coupling will in general be sub-optimal for the transport problem, this estimator provides upper bound of $\was_2^2(\pi_0,\pi_T)$ for any non-zero $\sigma$ (up to noise and approximation errors). For the 100 independent simulations performed in the LQG example of Section \ref{sec:lqg} with $\sigma = 1$ and $M = 10$, the average estimated value of $\was_2(\pi_0,\pi_T)$ was $4.27$ with a standard deviation of $0.033$, whereas the exact distance is equal to $4.09$. The discrepancy stems from the large value of $\sigma$.

Alternatively, one can use the approximated Schr\"odinger potentials together with the identity $\mathrm{KL}(s_{0,T} | q_{0,T}) = \int_{\mathsf{E}\times\mathsf{E}} \left\{ \log \varphi^\circ(x_0) + \log \varphi^\star(x_T) \right\} s_{0,T}(\mathrm{d}x_{0},\mathrm{d}x_{T})$ to construct a particle-based estimate. Thirdly, by analogy with the continuous-time problem, one can approximate \eqref{eq:benamou_objective} using the estimated policies and the associated particle trajectories.

\subsection{Flow transport}
Recall the control problem stated in Section \ref{sec:ssb_langevin}, in which we want to find $s\mapsto u_s = \nabla \log \psi_s$ such that the process defined in \eqref{eq:langevin_controlled} with $b_s = \frac{1}{2}\nabla \log \pi_s$ satisfies $Y_s\sim\pi_s$ for every $s\in[0,\tau]$, and $s\mapsto u_s$ minimizes the cost function $\mathbb{E}\int_0^\tau \|u_s(Y_s)\|^2 \mathrm{d}s$. Re-expressing this problem in the language of \eqref{eq:benamou_objective}-\eqref{eq:benamou_constraints}, we can write
\begin{align}
& \min_{\psi} \int_0^\tau \int_\mathsf{E}\|\nabla \log \psi_s(x)\|^2\pi_s(\mathrm{d}x)\mathrm{d}s, \label{eq:flow_objective}\\
& \frac{\mathrm{d}\pi_s}{\mathrm{d}s} = -\text{div}\left(\pi_s \left\{\nabla \log \psi_s + \frac{1}{2}\nabla \log \pi_s\right\}\right) + \frac{1}{2}\Delta \pi_s. \label{eq:flow_FP}
\end{align}
Under invariance of the Langevin dynamics, the Fokker-Planck equation \eqref{eq:flow_FP} reduces to the continuity equation
\begin{equation}\label{eq:continuity_equation}
\frac{\mathrm{d}\pi_s}{\mathrm{d}s} = -\text{div}\left(\pi_s \nabla \log \psi_s \right).
\end{equation}
For any $s \mapsto \nabla \log \psi_s$ that solves \eqref{eq:continuity_equation}, we have that if $\mathrm{d}x_s/\mathrm{d}s = \nabla \log\psi_s(x_s)$ subject to $x_0\sim \pi_0$, then $x_s\sim \pi_s$ for any $s\in[0,\tau]$. Finding such a policy is often called the \textit{flow transport problem} \citep{heng2015gibbs}. Simultaneously solving \eqref{eq:flow_objective} yields the \textit{minimum kinetic energy} solution among all solutions of the flow transport problem, and has been considered by e.g. \citet{reich2011dynamical,reich2012gaussian}.

In the linear quadratic Gaussian case, the minimal kinetic energy solution of the associated continuous time flow transport problem is known exactly \citep{bergemann2012ensemble}. The optimal policy is given by
$$s\mapsto u_s^\star(x_s) = \nabla \log \psi^\star_s(x_s) = -\frac{1}{2\tau}\Sigma_s R^{-1}(x_s + \mu_s - 2y),$$
where $\Sigma_s$ and $\mu_s$ are defined as in the discrete setting of Section \ref{sec:lqg}, but with  $\lambda_s = s/\tau$. As illustrated in Section \ref{sec:ssb_lqg}, the sequential Schr\"odinger bridge algorithm yields policies that in turn induce marginal distributions $q_t^{(I)}$ that are very close to $\pi_t$. Here, we can also numerically approximate the optimal cost and compare it to the cost estimates produced by the SSB sampling algorithm.

Using the same discretization of $[0,\tau]$ as in Section \ref{sec:ssb_lqg}, we numerically solve $\mathrm{d}x_s/\mathrm{d}s = \nabla \log\psi^\star_s(x_s)$ for $100,000$ particles initialized from $\pi_0$. These trajectories were then used to approximate the cost $\mathbb{E}\int_0^\tau \|\nabla \log \psi^\star_s(Y_s)\|^2 \mathrm{d}s$,  which was estimated to be $7.69$. Over the 100 independent runs of the SSB algorithm with $N = 1,000$ particles, the associated cost was estimated to be $8.50$ with a standard deviation of $0.036$. In other words, the policies approximated with the SSB algorithm yield intermediate distributions that are close to the targets, as shown in Section \ref{sec:ssb_lqg}, but appear to have not completely converged to optimality in terms of cost.

\subsection{Shortcuts to adiabaticity}
In the thermodynamics literature, it is well known that it takes an infinitely long time to transition between equilibrium states of a system that stays in equilibrium with the thermal reservoir. This is the physical intuition for why we require $\tau \to \infty$ to satisfy $X_s \sim \pi_s$ for all $s$ in \eqref{eq:langevin_reference} with $b_s = \frac{1}{2}\nabla \log \pi_s$. An interesting question is whether such transitions can be realized in a finite time $\tau$ by allowing the system to not be in equilibrium with the thermal reservoir. A process that achieves such a transition is said to be a \textit{shortcut to adiabaticity}; see e.g. \citet{betancourt2014adiabatic} and \citet{patra2017shortcuts}.

Solutions of both the two-marginal Schr\"odinger bridge problem with Langevin dynamics reference and the flow transport problem yield such shortcuts. Maps $s\mapsto -\log \psi_s$ that are feasible for the flow transport problem are often called \textit{counterdiabatic} potentials, as the system follows the adiabatic evolution $s\mapsto \pi_s$. On the other hand, maps that solve the two-marginal Schr\"odinger bridge problem are often called \textit{fast-forward} potentials, as they allow the intermediate distributions to deviate from $\pi_s$, but return to the adiabatic evolution as $s\to\tau$. In addition to the stochastic setting considered here, there has been growing interest in defining protocols that achieve shortcuts to adiabaticity in both quantum and classical Hamiltonian systems; see e.g. \citet{sels2017minimizing} and the recent survey of \citet{del2019focus}. We hope that the methods developed in this paper can potentially be useful in these fields.

\subsection{Particle filtering}\label{sec:pf}
Consider a hidden Markov chain $\{X_t\}_{t\in[0:T]}$ with distribution
$$\mathbb{Q}(\mathrm{d}x_{0:T}) = \pi_0(\mathrm{d}x_0)\prod_{t=1}^T f_t(x_{t-1},\mathrm{d}x_t),$$
and a sequence of observations $\{Y_{t}\}_{t\in[1:T]}$ assumed to be conditionally independent given $\{X_t\}_{t\in[1:T]}$, and distributed with densities $g_t(X_t,\cdot)$. The goal of particle filtering is to develop online approximations of the sequence of filtering distributions $\pi_t$, defined as the marginal laws of the hidden states $X_t$ given a realization of the observations $y_{1:t}$. The filtering distribution at time $t$ satisfies the recursion
\begin{equation}
\pi_{t}(\mathrm{d}x_t) \propto g_t(x_t,y_t)\int_\mathsf{E} f_t(x_{t-1},\mathrm{d}x_t) \pi_{t-1}(\mathrm{d}x_{t-1}).
\end{equation}

We could envision solving the multi-marginal Schr\"odinger bridge problem associated with the reference measure $\mathbb{Q}$ and marginal constraints $\{\pi_t\}_{t\in[0:T]}$ using the sequential algorithm developed in Section \ref{sec:ssb}. Note that resulting multi-marginal Schr\"odinger bridge would be different from the smoothing distribution, i.e. the law of $X_{0:T}$ given $y_{1:T}$, which is the target of controlled SMC \citep{heng2017controlled} and similar methods by \citet{richard2007efficient}, \citet{scharth2016particle} and \citet{guarniero2017iterated}.
Such an approach would require estimates of pointwise evaluations of Radon--Nikodym derivatives of the form $\mathrm{d}\pi_t/\mathrm{d}q_t^\psi$. This estimation is harder in the filtering setting than in the SMC sampling setting considered earlier, in part due to the intractability of pointwise evaluations of $\pi_t$ and any of its unnormalized counterparts.

If the transition kernels $\{f_t\}_{t\in[1:T]}$ admit densities that can be evaluated, instead of relying on exact evaluations of (unnormalized) filtering densities, we can use sample-based estimates, at the cost of $N^2$ density evaluations per iteration of IPF. Assuming $\{X_{t-1}^{n}\}_{n\in[1:N]}$ are distributed according to $\pi_{t-1}$ and that $X_{t}^{n} \sim f_t^{\psi}(X_{t-1}^{n},\cdot)$ for each $n\in[1:N]$ and some policy $\psi$,
an estimate of $\mathrm{d}\pi_t/\mathrm{d}q_t^\psi$ evaluated at $X_t^n$ is given by
\begin{equation}
\frac{g_t(X_t^n,y_t)\sum_{k = 1}^N f_t(X_{t-1}^k, X_t^n)}{\sum_{k = 1}^N f^{\psi}_t(X_{t-1}^k,X_t^n)}.
\end{equation}

We illustrate this approach on a simple linear quadratic Gaussian model, in which the hidden states $\{X_t\}_{t\in[0:T]}$ arise as a discretization of $\mathrm{d}X_s = AX_s \mathrm{d}s + \mathrm{d}W_s$ over the interval $[0,\tau]$, and where $W_s$ denotes a standard Brownian motion and $\pi_0(\mathrm{d}x_0) = \mathcal{N}(x_0; 0,\mathcal{I}_d)\mathrm{d}x_{0}$. In other words, $f_t(x_{t-1},\mathrm{d}x_t) = \mathcal{N}(x_t; x_{t-1} + hAx_{t-1}, h \mathcal{I}_d)\mathrm{d}x_t$ for a step size $h >0$ such that $hT = \tau$. The matrix $A\in\mathbb{R}^{d\times d}$ is defined by $A_{ij} = A_{ji} = \alpha^{|i-j|+1}$ for $i,j\in[1:d]$. The observation densities are given by $g_t(x_t, \mathrm{d}y_t) = \mathcal{N}(y_t; x_t, \sigma^2 \mathcal{I}_d)\mathrm{d}y_t$. In this setting, the filtering distributions can be calculated exactly using a Kalman filter, to which we compare the Schr\"odinger bridge particle filter discussed above. We also make comparisons with the classical bootstrap particle filter \citep{gordon:salmon:smith:1993}.

In Figure \ref{fig:LQG_sbpf_bpf} we illustrate numerical results in the setting where $d = 2$, $\alpha = 0.1$, $\sigma = 0.1$,  $\tau = 2$ and $T = 80$. In Figure \ref{fig:LQG_sb_kf_sigmaobs0.1_T80_h0.025_IPF_exact}, we use the exact expressions for $\{\pi_t\}_{t\in[1:T]}$ derived with a Kalman filter to perform exact IPF, and plot $\log \was_2(\pi_t,q_t^{(i)})$ for $t \in [1:T]$ and $i \in [1:5]$. For the particle based methods, we took the number of particles to be $N = 200$ for the Schr\"odinger bridge-based method and $N = 150,000$ for the bootstrap particle filter, leading to average running times of $9.09s$ and $9.74s$ respectively over $100$ independent runs on an Intel Core i5 (2.5 GHz). For the Schr\"odinger bridge particle filter, we used IPF with early stopping. In Figure \ref{fig:LQG_sbpf_bpf_logwas_reps100_sigmaobs0.1_T80_h0.025}, we plot the log 2-Wasserstein distance between $\pi_t$ and Gaussian distributions with means and covariance matrices estimated using the particle systems, illustrating that the Schr\"odinger bridge particle filter tends to yield approximations closer than the bootstrap particle filter, with the exception of certain times $t\in[1:T]$. In Figures \ref{fig:LQG_sbpf_bpf_log_normconst_reps100_sigmaobs0.1_T80_h0.025} and \ref{fig:LQG_sbpf_bpf_ESS_reps100_sigmaobs0.1_T80_h0.025} we also plot the normalizing constant (or marginal likelihood) estimates and effective sample sizes as percentage of the sample size, again illustrating improvements of the Schr\"odinger bridge scheme relative to the bootstrap particle filter.

However, we do not anticipate that the Schr\"odinger bridge particle filter using the Radon--Nikodym derivative estimates discussed here will scale well with dimension of the state space. This is because larger dimensions $d$ would require larger sample sizes $N$ to yield estimates with sufficiently low variance, but the cost of increasing $N$ is quadratic. A related $\mathcal{O}(N^2)$ ensemble method using stochastic discretization of the state space in combination with IPF was recently proposed by \citet{reich2018data}. A detailed comparison with our proposed methodology and the application of Schr\"odinger bridges for other state space models could be explored in future work.

\begin{figure}[hp]
\centering
        \begin{subfigure}[t]{0.4\textwidth}
            \centering
            \includegraphics[width=\textwidth]{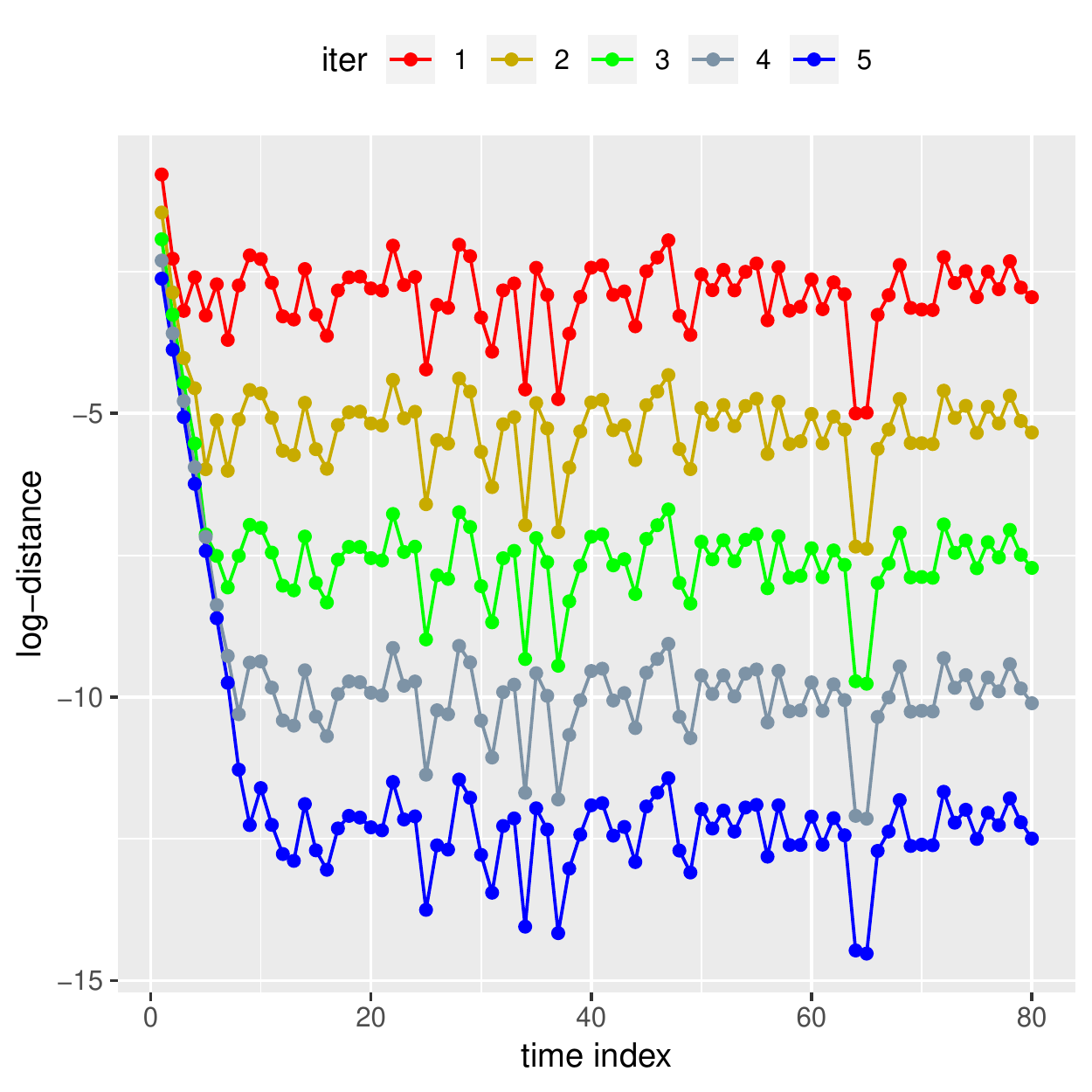}
            \caption{{\small Exact IPF, $\log \was_2(\pi_t,q_t^{(i)})$.}}
            \label{fig:LQG_sb_kf_sigmaobs0.1_T80_h0.025_IPF_exact}
        \end{subfigure}
        \hspace*{1cm}
                \begin{subfigure}[t]{0.4\textwidth}
            \centering
            \includegraphics[width=\textwidth]{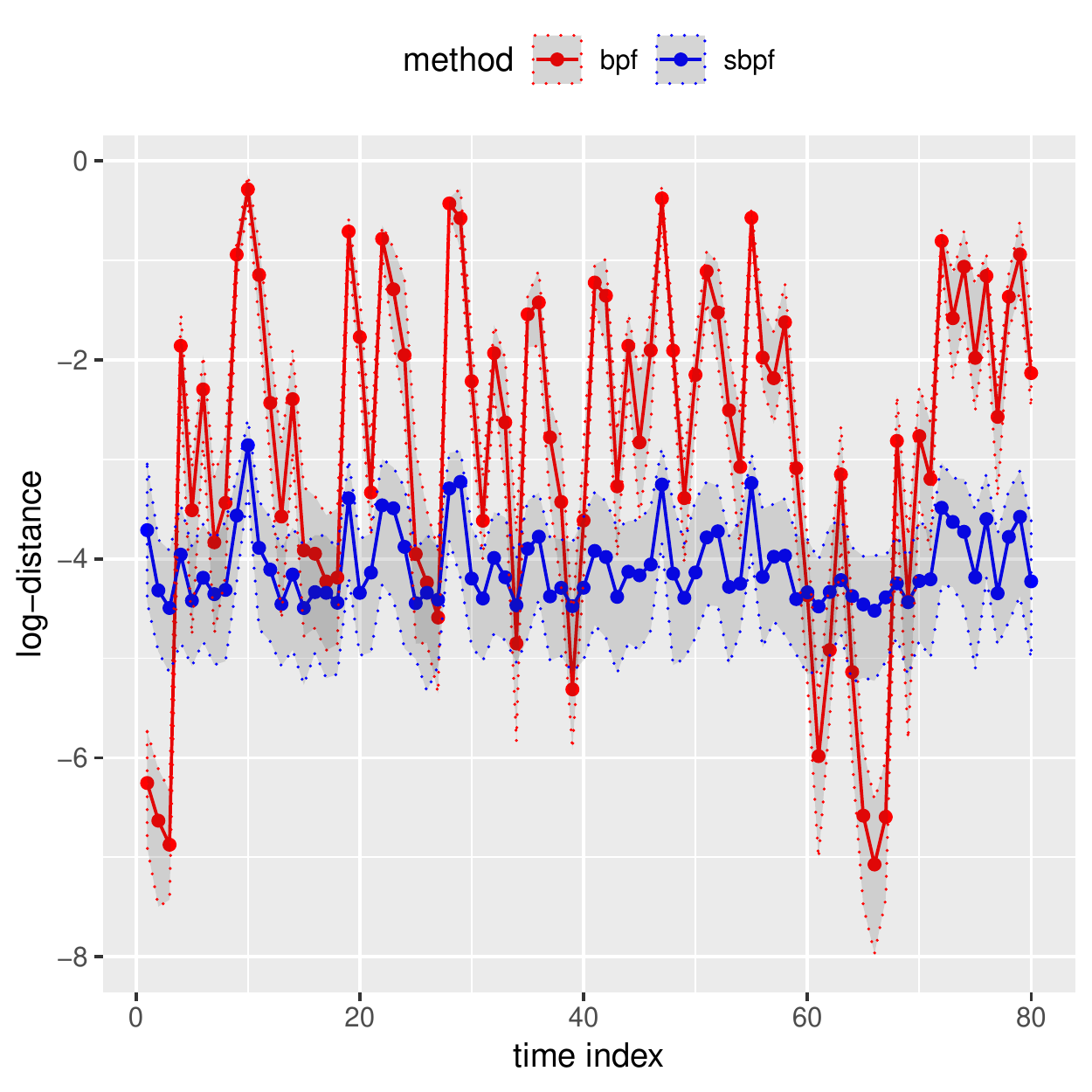}
            \caption{{\small Approx.~IPF, $\log \was_2(\pi_t,\hat{q}_t^{(I)})$.}}
            \label{fig:LQG_sbpf_bpf_logwas_reps100_sigmaobs0.1_T80_h0.025}
        \end{subfigure}

         \begin{subfigure}[t]{0.4\textwidth}
            \centering
            \includegraphics[width=\textwidth]{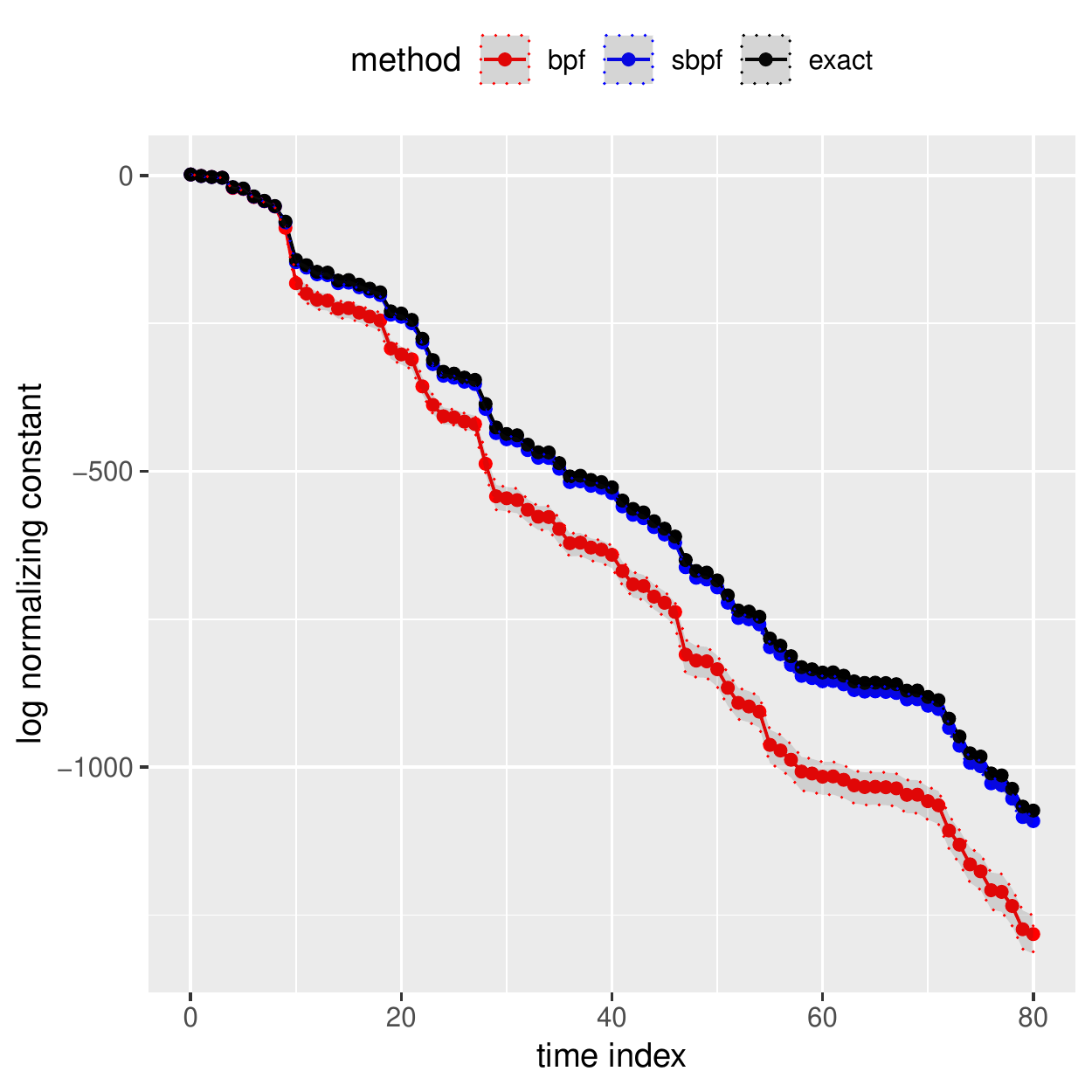}
            \caption{{\small Values of $\log \hat{Z}_t$.}}
            \label{fig:LQG_sbpf_bpf_log_normconst_reps100_sigmaobs0.1_T80_h0.025}
        \end{subfigure}
                \hspace*{1cm}
                \begin{subfigure}[t]{0.4\textwidth}
            \centering
            \includegraphics[width=\textwidth]{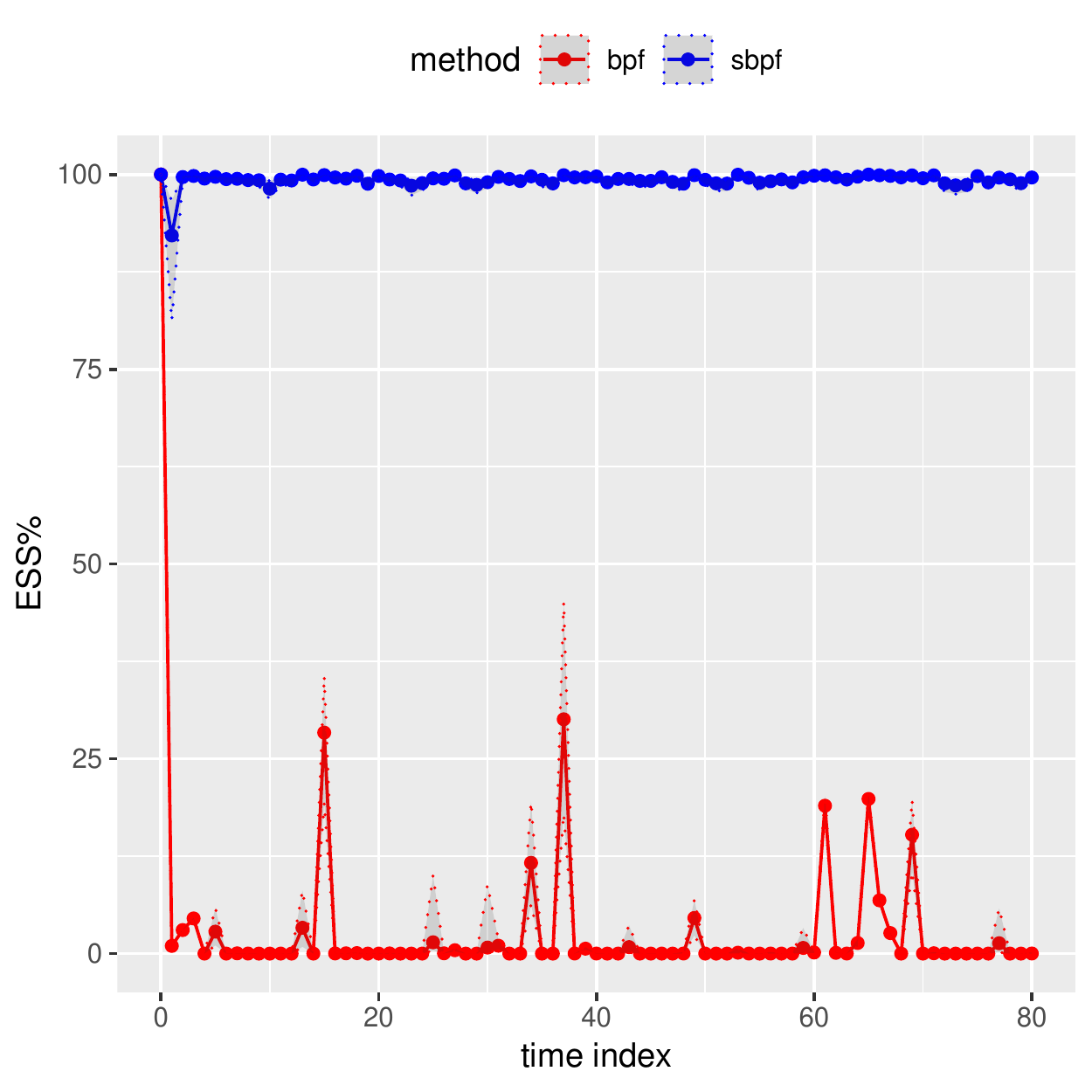}
            \caption{{\small ESS as $\%$ of sample size.}}
            \label{fig:LQG_sbpf_bpf_ESS_reps100_sigmaobs0.1_T80_h0.025}
        \end{subfigure}

        \caption{ {\small Results of the numerical experiments performed for the LQG model discussed in the particle filtering setting of Section \ref{sec:pf}. In Figure \ref{fig:LQG_sb_kf_sigmaobs0.1_T80_h0.025_IPF_exact}, we perform exact IPF and plot $\log \was_2(\pi_t,q_t^{(i)})$ for $t \in [1:T]$ and $i \in [1:5]$, where $\pi_t$ was derived using a Kalman filter. In Figure \ref{fig:LQG_sbpf_bpf_logwas_reps100_sigmaobs0.1_T80_h0.025}, we plot the log 2-Wasserstein distances between $\{\pi_t\}_{t\in[1:T]}$ and their Gaussian approximations based on the Schr\"odinger bridge and bootstrap particle filters. In Figures \ref{fig:LQG_sbpf_bpf_log_normconst_reps100_sigmaobs0.1_T80_h0.025} and \ref{fig:LQG_sbpf_bpf_ESS_reps100_sigmaobs0.1_T80_h0.025} we  plot the marginal likelihood estimates and effective sample sizes as a percentage of $N$ for the two particle filters. In each plot based on particle approximations, solid lines correspond to median values calculated over $100$ independent simulations, and the corresponding confidence bands represent the $5\%$ and $95\%$ quantiles.}
        \label{fig:LQG_sbpf_bpf}}
\end{figure}

\section{Numerical experiments}\label{sec:numerical_experiments}
In this section, we illustrate the proposed method in two different settings. In the first, we compare the performance of the SSB sampler to a standard SMC sampler in a Gaussian model of varying dimension. In the second, we apply the SSB sampler to a Bayesian logistic regression model with an instance of the weakly informative priors recommended by \citet{gelman2008weakly}. As these priors are non-Gaussian and consequently not conjugate with respect to Gaussian policies, this represents a setting where controlled SMC \citep{heng2017controlled} is not easily applicable. 

\subsection{LQG in various dimensions} \label{sec:lqg_highdim}
In this section, we apply the SSB sampler to the Gaussian model detailed in Section \ref{sec:lqg} of varying dimension.
As we vary $d=2^k$ for $k = [2:6]$, we adopt a prior distribution with mean $\mu_0 = 0$ and covariance $\Sigma_0 = \mathcal{I}_d$, and a log-likelihood function that is specified by taking $\xi = 25$ and $\rho = 0.8$.
We consider the discretized Langevin reference process defined by the kernels \eqref{eq:em_langevin} with $h = \tau/T$ and $\lambda_t = t/T$ and set $\tau = 2$ and $T = 40$.


We compare the standard SMC sampler defined by the reference process with three versions of the SSB sampler. The first method corresponds to the same algorithm applied in the two-dimensional LQG setting of Section \ref{sec:ssb_lqg}. The second method uses the modification of this algorithm discussed in Section \ref{sec:ssb}, in which we refresh the particles $\{X_{t-1}^n\}_{n \in [1:N]}$ at time $t-1$ using a Markov kernel that is invariant to $\pi_{t-1}$ for each iteration of IPF. This refreshment step can help prevent the regression-based approximation of $\psi_t$ from overfitting. In particular, we used one step of the Metropolis-adjusted Langevin algorithm \citep{roberts1996exponential} scaled with a diagonal covariance matrix estimated using $\{X_{t-1}^n\}_{n \in [1:N]}$ and step-size chosen to be $3/d^{1/3}$. The third method uses the Euler--Maruyama approximation to the exact Markov kernel twisting discussed in Section \ref{sec:sb_continuous} in combination with refreshment steps. For all three methods, we used $M = 0$ CSMC iterations and resampled the particles according to their incremental importance weights at each time step.

In $d$-dimensional space, parameterizing the quadratic function class containing policies of the form $-\log\psi_t(x_t) = x_t^\top A_t x_t + x_t^\top b_t+ c_t$ requires a total of  $d(d+1)/2 + d + 1$ parameters. By restricting $A_t$ to be diagonal, this number is reduced to $2d + 1$. This restriction is necessary in high-dimensional scenarios as the regression problems in step (c) of Algorithm \ref{algorithm:aIPF} becomes computationally prohibitive otherwise. In the following simulations, we use the fully parameterized form of $A_t$ when $d = 4, 8$ and the diagonal form of $A_t$ when $d = 32, 64$.  We compare the impact of the two different parameterizations in the $d=16$ case.

Of the three SSB samplers, the one using refreshment steps and exact kernel twisting  was the most computationally expensive. For this method, we set the number of particles to be $N = 1,000$ for $d = 4, 8,16$, and $N = 2,000$ for $d = 32, 64$ as more samples are required to learn more policy parameters in higher dimensions. The sample sizes of the other schemes were tuned so that the four methods ran in approximately the same wall-clock time; for each simulation setting, the most time consuming method among the four took no more than 10\% longer than the least time consuming (averaged over 100 runs). For each SSB method, we used IPF with warm starts, early stopping and a maximum of $I = 100$ iterations, while setting $\mathcal{K} = [0:T]$.

In Figure \ref{fig:lqg_highdim}, we plot the 2-Wasserstein distance between the target sequence $\{\pi_t\}_{t\in[1:T]}$ and the marginal distributions induced by the different approximate IPF schemes (left column), and the RMSE of the corresponding log-normalizing constant estimators $\log \hat{Z}_t$ (right column). Across all dimensions, we observe that the SSB samplers outperform the standard SMC sampler in estimating the log-normalizing constant. For $d=4,8$, there appears to be only limited benefit in including refreshment steps, but the latter seems to become increasingly useful as the dimension increases. Moreover, the Euler--Maruyama approximation provides a reasonable alternative to exact twisting in all the examples. 
For instance, in the $d=64$ setting, the RMSE of the standard SMC log-normalizing constant estimator at the terminal time was 21 and 17 times higher than the SSB estimators using exact twisting and Euler--Maruyama twisting, respectively. 

On the other hand, the policies obtained with the different SSB schemes appear to yield marginal distributions $q_t^{(I)}$ that are of similar 2-Wasserstein distance from the targets, at least for $t$ close to $T$. This illustrates that the quality of the marginal approximations alone does not explain the difference in performance of the log-normalizing constant estimators. This is further illustrated in the $d=16$ case, in which the distances between the marginals and the targets are observed to increase when moving from the full parameterization of $A_t$ to the diagonal parameterization, while the performance of the corresponding SSB samplers do not appear to deteriorate. On the contrary, the SSB sampler without rejuvenation looks to benefit from using the diagonal representation. A possible explanation for this behavior is that the relatively high number of parameters in the full parameterization leads the regression within the IPF steps to overfit, and that using the diagonal parameterization acts as a regularization. The fact that there is only marginal difference between the performance of the SSB samplers with rejuvenation between the full and diagonal parameterizations suggests that refreshing the samples can help mitigate such overfitting. Other approaches to prevent overfitting, such as estimating the policies with penalized regression, would also be worth investigating.

\begin{figure}[h]
    \centering
        \begin{subfigure}[b]{0.4\textwidth}
            \centering
            \includegraphics[width=\textwidth]{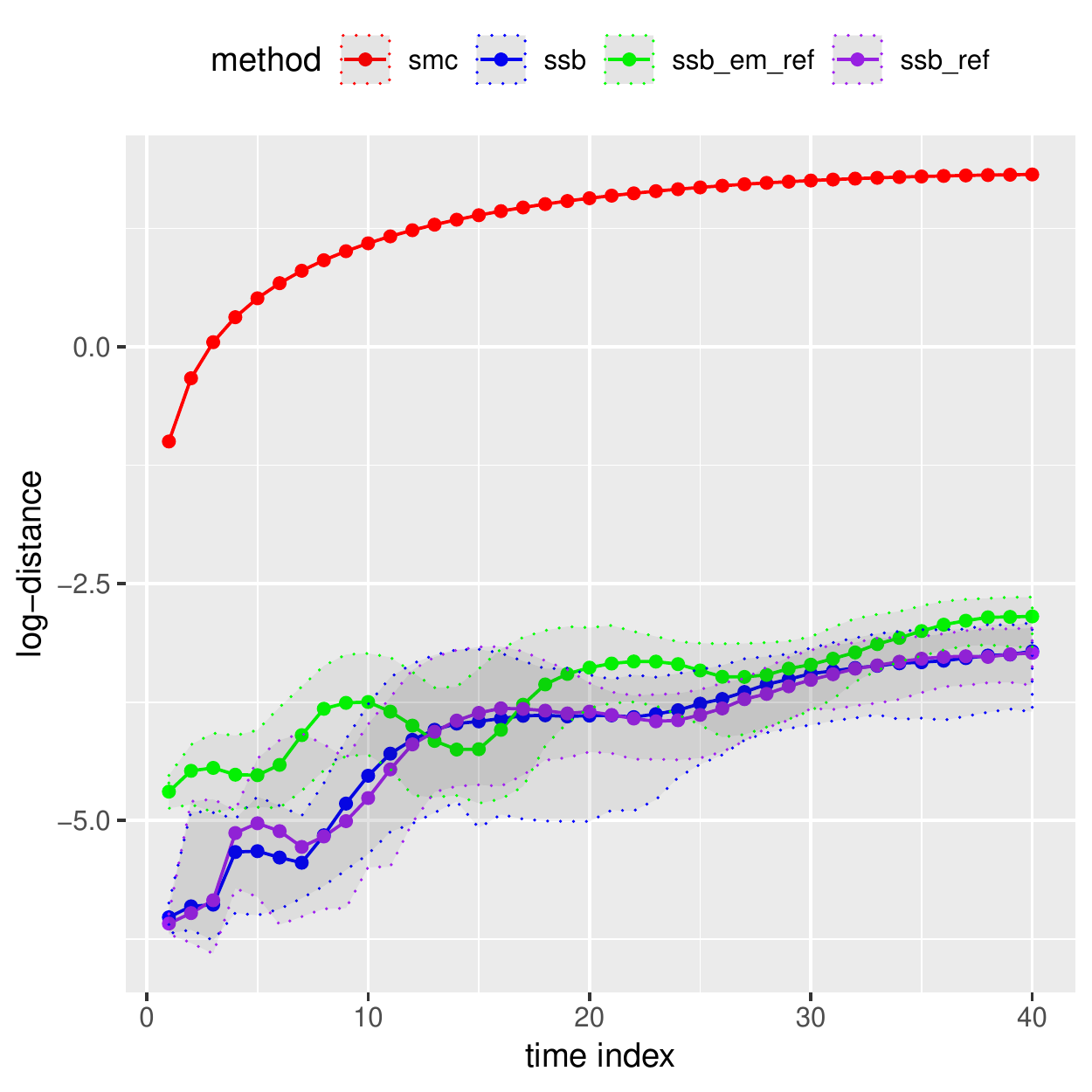}
            \caption{{\small $d = 4$, $\log \was_2(\pi_t,q_t^{(I)})$.}}
            \label{fig:langevin_ssb_logwas_IPF100_reps100_n1000_d4_T40_terminaltime2_diagonalFALSE}
        \end{subfigure}
        \hskip1cm
                  \begin{subfigure}[b]{0.4\textwidth}
            \centering
            \includegraphics[width=\textwidth]{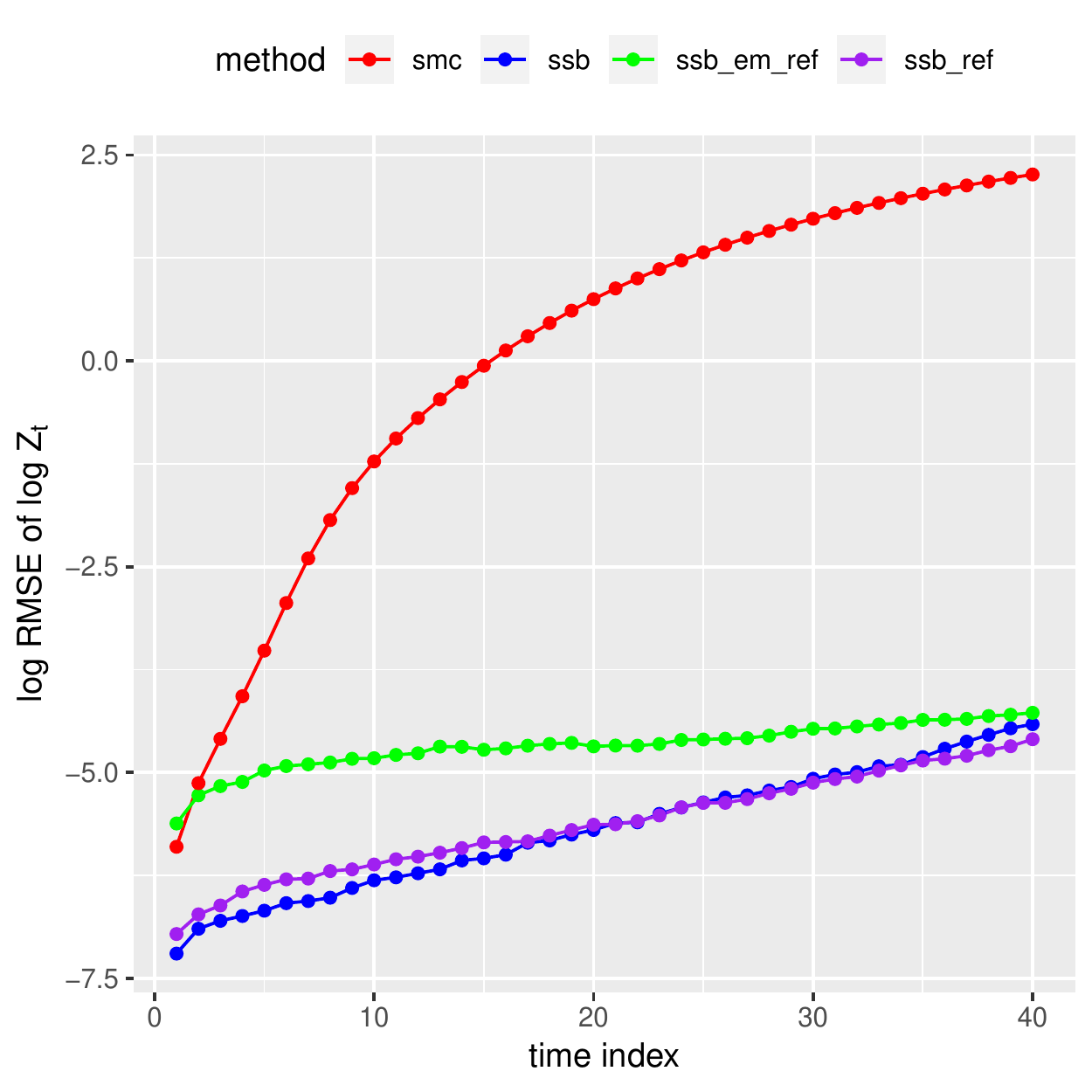}
            \caption{{\small $d = 4$,  log-RMSE of $\log \hat{Z}_t$.}}
            \label{fig:langevin_ssb_RMSE_lognormconst_IPF100_reps100_n1000_d4_T40_terminaltime2_diagonalFALSE}
        \end{subfigure}

        \begin{subfigure}[b]{0.4\textwidth}
            \centering
            \includegraphics[width=\textwidth]{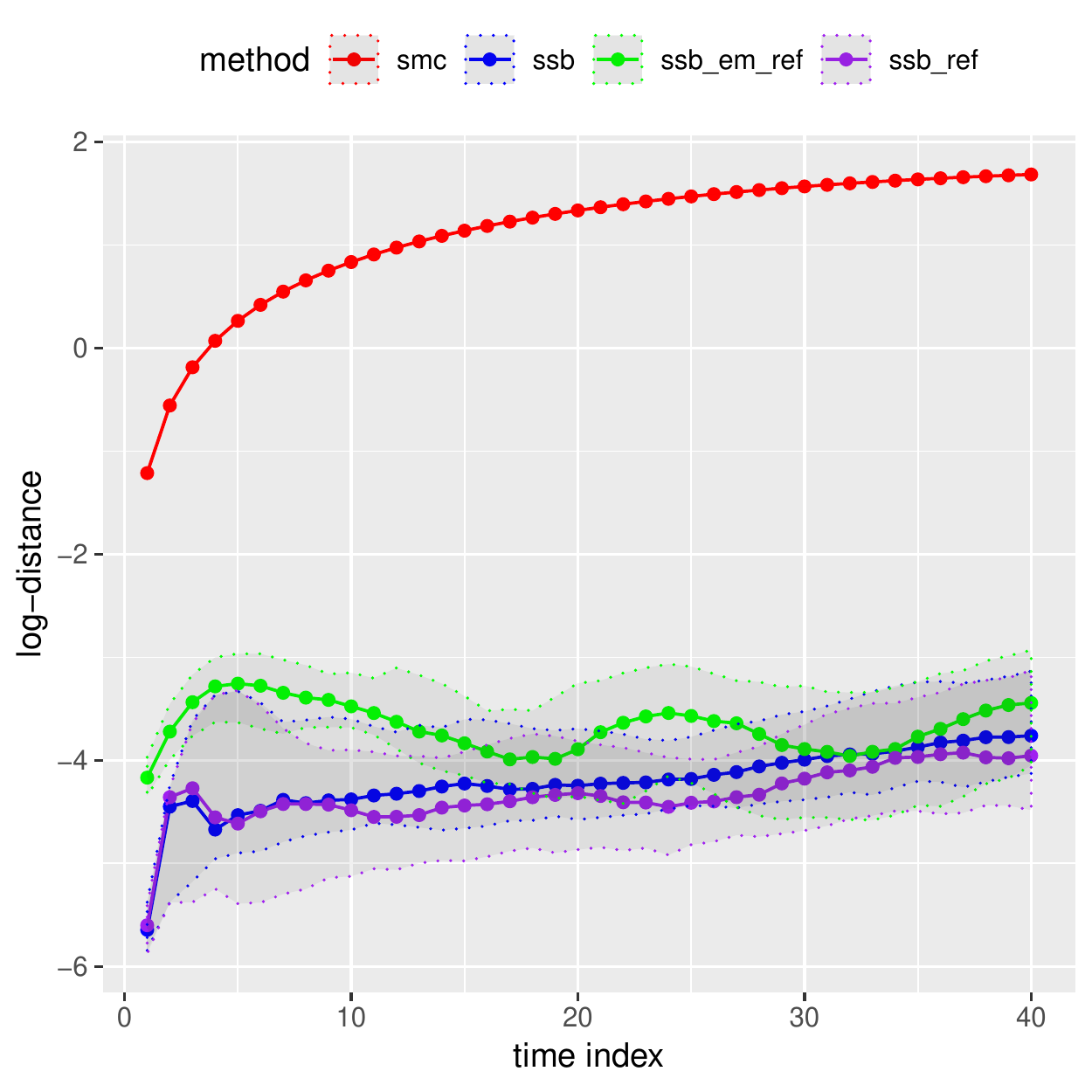}
            \caption{{\small $d = 8$, $\log \was_2(\pi_t,q_t^{(I)})$.}}
            \label{fig:langevin_ssb_logwas_IPF100_reps100_n1000_d8_T40_terminaltime2_diagonalFALSE}
            \end{subfigure}
            \hskip1cm
                    \begin{subfigure}[b]{0.4\textwidth}
            \centering
            \includegraphics[width=\textwidth]{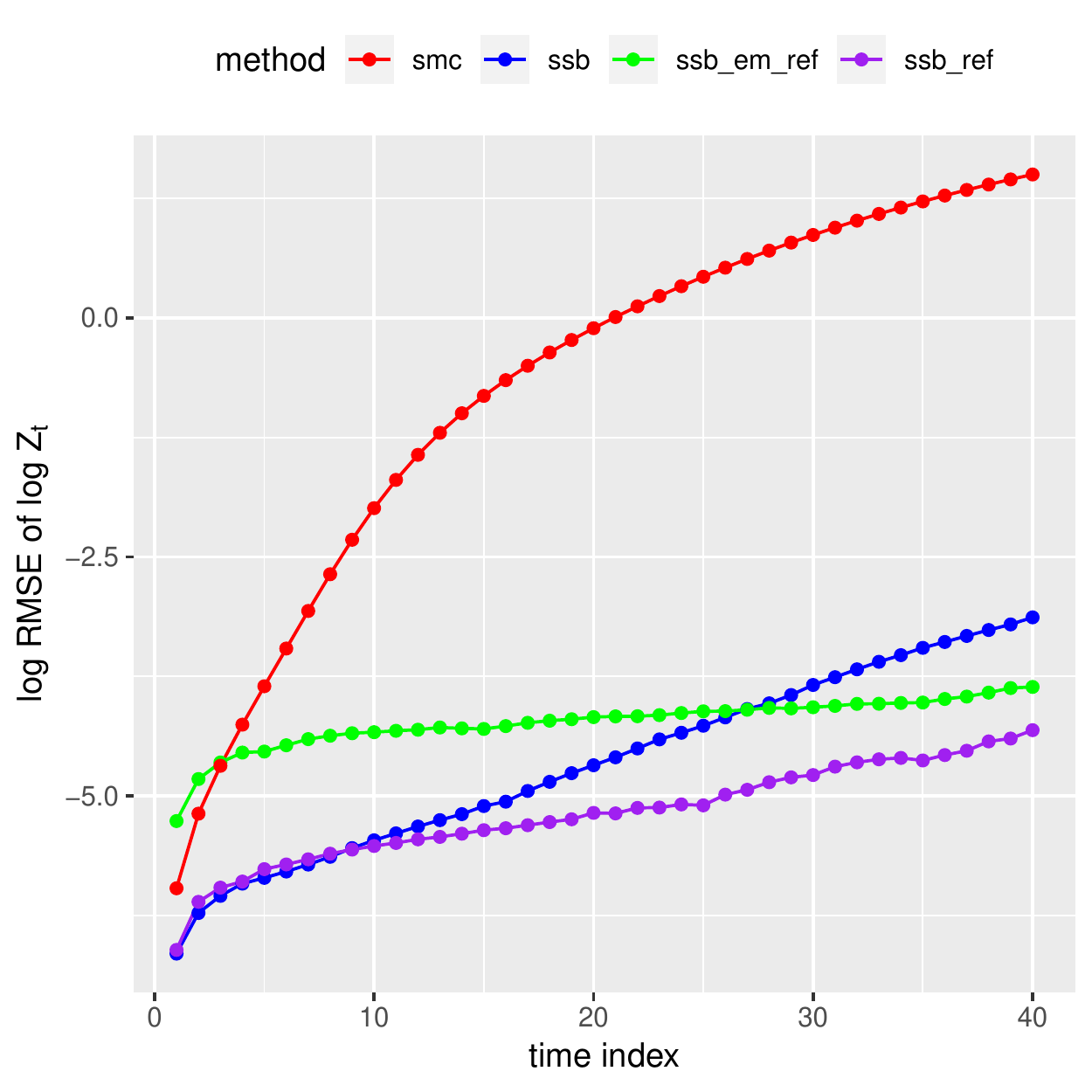}
            \caption{{\small $d = 8$,  log-RMSE of $\log \hat{Z}_t$.}}
            \label{fig:langevin_ssb_RMSE_lognormconst_IPF100_reps100_n1000_d8_T40_terminaltime2_diagonalFALSE}
            \end{subfigure}

               \caption{{\bf Part 1/3.} {\small Distances between $\pi_t$ and marginals of the IPF-based approximations $q_t^{(I)}$ of the multi-marginal Schr\"odinger bridge, measured as $\log \was_2(\pi_t,q_t^{(I)})$, for the LQG setting of Section \ref{sec:lqg_highdim} with discretized Langevin diffusion reference dynamics (left column), and log-RMSE of the log-normalizing constant estimates $\log \hat{Z}_t$ (right column). The red lines correspond to the reference process and the estimators obtained with the corresponding SMC sampler. The blue and purple lines correspond to the SSB sampler with exact twisting with and without refreshment steps, respectively. The green line correspond to the SSB sampler using Euler--Maruyama approximate twisting and refreshment steps. Figure continued on the next page. \label{fig:lqg_highdim}}}
\end{figure}

\begin{figure}[h]\ContinuedFloat
    \centering
         \begin{subfigure}[b]{0.4\textwidth}
            \centering
            \includegraphics[width=\textwidth]{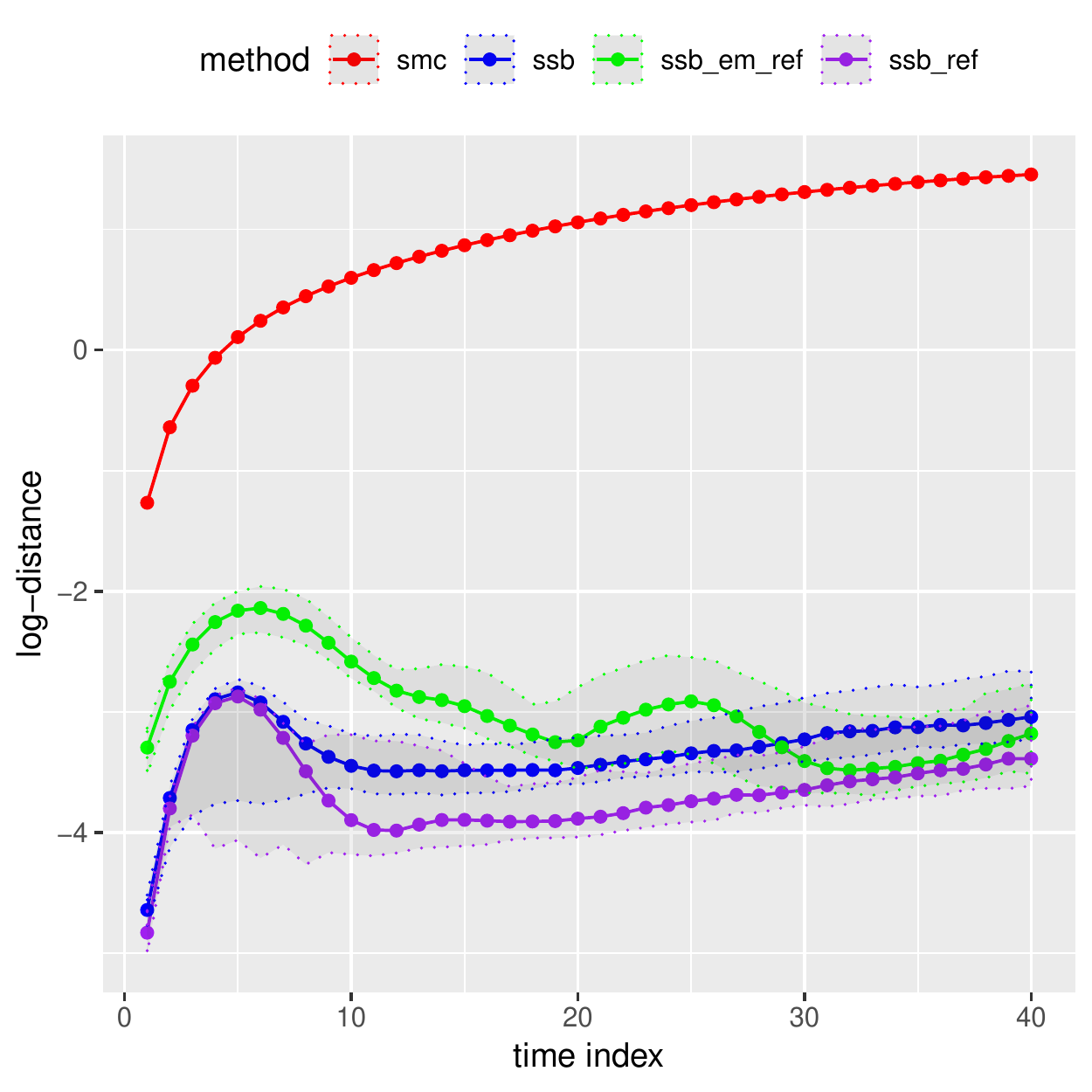}
             \caption{{\small $d = 16$, $\log \was_2(\pi_t,q_t^{(I)})$, policies estimated using full $A_t$.}}
            \label{fig:langevin_ssb_logwas_IPF100_reps100_n1000_d16_T40_terminaltime2_diagonalFALSE}
        \end{subfigure}
        \hskip1cm
                  \begin{subfigure}[b]{0.4\textwidth}
            \centering
            \includegraphics[width=\textwidth]{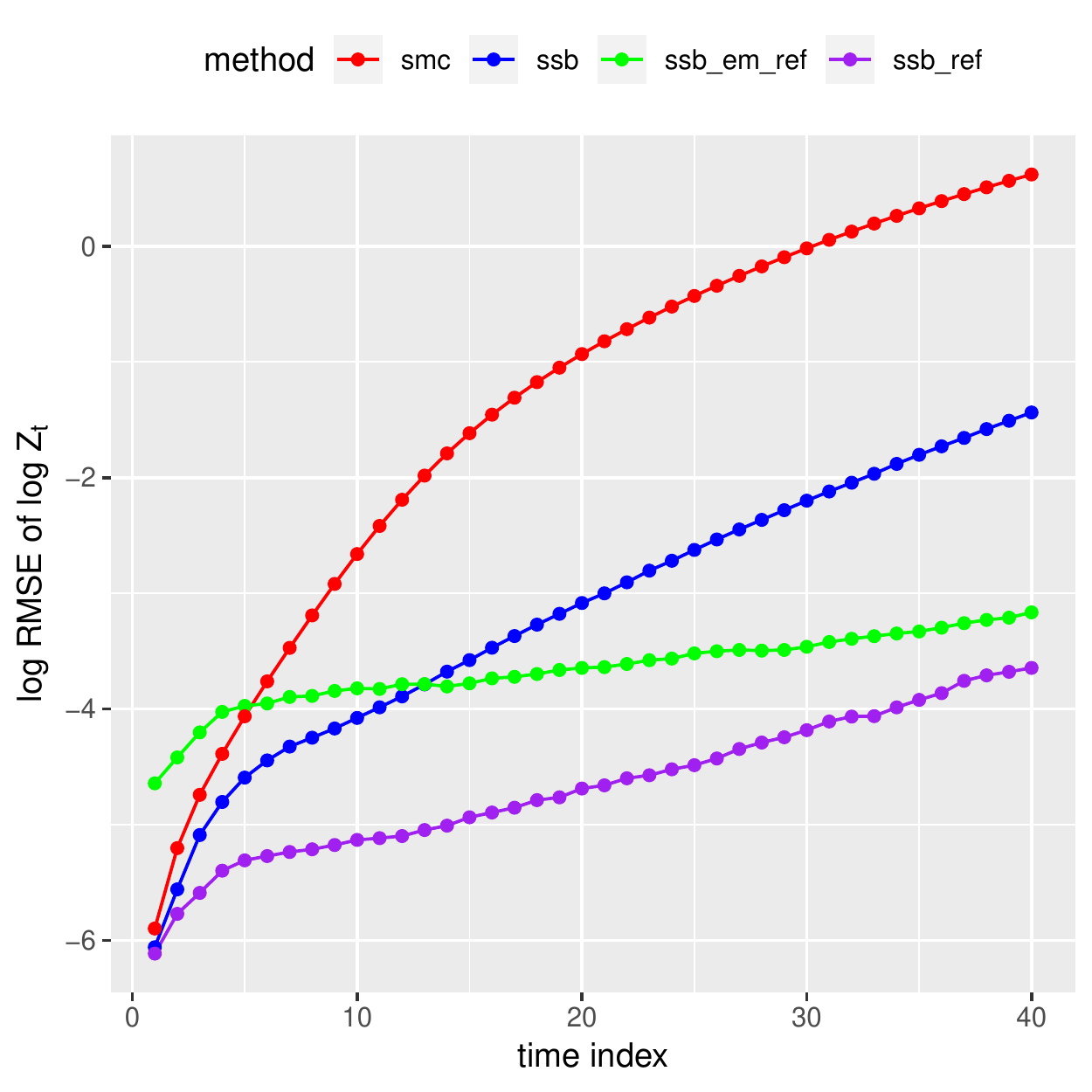}
            \caption{{\small $d = 16$, log-RMSE of $\log \hat{Z}_t$, policies estimated using full $A_t$.}}
            \label{fig:langevin_ssb_RMSE_lognormconst_IPF100_reps100_n1000_d16_T40_terminaltime2_diagonalFALSE}
        \end{subfigure}

        \begin{subfigure}[b]{0.4\textwidth}
            \centering
            \includegraphics[width=\textwidth]{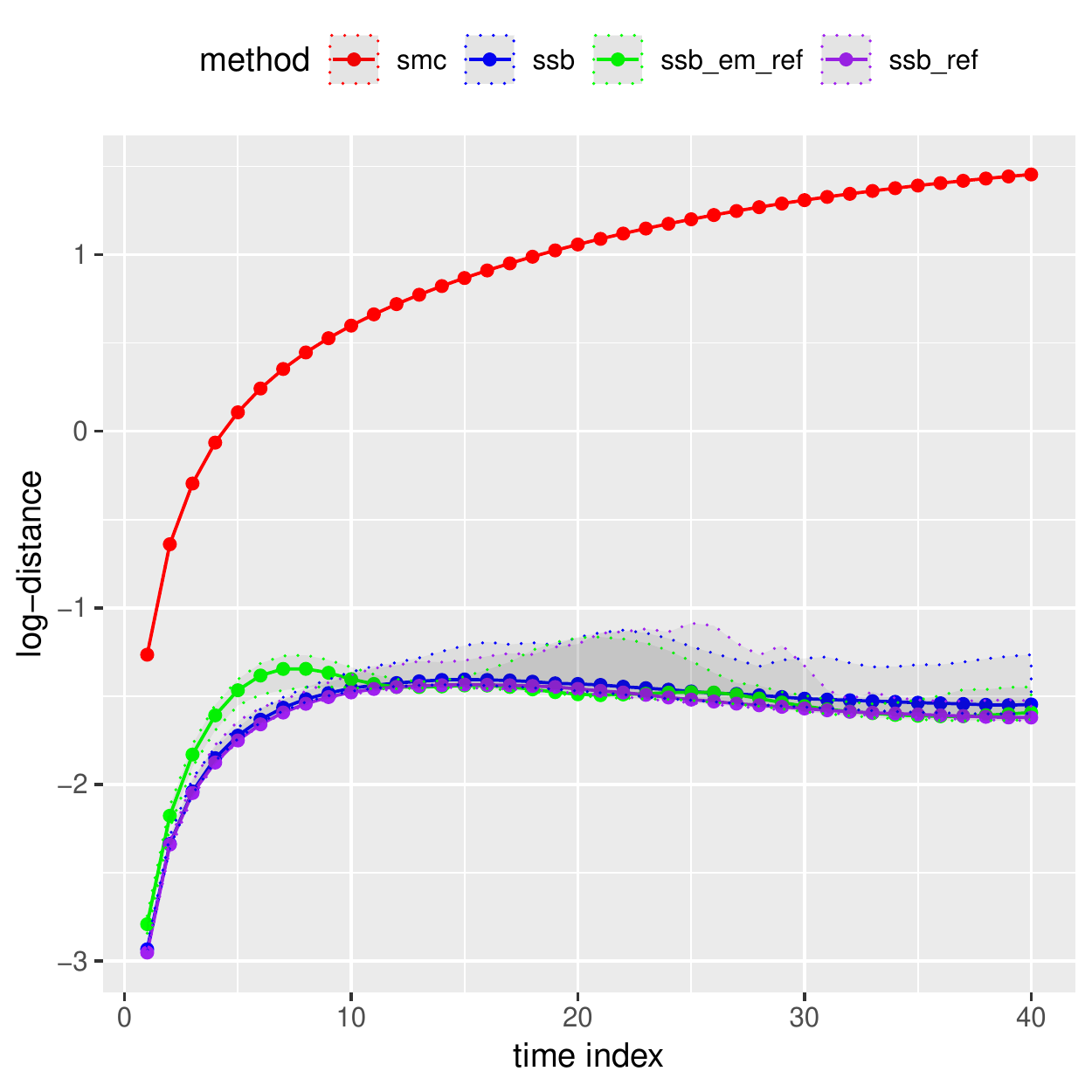}
             \caption{{\small $d = 16$, $\log \was_2(\pi_t,q_t^{(I)})$,  policies estimated using diagonal $A_t$.}}
            \label{fig:langevin_ssb_logwas_IPF100_reps100_n1000_d16_T40_terminaltime2_diagonalTRUE}
            \end{subfigure}
            \hskip1cm
             \begin{subfigure}[b]{0.4\textwidth}
            \centering
            \includegraphics[width=\textwidth]{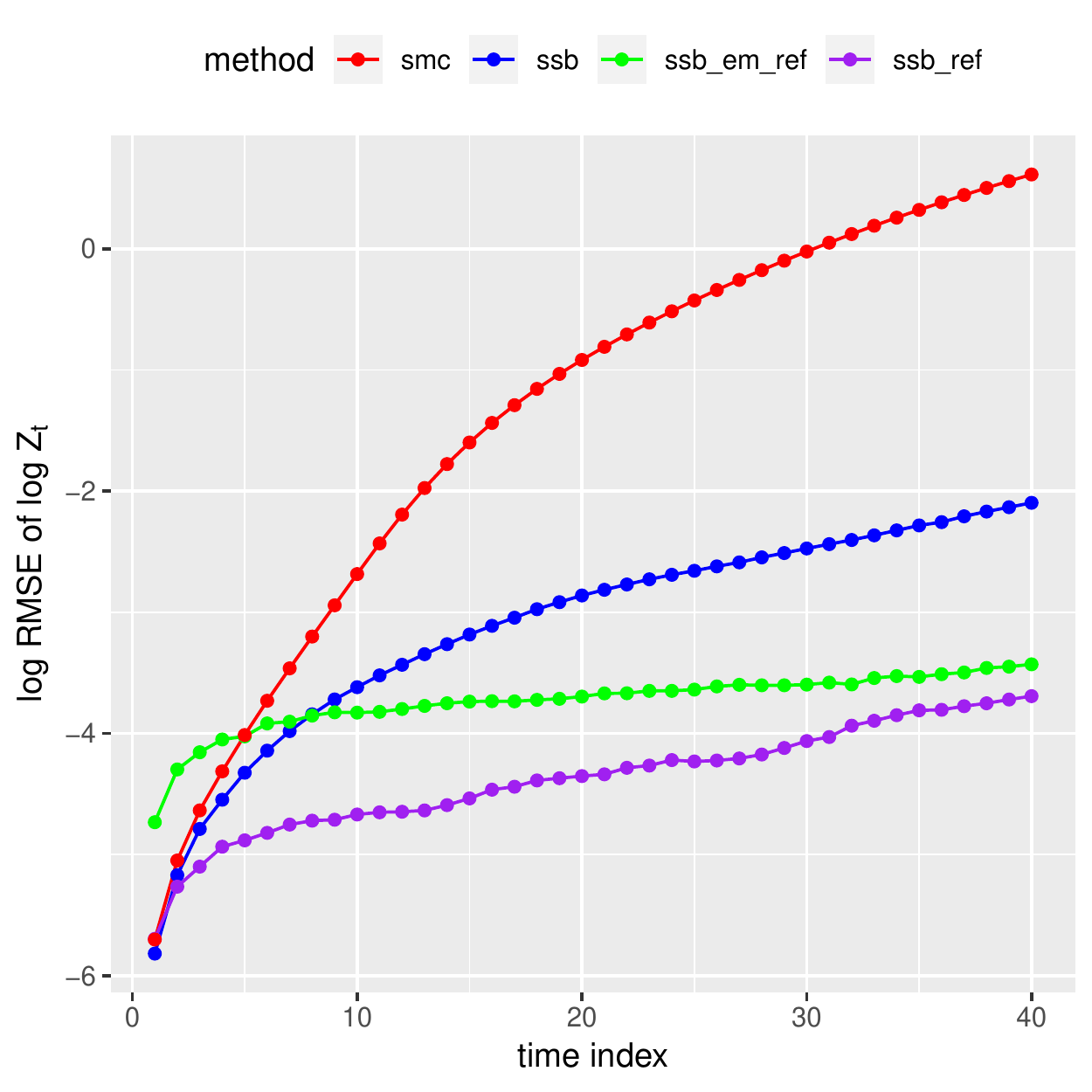}
            \caption{{\small $d = 16$, log-RMSE of $\log \hat{Z}_t$, policies est.~using diagonal $A_t$.}}
            \label{fig:langevin_ssb_RMSE_lognormconst_IPF100_reps100_n1000_d16_T40_terminaltime2_diagonalTRUE}
            \end{subfigure}

                       \caption[]{ {\bf Part 2/3.} {\small Distances between $\pi_t$ and marginals of the IPF-based approximations $q_t^{(I)}$ of the multi-marginal Schr\"odinger bridge, measured as $\log \was_2(\pi_t,q_t^{(I)})$, for the LQG setting of Section \ref{sec:lqg_highdim} with discretized Langevin diffusion reference dynamics (left column), and log-RMSE of the log-normalizing constant estimates $\log \hat{Z}_t$ (right column). The red lines correspond to the reference process and the estimators obtained with the corresponding SMC sampler. The blue and purple lines correspond to the SSB sampler with exact twisting with and without refreshment steps, respectively. The green line correspond to the SSB sampler using Euler--Maruyama approximate twisting and refreshment steps. Figure continued on the next page. }}
\end{figure}

\begin{figure}[h]\ContinuedFloat
    \centering
         \begin{subfigure}[b]{0.4\textwidth}
            \centering
            \includegraphics[width=\textwidth]{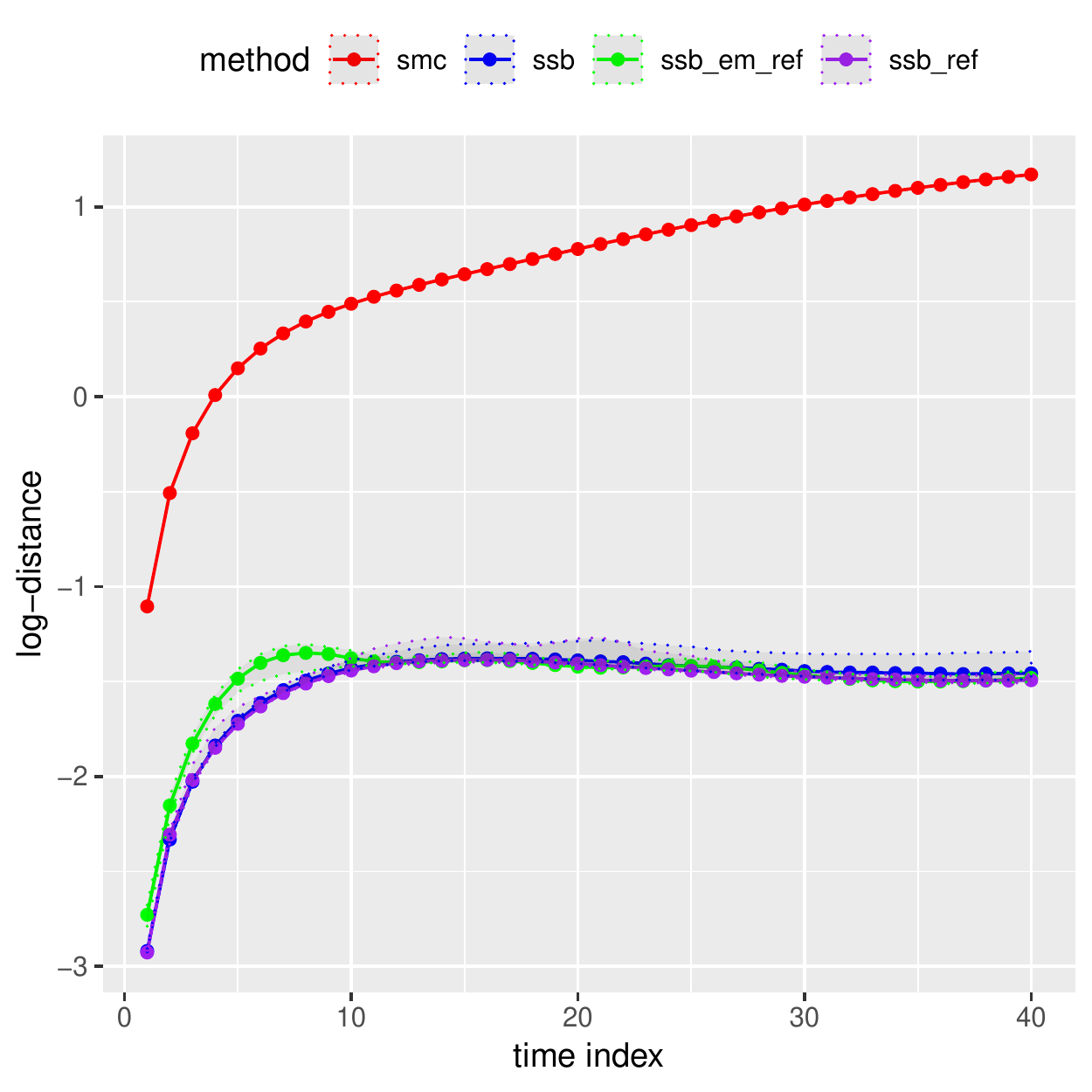}
            \caption{{\small $d = 32$, $\log \was_2(\pi_t,q_t^{(I)})$.}}
            \label{fig:langevin_ssb_logwas_IPF100_reps100_n2000_d32_T40_terminaltime2_diagonalTRUE}
        \end{subfigure}
        \hskip1cm
                  \begin{subfigure}[b]{0.4\textwidth}
            \centering
            \includegraphics[width=\textwidth]{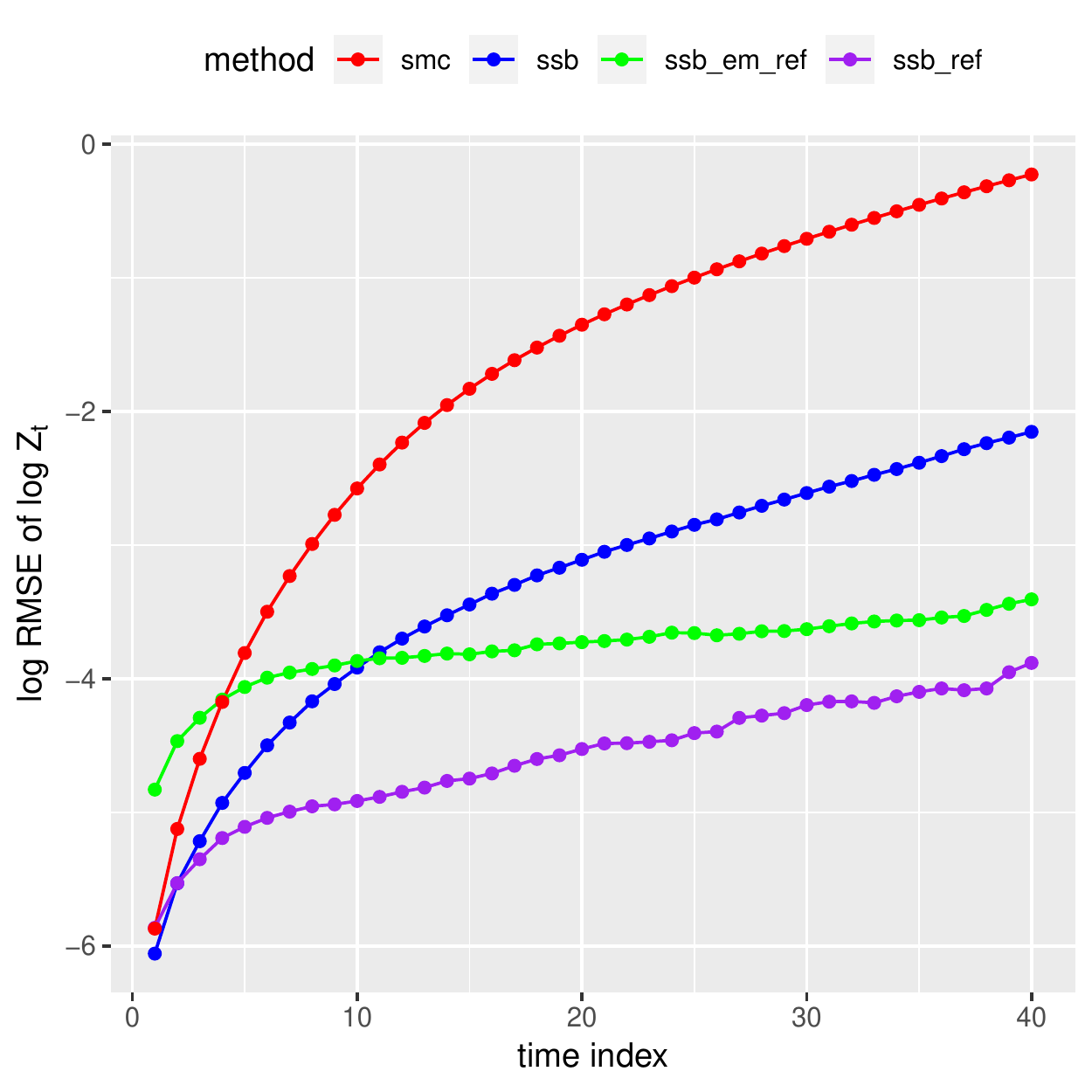}
            \caption{{\small $d = 32$, log-RMSE of $\log \hat{Z}_t$.}}
            \label{fig:langevin_ssb_RMSE_lognormconst_IPF100_reps100_n2000_d32_T40_terminaltime2_diagonalTRUE}
        \end{subfigure}

        \begin{subfigure}[b]{0.4\textwidth}
            \centering
            \includegraphics[width=\textwidth]{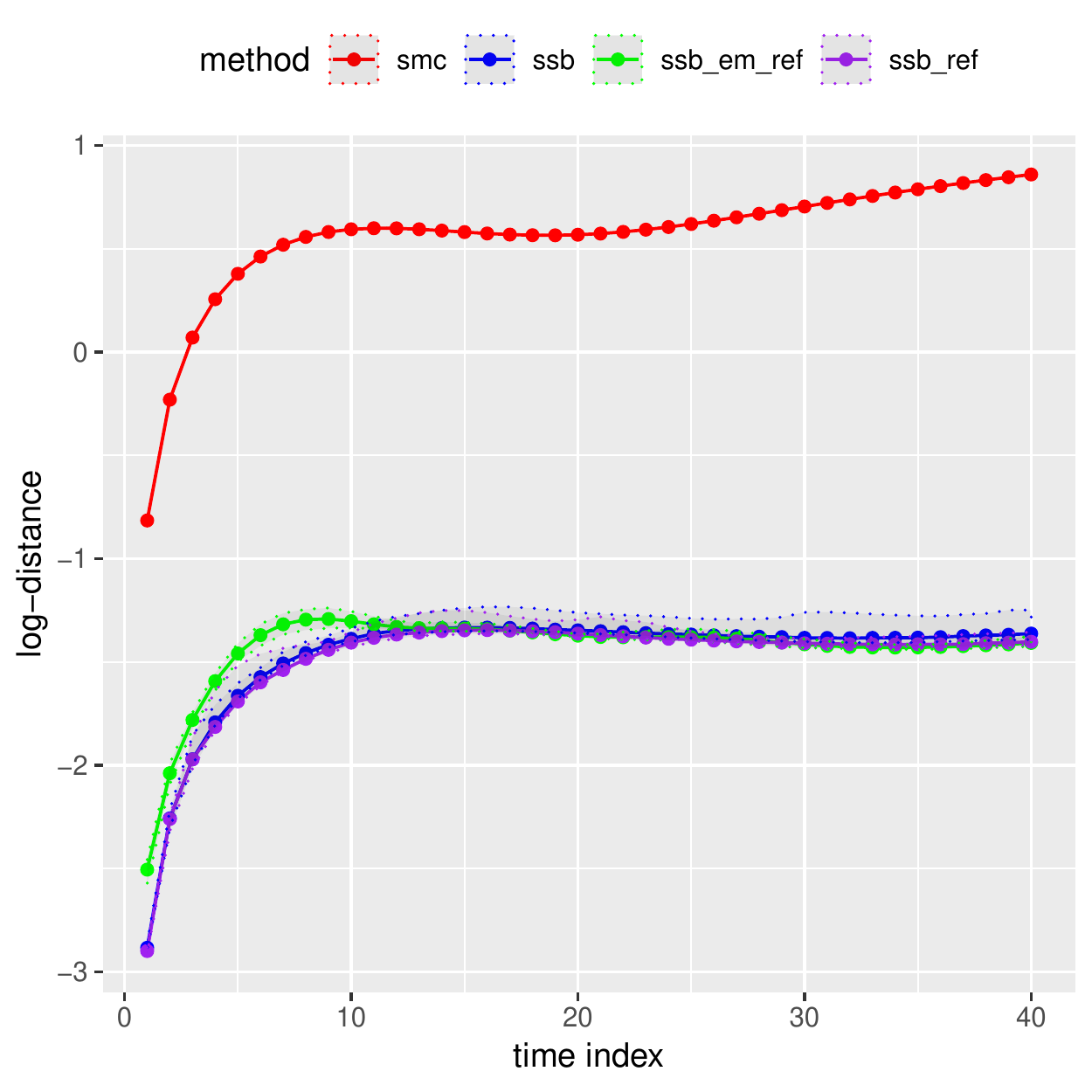}
            \caption{{\small $d = 64$, $\log \was_2(\pi_t,q_t^{(I)})$.}}
            \label{fig:langevin_ssb_logwas_IPF100_reps100_n2000_d64_T40_terminaltime2_diagonalTRUE}
            \end{subfigure}
            \hskip1cm
                    \begin{subfigure}[b]{0.4\textwidth}
            \centering
            \includegraphics[width=\textwidth]{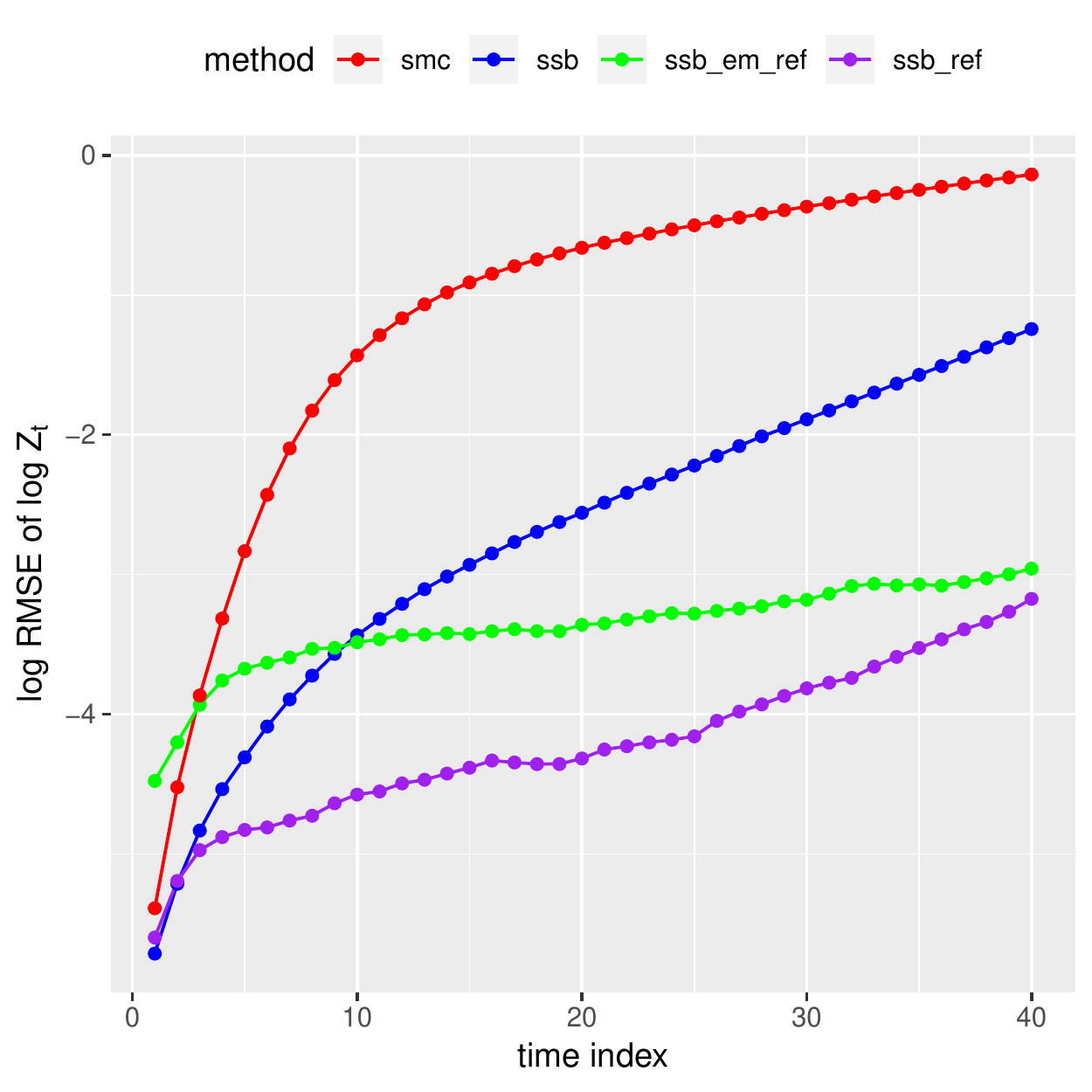}
            \caption{{\small $d = 64$, log-RMSE of $\log \hat{Z}_t$.}}
            \label{fig:langevin_ssb_RMSE_lognormconst_IPF100_reps100_n2000_d64_T40_terminaltime2_diagonalTRUE}
            \end{subfigure}
                       \caption[]{{\bf Part 3/3.}  {\small Distances between $\pi_t$ and marginals of the IPF-based approximations $q_t^{(I)}$ of the multi-marginal Schr\"odinger bridge, measured as $\log \was_2(\pi_t,q_t^{(I)})$, for the LQG setting of Section \ref{sec:lqg_highdim} with discretized Langevin diffusion reference dynamics (left column), and log-RMSE of the log-normalizing constant estimates $\log \hat{Z}_t$ (right column). The red lines correspond to the reference process and the estimators obtained with the corresponding SMC sampler. The blue and purple lines correspond to the SSB sampler with exact twisting with and without refreshment steps, respectively. The green line correspond to the SSB sampler using Euler--Maruyama approximate twisting and refreshment steps. }}
\end{figure}

\subsection{Bayesian logistic regression}
We fit a Bayesian logistic regression model to the Cleveland heart disease database in the UCI machine learning repository\footnote{Available at http://archive.ics.uci.edu/ml/datasets/heart+disease.}, which contains information about the presence of heart disease in 303 patients as well as measurements of 13 other predictors. We removed 6 entries with missing  covariates, and turned categorical variables with more than two categories into corresponding binary vectors, leaving us with a data set with $M = 297$ individuals, each with $d = 20$ binary or continuous covariates. Additionally, we made the response variable binary by only indicating the presence or absence of heart disease even though the data set contained four different categories of disease presence.

We use a prior distribution on the $d=20$ regression coefficients from the class of weakly informative default priors for logistic regression suggested by \citet{gelman2008weakly}. In particular, we first shift all the binary predictors to have mean $0$ and to differ by $1$ in their lower and upper conditions. Similarly, the non-binary predictors are shifted to have means of $0$ and normalized to have standard deviations equal to 0.5. This puts all the variables on the same scale. Next, we place independent $t$-distributions with centers $0$, scales $2.5$ and four degrees of freedom on each of the coefficients. The log-likelihood is defined via the logistic link function and is given by
\begin{equation}
\ell(x) = y^\top X x - \sum_{m=1}^M \log(1 + \exp(x^\top X_m)),
\end{equation}
where $y\in\{0,1\}^M$ is the response variable, $X_m \in \mathbb{R}^d$ is the $m$-th row of the design matrix $X\in\mathbb{R}^{M\times d}$. The gradient of the log-likelihood is in turn given by
\begin{equation}
\nabla \ell(x) = X^\top y - \sum_{m=1}^M (1 + \exp(-x^\top X_m))^{-1}X_m.
\end{equation}

To draw samples from the resulting posterior distribution, we obtain a reference process by discretizing the Langevin diffusion associated with the interpolation \eqref{eq:geometric_path}, setting $\tau = 2$, $T = 40$ and $h = \tau/T$ as in the previous sections. Unlike the earlier sections, however, we choose $\lambda_t = t^2/T^2$ for each $t\in[0:T]$. Compared to the linear interpolation, this makes consecutive distributions $\pi_t$ and $\pi_{t+1}$ closer when $t$ is small. Using the quadratic interpolation appeared to make the samplers better behaved, which is possibly because the prior distribution is disperse compared to the likelihood by virtue of being weakly informative. We used the SSB sampler with policies estimated within the Gaussian function class, which, unlike in the Gaussian setting, is no longer guaranteed to contain the optimal policy. Additionally, we restricted the matrix $A_t$ to be diagonal. Again, we used IPF with warm starts and early stopping, allowing a maximum of $I = 20$ IPF iterations per time step. For each IPF iteration, we refresh the particles using the same MALA kernel as in Section \ref{sec:lqg_highdim} with a step-size equal to $1/d^{1/3}$. The number of particles was chosen to be $N = 2,000$. We compare the SSB sampler with the standard SMC sampler induced by the reference process, for which we tuned the number of particles to be such that the two algorithms ran in the same amount of wall-clock time. Over 100 independent runs, log-normalizing constant estimators were on average -126.47 and -128.11 with standard deviations of 0.034 and 1.47 for SSB and SMC respectively, illustrating a large reduction in variance.

\section{Discussion}\label{sec:discussion}
In this paper, we have presented a new SMC algorithm based on an original method to sequentially approximate the solution of a multi-marginal Schr\"odinger bridge problem, which we termed the sequential Schr\"odinger bridge (SSB) sampler. The algorithm computes two-marginal Schr\"odinger bridges based on an approximate version of IPF. We showed that the first iteration of IPF can be seen as finding the optimal SMC auxiliary target distribution derived by \citet{del2006sequential} for a given set of  Markov kernels, and that further iterations modify the kernels in such a way that the induced joint distribution converges to the Schr\"odinger bridge.

By making use of the problem's equivalent formulation in terms of optimal control, we formulated the steps of the IPF algorithm as policy refinements and showed how to estimate them using an approximate dynamic programming methodology similar to the one developed in \citet{heng2017controlled}. Unlike the controlled SMC algorithm proposed therein, the SSB sampler is applicable when the initial distribution is not conjugate with respect to the chosen policy. Furthermore, by analogy with the continuous-time formulation of the Schr\"odinger bridge problem for a Langevin diffusion reference process, we proposed an Euler--Maruyama approximation to the ADP algorithm that can also alleviate the need for using policies that are conjugate with respect to the Markov transition kernels.

We illustrated our approach in various numerical experiments, including a linear quadratic Gaussian setting and a Bayesian logistic regression model. 
While controlling for computational time, SSB samplers outperform a standard SMC algorithm in the estimation of log-normalizing constants, in some cases reducing the RMSE of the estimators by several orders of magnitude. 
In the Gaussian setting, we also showed that the Euler--Maruyama approximation of the ADP algorithm provided a reasonable alternative to its exact counterpart. However, there are several methodological extensions left to consider. In particular, applying more general and flexible policy approximation methods, such as neural networks, could potentially lead to improvements. As done in \citet{genevay2018learning} for finite spaces, we could also unroll the IPF recursion over a few iterations and then learn the parameters of the policy using gradient techniques by maximizing the log-normalizing constant estimate as the algorithm is end-to-end differentiable. In the context of controlled SMC, such an approach has been employed in \citet{lawsontwisted}.

There are also several theoretical aspects that are left to consider, such as a more thorough analysis of the asymptotic properties of the Schr\"odinger bridge approximation as $N$ and $I$ (and also $M$ and $T$) grow, and how their relative sizes impact the efficiency of the method.
For instance, we have not formally studied the behavior of the IPF algorithm when the function classes used within ADP are misspecified in the sense that they do not contain the optimal policy given by the underlying system of Schr\"odinger equations.

The Schr\"odinger bridge problem falls at the intersection of many different literatures, including probability theory, optimal transport, control theory, and physics. We have discussed some of the connections between our methodology and computational approaches to approximate the 2-Wasserstein distance between two distributions, the flow transport problem, and the problem of constructing shortcuts to adiabaticity in thermodynamics. We anticipate that further exploration into these and other related problems, such as particle filtering and inference for diffusion processes, would  be interesting.

\singlespacing
\bibliographystyle{apalike}
\bibliography{references}

\begin{appendix}
\section{Proofs}\label{appendix:proofs}
\begin{proof}[Proof of Proposition \ref{prop:IPF_convergence}]
Recall that $\mathbb{S}^{(2i)}  = \mathbb{Q}^{(i)} \in \mathcal{P}_0(\pi_0)$ and $\mathbb{S}^{(2i+1)}  = \mathbb{P}^{(i+1)}\in \mathcal{P}_T(\pi_T)$, and that for any $\mathbb{H} \in \mathcal{P}_T(\pi_T)$ and $\mathbb{G} \in \mathcal{P}(\mathsf{E}^{T+1})$,
\begin{equation}
\mathrm{KL}(\mathbb{H}|\mathbb{G})  =  \mathrm{KL}(\pi_T | g_T) + \int_{\mathsf{E}} \mathrm{KL}\left(\mathbb{H}(\mathrm{d}x_{0:T-1} | x_T) | \mathbb{G}(\mathrm{d}x_{0:T-1} | x_T)\right) \pi_T(\mathrm{d}x_{T}).
\end{equation}
Hence, $\mathbb{P}^{(i+1)} = \argmin_{\mathbb{H} \in \mathcal{P}_T(\pi_T)} \mathrm{KL}(\mathbb{H} | \mathbb{Q}^{(i)})$ is such that $\mathbb{P}^{(i+1)}(\mathrm{d}x_{0:T-1} | x_T) = \mathbb{Q}^{(i)}(\mathrm{d}x_{0:T-1} | x_T)$ for $\pi_T$-almost every $x_T$, since the constraint pertains only to the first term in the decomposition above. So, for any $\mathbb{H} \in \mathcal{P}_T(\pi_T)$,
\begin{align}
\mathrm{KL}(\mathbb{H}|\mathbb{S}^{(2i)})  - \mathrm{KL}(\mathbb{H}|\mathbb{S}^{(2i+1)}) &= \mathrm{KL}(\pi_T | s^{(2i)}_T) \\
&= \mathrm{KL}(\pi_0 | s^{(2i)}_0) + \mathrm{KL}(\pi_T | s^{(2i)}_T),
\end{align}
since $s^{(2i)}_0 = \pi_0$. By analogous reasoning, we have that $\mathbb{Q}^{(i)}(\mathrm{d}x_{1:T} | x_0) = \mathbb{P}^{(i)}(\mathrm{d}x_{1:T} | x_0)$ for $\pi_0$-almost every $x_0$ and that for any $\mathbb{H} \in \mathcal{P}_0(\pi_0)$,
\begin{align}
\mathrm{KL}(\mathbb{H}|\mathbb{S}^{(2i-1)})  - \mathrm{KL}(\mathbb{H}|\mathbb{S}^{(2i)}) &= \mathrm{KL}(\pi_0 | s^{(2i-1)}_0) \\
&= \mathrm{KL}(\pi_0 | s^{(2i-1)}_0) + \mathrm{KL}(\pi_T | s^{(2i-1)}_T).
\end{align}

Hence, for any $\mathbb{H} \in \mathcal{P}_{0,T}(\pi_{0,T})$ and $k \in\mathbb{N}$ we have that
\begin{equation}
\mathrm{KL}(\mathbb{H}|\mathbb{S}^{(0)}) - \mathrm{KL}(\mathbb{H}|\mathbb{S}^{(k)}) =  \sum_{i=1}^{k-1} \mathrm{KL}(\pi_0 | s^{(i)}_0) + \mathrm{KL}(\pi_T | s^{(i)}_T).
\end{equation}
Let $\varepsilon > 0$. Let $k^\star$ be the first iteration such that $\mathrm{KL}(\pi_0 | s^{(i)}_0) + \mathrm{KL}(\pi_T | s^{(i)}_T) < \varepsilon$, the existence of $k^\star$ being 
guaranteed by \citet[][Proposition 2.1]{ruschendorf1995convergence}. By the definition of $k^\star$, we have
\begin{equation}
\mathrm{KL}(\mathbb{H}|\mathbb{S}^{(0)}) \geq \mathrm{KL}(\mathbb{H}|\mathbb{S}^{(0)}) - \mathrm{KL}(\mathbb{H}|\mathbb{S}^{(k^\star)}) \geq k^\star \varepsilon,
\end{equation}
from which it follows that
\begin{equation}
k^\star \leq  \frac{\mathrm{KL}(\mathbb{H}|\mathbb{S}^{(0)})}{\varepsilon} = \frac{\mathrm{KL}(\mathbb{H}|\mathbb{Q})}{\varepsilon}.
\end{equation}
In particular, $k^\star \leq \mathrm{KL}(\mathbb{S}|\mathbb{Q})/\varepsilon$ since $\mathbb{S} \in \mathcal{P}_{0,T}(\pi_{0,T})$. The result follows.
\end{proof}

\section{Linear quadratic Gaussian} \label{appendix:lqg}
Recall the setting introduced in Section \ref{sec:lqg}, in which we set $\pi_0(\mathrm{d}x_0) = \mathcal{N}(x_0; \mu_0, \Sigma_0)\mathrm{d}x_0$ and $\pi_T(\mathrm{d}x_T) = \mathcal{N}(x_T; \mu_T, \Sigma_T)\mathrm{d}x_T$ for some $\mu_0,\mu_T\in \mathbb{R}^d$ and $\Sigma_0,\Sigma_T\in \mathbb{R}^{d\times d}$, and for each $t\in[1:T]$, the kernel $M_t(x_{t-1},\mathrm{d}x_t) = \mathcal{N}(x_t; K_tx_{t-1} + r_t, H_t)\mathrm{d}x_t$ for some $r_t \in \mathbb{R}^d$ and $K_t, H_t\in \mathbb{R}^{d\times d}$. We derive the form of the exact policies $\psi^{(i)}$, which enables the comparison with the approximate method.

To find $\psi^{(i)}$ and the corresponding path measure $\mathbb{Q}^{\psi^{(i)}}$ defined by the IPF iterations, we proceed by induction. Suppose that $\mathbb{Q}^{\psi^{(i-1)}} (\mathrm{d}x_{0:T})= \pi_0(\mathrm{d}x_0)\prod_{t=1}^T M_t^{\psi^{(i-1)}}(x_{t-1},\mathrm{d}x_t)$, where $M_t^{\psi^{(i-1)}}(x_{t-1},\mathrm{d}x_t) = \mathcal{N}(x_t; K^{(i-1)}_tx_{t-1} + r^{(i-1)}_t, H^{(i-1)}_t)\mathrm{d}x_t$. The marginal distributions of $\mathbb{Q}^{\psi^{(i-1)}}$ are then Gaussian: for each $t\in[1:T]$, $q_t^{\psi^{(i-1)}}(\mathrm{d}x_t) = \mathcal{N}(x_t; \mu^{(i-1)}_t, \Sigma^{(i-1)}_t)\mathrm{d}x_t$, where $\mu^{(i-1)}_t$ and $\Sigma^{(i-1)}_t$ are given by the recursions
\begin{equation}
\mu^{(i-1)}_t = K_t^{(i-1)}\mu_{t-1}^{(i-1)} + r_t^{(i-1)}, \quad \Sigma^{(i-1)}_t = H_t^{(i-1)} + K_t^{(i-1)}\Sigma_{t-1}^{(i-1)}(K_t^{(i-1)})^\top, \quad t\in[1:T].
\end{equation}
In particular, this representation holds for $t=T$, which allows us to express $\psi_T^{(i)} = \mathrm{d}\pi_T/\mathrm{d}q_T^{\psi^{(i-1)}}$ on the form
$-\log\psi_T^{(i)}(x_T) = x_T^\top A_T^{(i)}x_T + x_T^\top b_T^{(i)} + c_T^{(i)}$, where
\begin{align*}
\begin{gathered}
A^{(i)}_T = \frac{1}{2}\left(\Sigma^{-1}-(\Sigma_T^{(i-1)})^{-1}\right), \qquad b^{(i)}_T = (\Sigma_T^{(i-1)})^{-1}\mu^{(i-1)}_T - \Sigma^{-1}\mu, \\
c^{(i)}_T = \frac{1}{2}\left[ \mu^\top \Sigma^{-1}\mu - \mu_T^{(i-1)\top} (\Sigma_T^{(i-1)})^{-1}\mu^{(i-1)}_T + \log \left(\det (\Sigma_T^{(i-1)})^{-1}\right) -\log\left( \det \Sigma^{-1}\right) \right].
\end{gathered}
\end{align*}

By induction on $t$ (while keeping $i$ fixed), one can show that $\psi^{(i)}_{t-1}(x_{t-1})$ can be written
$$-\log\psi^{(i)}_{t-1}(x_{t-1}) = x_{t-1}^{\top} A^{(i)}_{t-1}x_{t-1} + x_{t-1}^{\top}b^{(i)}_{t-1} + c^{(i)}_{t-1},$$
where $A^{(i)}_{t-1}, b^{(i)}_{t-1}$ and $c^{(i)}_{t-1}$  satisfy the backward recursions
\begin{align*}
A^{(i)}_{t-1} &= \frac{1}{2}(K_t^{(i-1)})^\top\left[(H_t^{(i-1)})^{-1} - \frac{1}{2}(H_t^{(i-1)})^{-1} \left(A^{(i)}_t + \frac{1}{2}(H_t^{(i-1)})^{-1} \right)^{-1}(H_t^{(i-1)})^{-1} \right]K^{(i-1)}_t,\\
b^{(i)}_{t-1} &= (K_t^{(i-1)})^\top (H_t^{(i-1)})^{-1} \left[ r^{(i-1)}_t - \frac{1}{2} \left(A^{(i)}_t + \frac{1}{2}(H_t^{(i-1)})^{-1}\right)^{-1}\left((H_t^{(i-1)})^{-1}r^{(i-1)}_t  - b^{(i)}_t\right)\right], \\
c^{(i)}_{t-1} &= c^{(i)}_t + \frac{1}{2}r_t^{(i-1)\top} (H_t^{(i-1)})^{-1} r^{(i-1)}_t \\
&- \frac{1}{4}\left((H_t^{(i-1)})^{-1}r^{(i-1)}_t  - b^{(i)}_t\right)^\top \left(A^{(i)}_t + \frac{1}{2}(H_t^{(i-1)})^{-1}\right)^{-1}\left((H_t^{(i-1)})^{-1}r^{(i-1)}_t  - b^{(i)}_t\right),
\end{align*}
for $t \in [1:T]$, initialized at the $A^{(i)}_T, b^{(i)}_T$ and $c^{(i)}_T$ given above.

The updated Markov kernels are given by $M_{t}^{\psi^{(i)}}(x_{t-1},\mathrm{d}x_{t})=\mathcal{N}\left(x_{t};K_{t}^{(i)}x_{t-1}+r_{t}^{(i)},H_{t}^{(i)}\right)\mathrm{d}x_{t}$, where the updated parameters $K_{t}^{(i)},H_{t}^{(i)}$ and $r_{t}^{(i+1)}$ satisfy
\begin{equation*}
\begin{gathered}
H_{t}^{(i)}=\left((H_{t}^{(i-1)})^{-1}+2A_{t}^{(i)}\right)^{-1},\qquad K_{t}^{(i)}=\left((H_{t}^{(i-1)})^{-1}+2A_{t}^{(i)}\right)^{-1}(H_{t}^{(i-1)})^{-1}K_{t}^{(i-1)},\\
r_{t}^{(i)}=\left((H_{t}^{(i-1)})^{-1}+2A_{t}^{(i)}\right)^{-1}\left((H_{t}^{(i-1)})^{-1}r_{t}^{(i-1)}-b_{t}^{(i)}\right).
\end{gathered}
\end{equation*}
The next IPF iterates can then be expressed
\begin{equation*}
\mathbb{P}^{\psi^{(i)}}(\mathrm{d}x_{0:T}) = p_{0}^{\psi^{(i)}}(\mathrm{d}x_{0})\prod_{t=1}^{T}M_{t}^{\psi^{(i)}}(x_{t-1},\mathrm{d}x_{t}),\quad \mathbb{Q}^{\psi^{(i)}}(\mathrm{d}x_{0:T}) =\pi_{0}(\mathrm{d}x_{0})\prod_{t=1}^{T}M_{t}^{\psi^{(i)}}(x_{t-1},\mathrm{d}x_{t}),
\end{equation*}
where $p_{0}^{\psi^{(i)}}(\mathrm{d}x_{0})=\mathcal{N}\left(x_{0};\gamma^{(i)},\Gamma^{(i)}\right)\mathrm{d}x_{0}$, with
$$\Gamma^{(i)}=\left(\Sigma_{0}^{-1}+2A_{0}^{(i)}\right)^{-1},\qquad\gamma^{(i)}=\left(\Sigma_{0}^{-1}+2A_{0}^{(i)}\right)^{-1}\left(\Sigma_{0}^{-1}\mu_{0}-b_{0}^{(i)}\right).$$

The potential function $\alpha^{(i)}$ defined by \eqref{eq:IPF_potential1} can be written
$-\log\alpha^{(i)}(x_{0}) =-\log\alpha^{(i-1)}(x_{0})+x_{0}^{\top}A_{\alpha}^{(i)}x_{0}+x_{0}^{\top}b_{\alpha}^{(i)}+c_{\alpha}^{(i)}$, where
\begin{equation*}
 \begin{gathered}
 A_{\alpha}^{(i)}=\frac{1}{2}\left(\Sigma_{0}^{-1}-(\Gamma^{(i)})^{-1}\right)=-A_{0}^{(i)},\qquad b_{\alpha}^{(i)}=(\Gamma^{(i)})^{-1}\gamma^{(i)}-\Sigma_{0}^{-1}\mu_{0}=-b_{0}^{(i)},\\
c_{\alpha}^{(i)}=\frac{1}{2}\left[\mu_{0}^{\top}\Sigma_{0}^{-1}\mu_{0}-\gamma^{(i)\top}(\Gamma^{(i)})^{-1}\gamma^{(i)}+\log\left(\det(\Gamma^{(i)})^{-1}\right)-\log\left(\det\Sigma_{0}^{-1}\right)\right]
\end{gathered}
\end{equation*}
In other words, $-\log\alpha^{(i)}(x_{0})=x_{0}^{\top}F_{\alpha}^{(i)}x_{0}+x_{0}^{\top}v_{\alpha}^{(i)}+d_{\alpha}^{(i)}$, where
$$F_{\alpha}^{(i)}=\sum_{j=1}^{i}A_{\alpha}^{(j)},\qquad v_{\alpha}^{(i)}=\sum_{j=1}^{i}b_{\alpha}^{(j)},\qquad d_{\alpha}^{(i)}=\sum_{j=1}^{i}c_{\alpha}^{(j)}.$$

Similarly, the potential $\beta^{(i)}$ defined by  \eqref{eq:IPF_potential2} can be written
\begin{align*}
-\log\beta^{(i)}(x_{T}) &=-\log\beta^{(i-1)}(x_{T})+x_{T}^{\top}A_{T}^{(i)}x_{T}+x_{T}^{\top}b_{T}^{(i)}+c_{T}^{(i)}\\
	&=x_{T}^{\top}F_{\beta}^{(i)}x_{T}+x_{T}^{\top}v_{\beta}^{(i)}+d_{\beta}^{(i)},
\end{align*}
where
$$F_{\beta}^{(i)}=\sum_{j=1}^{i}A_{T}^{(j)},\qquad v_{\beta}^{(i)}=\sum_{j=1}^{i}b_{T}^{(j)},\qquad d_{\beta}^{(i)}=\sum_{j=1}^{i}c_{T}^{(j)}.$$

\section{Additional algorithms}\label{appendix:algos}
 \begin{algorithm}{\small
\caption{\label{algorithm:aIPF_langevin} Approximate IPF for discretized Langevin dynamics}
\textbf{Input:} Number of time-steps $T\in\mathbb{N}$, terminal time $\tau>0$, function classes $\{\mathsf{F}_t\}_{t\in[0:T]}$, number of particles $N\in\mathbb{N}$, number of iterations $I\in\mathbb{N}$.
\begin{enumerate}
\item Initialize: Set $h = \tau/T$,  $\hat{\psi}^{(0)}_t = 1$ for $t \in [0:T]$, and
$$M_t(x_{t-1},x_t) = \mathcal{N}\left(x_t; x_{t-1} + \frac{h}{2}\nabla \log \pi_t(x_{t-1}), h\mathcal{I}_d \right), \quad t\in[1:T].$$
\item For $1\leq i \leq I$:
\begin{enumerate}
\item For each $t\in[0:T]$, define $$\log \bar{\psi}_t^{(i-1)}(x_{t-1},x_t) = \log \hat{\psi}^{(i-1)} _t (x_{t-1}) + \nabla \log \hat{\psi}_t^{(i-1)} (x_{t-1})^\top (x_t - x_{t-1}).$$
\item Sample trajectories $\{X_{0:T}^n\}_{n\in[1:N]}$ from $\mathbb{Q}^{\bar{\psi}^{(i-1)}}$: for each $n\in [1:N]$, sample $X^n_0 \sim \pi_0(\mathrm{d}x_0)$ and $$X^n_t \sim \mathcal{N}\left(X_{t-1}^n + \frac{h}{2}\nabla \log \pi_t(X_{t-1}^n) + h\nabla \log \psi_t^{(i-1)}(X_{t-1}^n), h\mathcal{I}_d \right), \quad t\in[1:T].$$
\item For each $n\in[1:N]$, compute an estimator $\tilde{\phi}_T^{(i-1)}(X_T^n)$ of $\phi_T^{(i-1)}(X_T^n) = \mathrm{d}\pi_T/\mathrm{d}q_T^{\bar{\psi}^{(i-1)}}(X^n_T)$ using Algorithm \ref{algorithm:csmc},
e.g. using the backward kernels for policy $\psi$
$$L_{t-1}^{\bar{\psi}}(x_t,x_{t-1}) =  \mathcal{N}\left(x_{t-1}; x_{t} + \frac{h}{2}\nabla \log \pi_{t-1}(x_{t}) - h\nabla \log \psi_t(x_{t}), h\mathcal{I}_d \right).$$
\item Approximate dynamic programming: perform the recursion
\begin{align*}
\hat{\phi}_T^{(i)} &= \argmin_{f\in\mathsf{F}_T} \sum_{n = 1}^N\left| \log f(X_T^n) - \log \tilde{\phi}_T^{(i-1)}(X_T^n)\right|^2, \\
\hat{\phi}_t^{(i)} &= \argmin_{f \in\mathsf{F}_t} \sum_{n = 1}^N\left| \log f(X_t^n) - \log M_{t+1}^{\bar{\psi}^{(i-1)}}(\bar{\phi}^{(i)}_{t+1})(X_{t}^n)\right|^2,\quad t\in[1:T-1],
\end{align*}
where $\bar{\phi}^{(i)}$ is defined analogously to $\bar{\psi}^{(i-1)}$.
\item Set $\hat{\psi}^{(i)}_0 = 1$ and $\hat{\psi}^{(i)}_t = \hat{\psi}^{(i-1)}_t \cdot \hat{\phi}^{(i)}_t$ for $t\in[1:T]$.
\end{enumerate}
\end{enumerate}
\textbf{Output:} Policy $\hat{\psi}^{(I)}$.}
\end{algorithm}

\end{appendix}

\end{document}